\documentclass{article}
\usepackage{kr}

\usepackage{times}
\usepackage{soul}
\usepackage{url}
\usepackage[hidelinks]{hyperref}
\usepackage[utf8]{inputenc}
\usepackage[small]{caption}
\usepackage{graphicx}
\usepackage{amsmath}
\usepackage{amsthm}
\usepackage{booktabs}
\usepackage{algorithm}
\usepackage{algorithmic}
\usepackage{latexsym}
\usepackage{amsfonts, amssymb, stmaryrd}
\usepackage{xspace}
\usepackage{multirow}
\usepackage{xcolor}
\usepackage{thm-restate} 
\usepackage{tikz}
\tikzset{>=latex}

\usepackage{graphicx}
\usepackage{caption}
\usepackage{subcaption}
\usepackage{pgfplotstable}
\usepackage{color}
\usepackage{colortbl}
\usepackage{hyperref}

\usepackage{fancyhdr}
\fancypagestyle{firstpage}
{
    \fancyhead[L]{Long version of the KR'25 paper with the same title. It contains an appendix with full proofs and corrects Definition 4 and some minor glitches (in particular a few misreported numbers in the experimental section).}    
    \fancyhead[R]{}
}

\newtheorem{example}{Example}
\newtheorem{theorem}{Theorem}
\newtheorem{lemma}{Lemma}

\newtheorem{corollary}{Corollary}
\newtheorem{definition}{Definition}
\newtheorem{remark}{Remark}

\newcommand{\mypar}[1]{\smallskip\noindent\textbf{#1}}

\newcommand{\inds}{{\bf{C}}}
\newcommand{\vars}{{\bf{V}}}
\newcommand{\preds}{{\bf{P}}}

\def\dllitecore{DL-Lite$_{\mn{core}}$\xspace}

\newcommand{\ans}{\vec{a}}
\newcommand{\signature}[1]{\mn{sig}(#1)}
\newcommand{\individuals}[1]{\mn{const}(#1)}

\newcommand{\sem}{\textup{Sem}\xspace}
\newcommand{\brave}{\textup{brave}\xspace}
\newcommand{\AR}{\textup{AR}\xspace}
\newcommand{\IAR}{\textup{IAR}\xspace}
\newcommand{\semmodels}[1]{\models_{\sem}^{#1}}
\newcommand{\bravemodels}[1]{\models_{\brave}^{#1}}
\newcommand{\armodels}[1]{\models_{\AR}^{#1}}
\newcommand{\iarmodels}[1]{\models_{\IAR}^{#1}}

\newcommand{\conflicts}[1]{\mi{Conf}(#1)}
\newcommand{\causes}[1]{\mi{Causes}(#1)}

\newcommand{\reps}[1]{\mi{SRep}(#1)}

\newcommand{\preps}[1]{\mi{PRep}(#1)}
\newcommand{\creps}[1]{\mi{CRep}(#1)}
\newcommand{\xreps}[1]{\mi{XRep}(#1)}

\newcommand{\id}{\ensuremath{\mathbf{id}}}
\newcommand{\data}{\ensuremath{\Fmc}}
\newcommand{\meta}{\ensuremath{\Mmc}}
\newcommand{\cond}{\ensuremath{\mn{Cond}}}

\newcommand{\pref}[2]{\ensuremath{\mn{pref}(#1,#2)}}
\newcommand{\eval}{\ensuremath{\mi{eval}}}
\newcommand{\level}{\ensuremath{\mi{level}}}

\newcommand{\idargs}{\mathsf{pos}_{\mathbf{ID}}}

\newcommand{\fullsucc}{\succ_{\Sigma,\Kmc,\meta}}

\newcommand{\metapreds}{\preds_{\bf{M}}}
\def\idinds{\inds_{\mathbf{ID}}}
\def\iso{\mu}

\newcommand{\rvars}{\ensuremath{\mi{var}}}

\newcommand{\succup}{\succ^u}
\newcommand{\succdown}{\succ^d}
\newcommand{\succrefup}{\succ^{ru}}
\newcommand{\succground}{\succ^g}

\def\kbex{\Kmc_\mathsf{ex}}
\def\dex{\Dmc_\mathsf{ex}}
\def\mex{\Mmc_\mathsf{ex}}
\def\tboxex{\Tmc_\mathsf{ex}}
\def\mdex{\data_\mathsf{ex}}
\def\idex{\ensuremath{\id_\mathsf{ex}}}
\def\rulesetex{\Sigma_\mathsf{ex}}

\def\ptime{\textsc{PTime}\xspace}
\def\np{\textsc{NP}\xspace}
\def\conp{co\textsc{NP}\xspace}
\def\piptwo{\ensuremath{\Pi^{p}_{2}}\xspace}
\def\sigmaptwo{\ensuremath{\Sigma^{p}_{2}}\xspace}

\def\true{\ensuremath{\mathsf{true}}}
\def\false{\ensuremath{\mathsf{false}}}

\newcommand{\mn}[1]{\ensuremath{\mathsf{#1}}}
\newcommand{\mi}[1]{\ensuremath{\mathit{#1}}}
\newcommand{\mt}[1]{\ensuremath{\mathtt{#1}}}

\newcommand{\Bmc}{\ensuremath{\mathcal{B}}}
\newcommand{\Cmc}{\ensuremath{\mathcal{C}}}
\newcommand{\Dmc}{\ensuremath{\mathcal{D}}}

\newcommand{\Fmc}{\ensuremath{\mathcal{F}}}

\newcommand{\Lmc}{\ensuremath{\mathcal{L}}}
\newcommand{\Kmc}{\ensuremath{\mathcal{K}}}
\newcommand{\Mmc}{\ensuremath{\mathcal{M}}}

\newcommand{\Rmc}{\ensuremath{\mathcal{R}}}
\newcommand{\Smc}{\ensuremath{\mathcal{S}}}
\newcommand{\Tmc}{\ensuremath{\mathcal{T}}}

\newcommand{\eg}{e.g.,~}

\newcommand{\ie}{i.e.,~}
\newcommand{\wrt}{w.r.t.~}
\newcommand{\cf}{cf.~}

\newcommand{\resp}{resp.~}

\title{A Rule-Based Approach to Specifying Preferences over Conflicting Facts \\ and Querying Inconsistent Knowledge Bases}

\author{%
Meghyn Bienvenu$^1$\and
Camille Bourgaux$^2$\and
Katsumi Inoue$^3$\and
Robin Jean$^1$ \\
\affiliations
$^1$Univ. Bordeaux, CNRS, Bordeaux INP, LaBRI, UMR 5800, Talence, France\\
$^2$DI ENS, ENS, CNRS, PSL University \& Inria, Paris, France\\
$^3$National Institute of Informatics, Tokyo, Japan\\
\emails
\{meghyn.bienvenu, robin.jean\}@u-bordeaux.fr,
camille.bourgaux@ens.fr,
inoue@nii.ac.jp
}

\begin{document}

\maketitle
\thispagestyle{firstpage}
\begin{abstract}
Repair-based semantics have been extensively studied 
as a means of obtaining meaningful answers to queries 
posed over inconsistent knowledge bases (KBs). 
While several works have considered how to exploit a priority 
relation between facts to select optimal repairs, the question of 
how to specify such preferences remains largely unaddressed. 
This motivates us to introduce 
a declarative rule-based framework for specifying and computing
a priority relation between conflicting facts. 
As the expressed preferences may contain undesirable cycles,
we consider the problem of determining when 
a set of preference rules always yields an acyclic relation, 
and we also explore a pragmatic approach that extracts an
acyclic relation by applying various cycle removal techniques. 
Towards an end-to-end system for querying inconsistent KBs, 
we present a preliminary implementation and experimental evaluation of the framework, 
which employs answer set programming
to 
evaluate the preference rules, apply the desired cycle resolution techniques to obtain a priority relation,
and answer queries under prioritized-repair semantics. 
\end{abstract}


\section{Introduction}

Inconsistency-tolerant semantics are a well-established approach to querying data inconsistent \wrt some constraints, both in the relational database and ontology-mediated query answering settings (\cf recent surveys \cite{DBLP:conf/pods/Bertossi19,DBLP:journals/ki/Bienvenu20}). Such semantics typically rely on \emph{(subset) repairs}, defined as maximal subsets of the data consistent \wrt the constraints. The most well-known, called the \emph{AR} semantics in the KR community and corresponding to consistent query answering in the database community, considers that a Boolean query holds true if it holds in every repair. The more cautious \emph{IAR} semantics amounts to querying the repairs intersection, and the less cautious \emph{brave} semantics only requires that the query holds in some repair.

Since an inconsistent dataset may have a lot of repairs, several notions of preferred repairs have been proposed in the literature, to restrict the possible worlds considered to answer queries, for example by taking into account some information about the reliability of the data \cite{DBLP:conf/icdt/LopatenkoB07,DBLP:journals/kais/DuQS13,DBLP:conf/aaai/BienvenuBG14,DBLP:journals/ai/CalauttiGMT22,DBLP:conf/kr/LukasiewiczMM23}. 
In particular, since its introduction by \citeauthor{DBLP:journals/amai/StaworkoCM12}~\shortcite{DBLP:journals/amai/StaworkoCM12}, the framework of \emph{prioritized databases}, in which a \emph{priority relation} between conflicting facts is used to define \emph{optimal repairs}, has attracted attention, with numerous theoretical results \cite{DBLP:conf/icdt/KimelfeldLP17,DBLP:journals/tcs/KimelfeldLP20,DBLP:conf/kr/BienvenuB20,DBLP:conf/kr/BienvenuB23}, and an implementation \cite{DBLP:conf/kr/BienvenuB22}. 
However, the crucial question of obtaining the priority relation was left unaddressed, preventing the adoption of this framework in practice. Indeed, it is not realistic 
to expect users to manually input 
 a binary relation between the facts. 
 
 In our work, we tackle this challenge by developing a general rule-based approach to specifying preferences over conflicting facts. 
After introducing the 
background on optimal repair-based inconsistency-tolerant semantics in Section~\ref{sec:prelims}, 
we present in Section~\ref{sec:framework} our framework for specifying a priority relation between conflicting facts via preference rules. 
A distinguishing feature of our work is that we address the issue of cycles in the expressed preferences, first by investigating the problem of deciding whether a given set of preference rules is guaranteed to produce an acyclic relation, and second by providing a pragmatic approach that uses
cycle removal strategies to extract an acyclic relation. 
Section~\ref{sec:implem} describes our implementation which employs answer set programming to evaluate the preference rules, apply the desired cycle resolution techniques to obtain a priority relation, and answer queries under optimal repair-based semantics. 
Finally, we present in Section~\ref{sec:expe} a preliminary experimental evaluation of the overall framework. Sections \ref{sec:related-work} 
and \ref{conclusion} discuss related and future work. 
Proofs and additional details on the implementation and experiments 
are provided in the appendix. All materials to reproduce the experiments (e.g.\ datasets, logic programs) are available at \url{https://github.com/rjean007/PreferenceRules-ASP}


\section{Preliminaries}\label{sec:prelims}

We recall the framework of inconsistency-tolerant querying of prioritized knowledge bases, as considered in \cite{DBLP:conf/kr/BienvenuB22,DBLP:conf/dlog/Bourgaux24}. 
We assume that readers are familiar with 
first-order logic and consider three disjoint sets of predicates $\preds$, constants $\inds$ and variables $\vars$. 
As usual, each predicate has an associated arity $n\geq 1$, and we shall use $\preds_n$ 
for the set of the $n$-ary predicates in $\preds$. 

\mypar{Knowledge bases}
A \emph{knowledge base (KB)} $\Kmc=(\Dmc,\Tmc)$ consists of a \emph{dataset} $\Dmc$ and a \emph{logical theory}~$\Tmc$: $\mathcal{D}$ is a finite set of \emph{facts} of the form $P(c_1, \ldots, c_n)$ with $P\in\preds_n$, $c_i\in\inds$ for $1 \leq i \leq n$, and $\Tmc$ is a finite set of first-order logic (FOL) sentences built from $\preds$, $\inds$ and $\vars$. 
We denote by $\signature{\Kmc}$ and $\individuals{\Kmc}$ (\resp $\signature{\Dmc}$, $\individuals{\Dmc}$) the set of predicates and constants that occur in $\Kmc$ (\resp in $\Dmc$).  
 A KB $\Kmc=(\Dmc,\Tmc)$ is \emph{consistent}, and $\Dmc$ is called \emph{$\Tmc$-consistent}, if $\Dmc \cup \Tmc$ has some model. Otherwise, $\Kmc$ is \emph{inconsistent}, denoted $\Kmc \models \bot$. 

Typically, $\Tmc$ will be either an \emph{ontology} (formulated in some description logic or decidable class of existential rules) or a set of \emph{database constraints}. 
In particular, we consider: 
\begin{itemize}
\item \emph{Description logics of the DL-Lite family} \cite{calvaneseetal:dllite}, such as \dllitecore, 
whose 
axioms take the form $\mathsf{B}_1 \sqsubseteq (\neg) \mathsf{B}_2$, with $\mathsf{B}_i \in \preds_1 \cup \{\exists \mathsf{R}, \exists \mathsf{R}^- \mid \mathsf{R} \in \preds_2\}$. 
\item \emph{Denial constraints} of the form $
\alpha_1 \wedge \ldots \wedge \alpha_n\rightarrow\bot$, where each $\alpha_i$ is a relational or inequality atom, 
which includes in particular \emph{functional dependencies} (FDs) such as $P(x,y,z) \wedge P(x,y,z') \wedge z\neq z' \rightarrow \bot$. 
\end{itemize}
 
\mypar{Queries} 
A \emph{conjunctive query} (CQ) is a conjunction of relational atoms $P(t_1, \ldots, t_n)$ ($P\in\preds_n$, $t_i\in\inds\cup\vars$),
where some variables may be existentially quantified. 
Given a query $q(\vec{x})$, with free variables~$\vec{x}$, and a tuple of constants $\vec{a}$ such that $|\vec{a}|=|\vec{x}|$, $q(\vec{a})$ denotes the first-order sentence obtained by replacing each variable in~$\vec{x}$ by the corresponding constant in~$\vec{a}$. 
A \emph{(certain) answer} to $q(\vec{x})$ over $\Kmc$ is a tuple~$\vec{a}\in\individuals{\Kmc}^{|\vec{x}|}$ such that $q(\vec{a})$ holds in every model of $\Kmc$, denoted $\Kmc \models q(\vec{a})$. 
A \emph{cause} for $q(\vec{a})$ w.r.t.\ $\Kmc=(\Dmc,\Tmc)$ is an inclusion-minimal $\Tmc$-consistent subset $\Cmc \subseteq \Dmc$ such that $(\Cmc, \Tmc) \models q(\vec{a})$. 
The set of causes for $q(\vec{a})$ w.r.t.\ $\Kmc$ is denoted by $\causes{q(\ans),\Kmc} $.

\mypar{Conflicts and repairs} 
A \emph{conflict} of $\Kmc=(\Dmc,\Tmc)$ is an inclusion-minimal subset $\Dmc' \subseteq \Dmc$ such that $(\Dmc',\Tmc) \models \bot$. 
The set of conflicts of $\Kmc$ is denoted $\conflicts{\Kmc}$. 
A \emph{(subset) repair} of $\Kmc$ is an inclusion-maximal subset $\Rmc \subseteq \Dmc$ such that $(\Rmc,\Tmc) \not \models \bot$.  The set of repairs of $\Kmc$ is denoted $\reps{\Kmc}$.

\mypar{Prioritized KBs} 
A \emph{priority relation} $\succ$ for a KB $\Kmc=(\Dmc,\Tmc)$ is an acyclic\footnote{In line with existing work on prioritized repairs, we do not require priority relations to be transitive. } 
binary relation over the facts of $\Dmc$ such that if $\alpha\succ\beta$, then there exists $\Cmc\in\conflicts{\Kmc}$ such that $\{\alpha,\beta\}\subseteq \Cmc$. 
It is \emph{total} if for every pair $\alpha\neq\beta$ such that $\{\alpha,\beta\}\subseteq \Cmc$ for some $\Cmc\in\conflicts{\Kmc}$, either $\alpha\succ\beta$ or $\beta\succ\alpha$. 
A \emph{completion} of $\succ$ is a total priority relation $\succ'\ \supseteq \ \,\succ$. 
A \emph{prioritized KB} $\Kmc_\succ$ is a KB $\Kmc$ with a priority relation $\succ$ for~$\Kmc$.  

Priority relations are used to select optimal repairs. 
We recall the notions of Pareto- and completion-optimal repairs\footnote{Staworko et al.\ (\citeyear{DBLP:journals/amai/StaworkoCM12}) further introduces a third notion of globally-optimal repair, see Section \ref{conclusion} for discussion. } from \cite{DBLP:journals/amai/StaworkoCM12}: 
\begin{definition} 
Consider a prioritized KB $\Kmc_\succ$ with $\Kmc=(\Dmc,\Tmc)$, and let $\Rmc \in \reps{\Kmc}$. 
A \emph{Pareto improvement} of $\Rmc$ 
 is a $\Tmc$-consistent $\Bmc\subseteq\Dmc$ such that there is 
 $ \beta\in\Bmc\setminus\Rmc$ with $\beta\succ\alpha$ for every $\alpha\in\Rmc\setminus\Bmc$. 
The repair $\Rmc$ is a  
\emph{Pareto-optimal repair} of $\Kmc_\succ$ if there is no Pareto improvement of $\Rmc$, and a 
\emph{completion-optimal repair} of $\Kmc_\succ$
if $\Rmc$ is a Pareto-optimal repair of $\Kmc_{\succ'}$, for some completion $\succ'$ of $\succ$. 
We denote by 
$\preps{\Kmc_\succ}$ and $\creps{\Kmc_\succ}$ the sets of  
Pareto- and completion-optimal repairs.
\end{definition}
The following relation between optimal repairs is known:$$\creps{\Kmc_\succ} 
\subseteq\preps{\Kmc_\succ}\subseteq\reps{\Kmc}.$$
 
\mypar{Repair-based semantics} 
We next recall three prominent inconsistency-tolerant semantics (brave, AR, and IAR), parameterized by the considered type of repair: 
\begin{definition}
Fix X $\in \{S,P,C\}$ and 
consider a prioritized KB $\Kmc_\succ$ with $\Kmc=(\Dmc,\Tmc)$, query $q(\vec{x})$, and tuple $\ans\in\individuals{\Kmc}^{|\vec{x}|}$. 
Then $\ans$ is an answer to $q$ over $\Kmc_\succ$ 
\begin{itemize}
\item under \emph{X-brave semantics}, denoted $\Kmc_\succ \bravemodels{X} q(\ans)$, if $(\Rmc,\Tmc) \models q(\ans)$ for some $\Rmc \in \xreps{\Kmc_\succ}$
\item under \emph{X-AR semantics}, denoted $\Kmc_\succ \armodels{X} q(\ans)$, if $(\Rmc,\Tmc) \models q(\ans)$ for every $\Rmc \in \xreps{\Kmc_\succ}$
\item under \emph{X-IAR semantics}, denoted $\Kmc_\succ \iarmodels{X} q(\ans)$, if $(\Bmc,\Tmc) \models q(\ans)$ where $\Bmc=\bigcap_{\Rmc \in \xreps{\Kmc_\succ}} \Rmc$
\end{itemize}
\end{definition}
\noindent It is known that $\Kmc_\succ \iarmodels{X} q \Rightarrow \Kmc_\succ \armodels{X} q  \Rightarrow \Kmc_\succ \bravemodels{X} q$.

\begin{example}\label{running-example-1}
Our running example considers a DL 
knowledge base $\kbex=(\dex,\tboxex)$ about a university. The ontology expresses that associate and full professors (\mn{APr} and \mn{FPr}) are faculty members (\mn{Fac}) and clerical staff workers (\mn{Cleric}) are administrative staff workers (\mn{Adm}). 
Moreover, one cannot be both an associate and a full professor, or an administrative staff worker and a faculty member.
\begin{align*}
\tboxex=\{&\mn{APr}\sqsubseteq\mn{Fac}, \mn{FPr}\sqsubseteq\mn{Fac}, \mn{APr} \sqsubseteq\neg \mn{FPr},  
\\
&\exists\mn{Teach}\sqsubseteq \mn{Fac}, \mn{Cleric}\sqsubseteq\mn{Adm}, \mn{Adm}\sqsubseteq \neg \mn{Fac}
\}\\
\dex=\{&\mn{APr}(a),  \mn{FPr}(a), \mn{Cleric}(a), \mn{Adm}(a), \mn{Teach}(a,c),\\
&\mn{Adm}(b), \mn{APr}(b)
\}
\end{align*}
The picture below represents the conflicts of $\kbex$ and a priority relation $\succ$: an arrow from $\alpha$ to $\beta$ indicates that $\alpha\succ\beta$ and a dotted line indicates that $\{\alpha,\beta\}\in\conflicts{\kbex}$ without priority between $\alpha$ and $\beta$.\newline
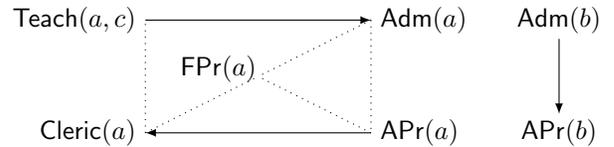
\begin{tikzpicture} 
\node [left] at (0,1.5) {$\mn{Teach}(a,c)$};
\node [right] at (3,1.5) {$\mn{Adm}(a)$};
\node [left] at (0,0) {$\mn{Cleric}(a)$};
\node [right] at (3,0) {$\mn{APr}(a)$};
\node [left] at (1.6,0.85) {$\mn{FPr}(a)$};

\node [above] at (5.5,1.2) {$\mn{Adm}(b)$};
\node [below] at (5.5,0.3) {$\mn{APr}(b)$};

\draw[dotted] (0,0) -- (0,1.5);
\draw[dotted] (3,1.5) -- (3,0);
\draw[dotted] (3,1.5) -- (1.5,0.75);
\draw[dotted] (0,0) -- (1.5,0.75);
\draw[dotted] (3,0) -- (1.5,0.75);
\draw[<-] (3,1.5) -- (0,1.5);
\draw[->] (3,0) -- (0,0);

\draw[->] (5.5,1.25) -- (5.5,0.25);
\end{tikzpicture}
There are six repairs: 
\begin{align*}
\Rmc_1=\{&\mn{APr}(a),\mn{Teach}(a,c),\mn{Adm}(b)\}, \\
\Rmc_2=\{&\mn{FPr}(a),\mn{Teach}(a,c),\mn{Adm}(b)\}, \\
\Rmc_3=\{&\mn{Cleric}(a),\mn{Adm}(a),\mn{Adm}(b)\}, \\
\Rmc_4=\{&\mn{APr}(a),\mn{Teach}(a,c),\mn{APr}(b)\}, \\
\Rmc_5=\{&\mn{FPr}(a),\mn{Teach}(a,c),\mn{APr}(b)\}, \\
\Rmc_6=\{&\mn{Cleric}(a),\mn{Adm}(a),\mn{APr}(b)\}, 
\end{align*}
and one can check that  
$\preps{{\kbex}_\succ}=\{\Rmc_1,\Rmc_2,\Rmc_3\}$ while $\creps{{\kbex}_\succ}
=\{\Rmc_1,\Rmc_2\}$, so, \eg ${\kbex}_\succ\iarmodels{P}\mn{Adm}(b)$, ${\kbex}_\succ\not\bravemodels{P}\mn{APr}(b)$, ${\kbex}_\succ\not\armodels{P}\mn{FPr}(a)$, ${\kbex}_\succ\iarmodels{C}\mn{Fac}(a)$, ...
\end{example}

\mypar{Data complexity} When considering the complexity of tasks 
involving an input dataset $\Dmc$, we always 
use \emph{data complexity}, where the sizes of the logical theory $\Tmc$ and query $q(\vec{x})$ are assumed to be fixed. 
Theorems \ref{ub-thm} 
and \ref{lb-thm} summarize known upper and lower bounds in the database and ontology settings \cite{DBLP:journals/amai/StaworkoCM12,DBLP:conf/kr/BienvenuB20,DBLP:conf/kr/BienvenuB22,DBLP:conf/ijcai/Rosati11,Bienvenu_TractableApproximation_long}.

\begin{theorem}\label{ub-thm} 
Let $\Lmc$ be an FOL fragment for which KB consistency and query entailment are in 
\ptime. 
Query entailment for $\Lmc$ KBs and X $\in \{S,P,C\}$  is 
in $\np$ under X-brave semantics, and 
in $\conp$ under X-AR and X-IAR semantics.  
\end{theorem}

\begin{theorem}\label{lb-thm}
Let $\Lmc$ be any FOL fragment that extends \dllitecore, $\mathcal{EL}_\bot$, or FDs. 
Query entailment for $\Lmc$ KBs is 
$\np$-hard under X-brave semantics (X $\in \{P,C\}$), 
$\conp$-hard under X-AR semantics  (X  $\in \{S,P,C\}$), 
$\conp$-hard under X-IAR semantics (X $\in \{P,C\}$).
\end{theorem}


\def\preflang{\mathcal{PL}}

\section{Specifying Priority Relations via Rules}\label{sec:framework} 
The optimal repair semantics recalled in Section \ref{sec:prelims}
suppose that we have a priority relation between the facts. However, 
the question of how to conveniently specify the priority relation has 
not yet been addressed in the literature. This will be the topic of the present section,
which 
presents a declarative rule-based framework for specifying priority relations.

\subsection{Preference Rules}
We propose to use \emph{preference rules} to state that, when some conditions are satisfied, a fact should generally be preferred to another fact. 
The rule conditions may naturally refer to the presence (or absence) of facts in the dataset. However, typically we may also want to exploit information about the facts themselves, provided in metadata, e.g.\ 
 to compare facts based upon the date they were added.

We now introduce some terminology and notation in order to be able to refer to metadata in rule conditions. 
First, we fix a subset $\metapreds \subsetneq \preds$ of metadata predicates and $\idinds \subsetneq \inds$ of fact identifiers, 
assumed distinct from the predicates and constants used in the considered dataset. We assume that each $n$-ary predicate
$P \in \metapreds$ has an associated set of $\idinds$-positions $\idargs(P) \subseteq \{1, \ldots, n\}$, 
indicating which positions of $P$ contain constants from $\idinds$. 
Given a KB $\Kmc=(\Dmc,\Tmc)$, a \emph{meta-database} for $\Kmc$ is a pair $\meta=(\id,\data)$, 
where $\id$ is an \emph{injective function} that associates to each fact $\alpha \in \Dmc$ an identifier $\id(\alpha)$ from 
$\idinds$, and $\data$ is a dataset with $\signature{\data} \subseteq \metapreds$ satisfying the following conditions: 
\begin{itemize}
\item if $P(c_1, \ldots, c_n) \in \data$, then $c_j \in \idinds$ iff $j \in \idargs(P)$
\item if $c \in \individuals{\data}\cap \idinds$, then  $c=\id(\alpha)$ for some $\alpha \in \Dmc$ 
\end{itemize} 
Intuitively, $\data$ provides information about the facts in $\Dmc$ by referring to their identifiers defined by $\id$. Every identifier in $\data$ must designate a unique fact in $\Dmc$, but it is not required that $\data$ contains information about all facts in $\Dmc$. 

\begin{example}\label{running-example-2}
We associate to the KB $\kbex$ from Example~\ref{running-example-1} the meta-database $\mex=(\idex,\mdex)$, 
which records the year that facts have been added to the university database: \smallskip\\
$\idex(\mn{APr}(a))=1$,  $\idex(\mn{FPr}(a))=2$, $\idex(\mn{Cleric}(a))=3$, $\idex(\mn{Adm}(a))=4$, $\idex(\mn{Teach}(a,c))=5$, \\ $\idex(\mn{Adm}(b))=6$, $\idex(\mn{APr}(b))=7$ 
\begin{align*}
\data=\{&\mn{Date}(1,2014), \mn{Date}(2,2025), \mn{Date}(3,2013), \\
&<\!(2013,2014), <\!(2013,2025), <\!(2014,2025)
\}.
\end{align*}
\end{example}
\begin{remark}
To simplify the presentation, we employ a meta-database predicate $<$ to compare years. 
However, such comparison facts could be avoided by extending the definition of meta-databases to allow for 
built-in comparison predicates for different datatypes and 
adding further typing constraints on predicate positions. 
\end{remark}

We now formulate a general definition of preference rules, 
which are evaluated over a KB and meta-database:

\begin{definition}\label{prefruledef}
A \emph{preference rule} $\sigma$ over  
$S \subseteq \preds$ 
takes the form $$\cond(x_1,x_2)\rightarrow \pref{x_1}{x_2}$$ where $\mathsf{pref} \not \in S$ 
is a special predicate 
(assumed not to occur in KBs nor meta-databases)
and $\cond(x_1,x_2)$ is an expression whose predicate symbols are drawn from $S$
and whose two distinguished free variables $x_1,x_2$
occur only in $\idinds$-positions of relational atoms over $S \cap \metapreds$ or in equality atoms of the form $x_i=\id(P(\vec{t}))$. 
We call $\cond(x_1,x_2)$ the \emph{body} of $\sigma$ and $\pref{x_1}{x_2}$ its \emph{head}. 
A preference rule language is a 
set of preferences rules (intuitively, stipulating the allowed syntax of rule bodies). 

The semantics of preference rule languages is defined using evaluation functions. 
An \emph{evaluation function for a preference rule language~$\preflang$} 
is a function $\eval$ that associates $\true$ or $\false$ to every KB $\Kmc=(\Dmc,\Tmc)$ and associated 
meta-database $\meta=(\id,\data)$, 
preference rule $\cond(x_1,x_2)\rightarrow \pref{x_1}{x_2} \in \preflang$, 
and pair of constants $(id_1, id_2) \in  \{(\id(\alpha),\id(\beta))\mid \alpha, \beta \in \Dmc\} $. 
We denote by $(\Kmc,\meta)\models \cond(id_1,id_2)$ the fact that 
$\eval(\Kmc,\meta, \cond(x_1,x_2) ,id_1,id_2)=\true$ and 
say that $\pref{id_1}{id_2}$ is \emph{induced by $\sigma$ over $(\Kmc,\meta)$ (\wrt $\eval$)} 
if $\sigma$ has body $\cond(x_1,x_2)$ and $(\Kmc,\meta)\models \cond(id_1,id_2)$. 
Given a set $\Sigma$ of preference rules with $\Sigma \subseteq \preflang$ 
and an evaluation function $\eval$ for $\preflang$, 
we denote by $\Sigma(\Kmc,\meta)$ the set of all $\pref{id_1}{id_2}$ induced by some $\sigma\in\Sigma$ over $(\Kmc,\meta)$. 
\end{definition}

Observe that the restrictions on the variables $x_1, x_2$ serve to ensure 
head variables are mapped to fact identifiers. Aside from this restriction,
we have the preceding definition very generic to encompass many different settings. 
The following example and definition illustrate how our framework can be instantiated,
by giving a concrete preference rule language.

\begin{example}
Let $\rulesetex$ contain three preference rules for the KB $\kbex$ and meta-database $\mex$ of our running example:
\begin{align*}
\sigma_1:\!\!&\quad\!\mn{Date}(x_1,y_1)\!\wedge\!\mn{Date}(x_2,y_2)\wedge\! <\!(y_2,y_1)\!\rightarrow\! \pref{x_1}{x_2}\\
\sigma_2:\!\!&\quad\! x_1=\id(\mn{FPr}(y)) \wedge x_2=\id(\mn{APr}(y)) \rightarrow \pref{x_1}{x_2}\\
\sigma_3:\!\!&\quad\! Y\sqsubseteq \mn{Adm}\wedge  Z\sqsubseteq \mn{Fac}  \wedge \neg(\exists z\mn{Teach}(y,z))\wedge   
\\& \quad\! x_1=\id(Y(y)) \wedge x_2=\id(Z(y)) \rightarrow \pref{x_1}{x_2}
\end{align*}
Rule $\sigma_1$ states a general preference for keeping more recently added facts. Rule $\sigma_2$ states if 
we have both $\mn{FPr}(p)$ and $\mn{APr}(p)$, we prefer to keep $\mn{FPr}(p)$, capturing the domain knowledge 
that associate professors are promoted into full professors. Rule $\sigma_3$ states that if a person is declared 
to belong both to a subclass of $\mn{Adm}$ and a subclass of \mn{Fac}, but there is no \mn{Teach}-fact for the person in the dataset,
then the $\mn{Adm}$-related facts are deemed more reliable. 
Observe that $\sigma_3$ uses ontology axioms with variables 
in order to simplify rule formulation (avoiding the need to write separate rules for every 
pair of subclasses of $\mn{Adm}$ and $\mn{Fac}$). 
\end{example}

\def\preflangdl{\ensuremath{\mathcal{PL}_{\mathsf{DL}}}}
\def\preflangpos{\ensuremath{\mathcal{PL}_{\mathsf{pos}}}}

\begin{definition}\label{specificpreflangdef}
The preference rule language \preflangdl\ contains rules whose bodies have the form $\varphi_{\mathsf{ont}}\wedge \varphi_{\mathsf{pos}}\wedge \varphi_{\mathsf{fo}}$ where: 
\begin{itemize}
\item 
$\varphi_{\mathsf{ont}}$ is of the form $X_1 \sqsubseteq P_1 \land \ldots \wedge X_k \sqsubseteq P_k$ with $X_i \in \vars$ and $P_i \in \preds \setminus \metapreds$, for $1\leq i \leq k$,
\item 
$\varphi_{\mathsf{pos}}$ is 
built from atomic statements using $\wedge$ and $\exists\,$,

\item $\varphi_{\mathsf{fo}}$ is 
built from atomic statements using $\wedge, \neg$ and $\exists\,$,
\end{itemize}
all free variables appear in $\varphi_{\mathsf{pos}}$, and atomic statements have the following forms:
\begin{itemize}
\item relational atoms $P(t_1, \ldots,t_n)$ such that $t_i \in \vars \cup \inds$ for $1\leq i\leq n$ and $P \in \preds_n \cup \{X_1,\dots,X_k\}$;
\item 
$x=\id(P(t_1, \ldots,t_n))$, where $x\in\vars$, $t_i \in \vars \cup \inds$ for $1\leq i\leq n$, and $P \in (\preds_n \setminus \metapreds) \cup  \{X_1,\dots,X_k\}$;
\item inequality atoms $x \neq t$, where $x\in \vars$ and $t \in \vars \cup \inds$.
\end{itemize}

The evaluation function for $\preflangdl$ is defined as follows: $\eval(\Kmc,\meta, \cond(x_1,x_2) ,id_1,id_2)=\true$ iff there exists a function $\nu$ that maps each free variable in $\cond(x_1,x_2)$ to an element of $\inds\cup\preds$ and is such that $\nu(x_1)=id_1$, $\nu(x_2)=id_2$ and $f(\Kmc,\meta, \nu(\cond(x_1,x_2)))=\true$, where $\nu(\cond(x_1,x_2))$ denotes the expression obtained by replacing each free variable $y$ by $\nu(y)$ in $\cond(x_1,x_2)$ and $f$ is defined recursively as follows: 
\begin{itemize}
\item $f(\Kmc,\meta,\phi_1\wedge \phi_2)=f(\Kmc,\meta,\phi_1)\wedge f(\Kmc,\meta,\phi_2)$, 
\item $f(\Kmc,\meta,\neg \phi)=\neg f(\Kmc,\meta,\phi)$, 
\item $f(\Kmc,\meta,\exists z\phi)=\true$ iff there exists $\nu_z:\{z\}\mapsto\inds$ such that $f(\Kmc,\meta,\nu_z(\phi))=\true$, 
\item for every $P\in\signature{\data}$ and tuple $\vec{c}$ of constants, $f(\Kmc,\meta,P(\vec{c}))=\true$ iff $\data\models P(\vec{c})$, 
\item for every $P\in\signature{\Kmc}$ and tuple $\vec{c}$ of constants, $f(\Kmc,\meta,P(\vec{c}))=\true$ iff $\Dmc\models P(\vec{c})$, 
\item  
$f(\Kmc,\meta,id_i=\id(P(\vec{c})))=\true$ iff $id_i=\id(P(\vec{c}))$,
\item for $c,d \in \inds$, $f(\Kmc,\meta,c\neq d)=\true$ iff $c\neq d$,
\item 
$f(\Kmc,\meta, A\sqsubseteq B)=\true$ iff $\Tmc\models A\sqsubseteq B$,
\item for any atom $\alpha$ of another form, $f(\Kmc,\meta,\alpha)=\false$.
\end{itemize}
\end{definition}

\begin{example}
By Definitions \ref{prefruledef} and \ref{specificpreflangdef}, 
$\rulesetex(\kbex,\mex)$ is:
$$\{\pref{2}{1},\pref{2}{3}, \pref{1}{3},\pref{6}{7}\}$$
with $\pref{2}{1}$ induced by both the first and second rules.
\end{example}

While preference rules allow users to describe in which cases one fact should be preferred to another,
we cannot immediately obtain a priority relation from $\Sigma(\Kmc,\meta)$. 
This is firstly because priority relations must satisfy the property that $\alpha\succ\beta$ implies 
that $\alpha$ and $\beta$ appear together in a conflict. 
While one could modify the definition of preference rules to enforce this property, 
it would lead to much more complicated rules, 
as users would need to include extra conditions 
in rule bodies to ensure only pairs of ids of conflicting facts occur in the head. 
We choose not to impose such a requirement, 
as it is more natural, we believe, to simply interpret a preference rule 
$\cond(x_1,x_2)\rightarrow \pref{x_1}{x_2}$
as meaning ``\emph{if the facts with ids $x_1$ and $x_2$ are in conflict}, and $\cond(x_1,x_2)$ is satisfied, then prefer fact $x_1$ to fact $x_2$". 
Formally, this means that instead of working with all pairs mentioned in $\Sigma(\Kmc,\meta)$, 
we consider the binary relation $\fullsucc$, defined by setting $\alpha\fullsucc \beta$ iff 
$\pref{\id(\alpha)}{\id(\beta)}\in\Sigma(\Kmc,\meta)$ \emph{and} there exists $\Cmc\in\conflicts{\Kmc}$ such that $\{\alpha,\beta\}\subseteq\Cmc$. 

The 
relation $\fullsucc$ may still fail to be a priority relation if it contains a cycle,
as priority relations are required to be acyclic. 
In what follows, we 
explore two complementary approaches to tackling this issue: 
identifying preference rules which are guaranteed to yield an acyclic relation, 
and employing different methods to extract an acyclic sub-relation. 

Finally let us note that while the definition of priority relation does not require 
transitivity, this is often considered a natural property for preferences. 
However, we argue that even in cases where transitivity is desired, one should first 
resolve any cycles in the `direct' preferences given in $\fullsucc$, 
then only afterwards close under transitivity. 

\subsection{Checking Acyclicity of Preference Rules}\label{sec:analysis}
It would be useful to be able to determine in advance, without knowing the dataset and meta-database,
whether a given set of preference rules is guaranteed to produce an acyclic relation (for example, to alert users
and allow them the option of modifying the rules if this is not the case). 
Let us first formalize precisely which property we aim to test:

\begin{definition}
Given a logical theory $\Tmc$, 
we say that a set $\Sigma$ of preference rules is \emph{
$\Tmc$-acyclic} if 
for every dataset $\Dmc$ 
and every meta-database 
$\meta$ for the KB $\Kmc=(\Dmc,\Tmc)$, the induced binary relation $\fullsucc$ is acyclic. 
\end{definition}

The decidability and complexity of verifying $\Tmc$-acyclicity naturally
depends on the expressivity of the logical theory and rule bodies. 
For our proposed language
$\preflangdl$, the problem is typically undecidable, since 
finite satisfiability of FO-sentences can be reduced to $\Tmc$-acyclicity: 

\begin{restatable}{theorem}{undecthm}\label{undec-thm}
Let $\Tmc$ be any non-trivial theory (i.e.\ which can generate some conflict). 
Then it is undecidable to test whether a set $\Sigma \subseteq \preflangdl$
is $\Tmc$-acyclic. 
\end{restatable}

We now present a positive result that covers some prominent 
ontology and constraint languages and supports reasonably expressive rule 
bodies. Specifically, we consider the 
language $\preflangpos$
obtained from $\preflangdl$ by disallowing ontology atoms and negation (retaining inequality atoms $x\neq t$), i.e., rule bodies only contain the component $\varphi_{\mathsf{pos}}$.

\begin{restatable}{theorem}{tacyclicthm}\label{t-acyclic-thm}
Given a theory $\Tmc$ consisting of binary denial constraints 
and a set $\Sigma$ of preference rules from $\preflangpos$, 
it is decidable whether $\Sigma$ is $\Tmc$-acyclic. Moreover, 
the problem can be decided in $\conp$ if the predicate arity is bounded. 
\end{restatable}

\begin{restatable}{corollary}{cordllite}\label{cor:dllite}
$\Tmc$-acyclicity testing is in \conp\
if $\Tmc$ is a DL-Lite ontology and the preference ruleset is in $\preflangpos$.  
\end{restatable}

We expect that the preceding result can be extended to arbitrary denial constraints 
(and ontology languages with bounded-size non-binary conflicts), but the
argument will become considerably more involved as one needs to ensure that
the shortened cycle constructed in the proof only involves pairs of facts 
that co-occur in a conflict. We observe however that the proof of Theorem \ref{t-acyclic-thm}
already provides us with a procedure for checking acyclicity of $\Sigma(\Kmc,\meta)$,
which provides a sufficient condition for $\Tmc$-acyclicity: 

\begin{definition}
We say that a set $\Sigma$ of preference rules is \emph{
strongly acyclic 
} 
if for every KB $\Kmc=(\Dmc,\Tmc)$ and 
every meta-database 
$\meta$ for $\Kmc$, 
the binary relation $\{(\alpha,\beta) \mid \pref{\id(\alpha)}{\id(\beta)}\in\Sigma(\Kmc,\meta)\}$ 
 is acyclic. 
\end{definition}

\begin{restatable}{theorem}{strongacyclicprop}\label{strong-acyclic-prop}
If $\Sigma$ is strongly acyclic, 
then it is $\Tmc$-acyclic. 
\end{restatable}

\begin{example} 
The ruleset $\Sigma = \{\sigma_2\}$ is strongly acyclic, as $\sigma_2$ can only induce 
$\pref{\id(\alpha)}{\id(\beta)}$ 
if $\alpha$ is an $\mn{FPr}$-fact 
and $\beta$ a $\mn{APr}$-fact, so no cycle can be constructed. 
\end{example}

It is also interesting to observe that if some metadata predicates enjoy special properties,
this information could be exploited to identify additional acyclic rulesets. 

\begin{example} Suppose now that $\Sigma = \{\sigma_1\}$.  Naturally we expect that the meta-database
contains a 
unique fact $\mn{Date}(\id(\alpha), d)$ for each fact $\alpha$
and that 
$<$ provides a total order over the values in the second argument of $\mn{Date}$. 
If we were to adapt our acyclicity notions to only quantify over meta-databases satisfying these constraints,
then we could conclude that 
$\Sigma$ is strongly acyclic. 
\end{example}

We leave it as future work to develop more sophisticated ($\Tmc$- or strong) acyclicity
checking procedures that can take into account such additional information.

\begin{table*}
\renewcommand{\arraystretch}{0.9}
\scalebox{0.85}{
\begin{tabular*}{1.15\textwidth}{l l l}
\toprule
 & \multicolumn{1}{c}{program facts and rules} &  \multicolumn{1}{c}{input encoded} 
\\
\midrule
$\Pi_\Dmc$ & $\mt{data}(i).$   &$P(c_1,\dots,c_n)\in\Dmc$, 
\\
& ${P}(i, c_1,\dots, c_n).$& $\id(P(c_1,\dots,c_n))=i$
\\
\midrule
$\Pi_\data$ & ${Q}(c_1,\dots, c_n).$ & $Q(c_1,\dots,c_n)\in\data$
\\
\midrule
$\Pi_C$ &  $\mt{conf\_init}((\mt{Id0},\dots, \mt{Idk}))\text{ :- }{P_0}(\mt{Id0},t^0_1,\dots,t^0_{n_0}),\dots,{P_k}(\mt{Idk},t^k_1,\dots,t^k_{n_k}).$   &  $\bigwedge_{i=0}^k P_i(t^i_1,\dots,t^i_{n_i})\rightarrow \bot\in \mi{Inc}(\Tmc)$
\\
& $\mt{inConf\_init}((\mt{Id0},\dots, \mt{Idk}), \mt{Idj})\text{ :- }{P_0}(\mt{Id0},t^0_1,\dots,t^0_{n_0}),\dots,{P_k}(\mt{Idk},t^k_1,\dots,t^k_{n_k}).$ 
& 
\\
\midrule
$\Pi_Q$ &  $\mt{cause}((x_0,\dots,x_m), (\mt{Id0},\dots, \mt{Idk}))\text{ :- }{P_0}(\mt{Id0},t^0_1,\dots,t^0_{n_0}),\dots,{P_k}(\mt{Idk},t^k_1,\dots,t^k_{n_k}).$   &  $\exists \vec{y}\bigwedge_{i=0}^k P_i(t^i_1,\dots,t^i_{n_i})$
\\
& $\mt{inCause}((\mt{Id0},\dots, \mt{Idk}), \mt{Idj})\text{ :- }{P_0}(\mt{Id0},t^0_1,\dots,t^0_{n_0}),\dots,{P_k}(\mt{Idk},t^k_1,\dots,t^k_{n_k}).$ 
& \multicolumn{1}{r}{$\rightarrow q(x_0,\dots,x_m) \in \mi{Rew}(q,\Tmc)$}
\\
\midrule
$\Pi_P$ & 
$\mt{pref\_init}(x_1, x_2, i)\text{ :- }\mt{inConf}(\mt{C}, x_1), \mt{inConf}(\mt{C}, x_2), $ 
& 
$\cond(x_1,x_2)\rightarrow \pref{x_1}{x_2}\in\Sigma_i$
\\
&
\phantom{$\mt{pref\_init}(x_1, x_2, i)$ :- }   
${P_0}(\mt{X0},t^0_1,\dots,t^0_{n_0}),\dots,{P_k}(\mt{Xk},t^k_1,\dots,t^k_{n_k}),$
 & 
 $\cond(x_1,x_2)=\exists \vec{y}\bigwedge_{i=0}^k P_i(t^i_1,\dots,t^i_{n_i})\wedge $
\\
&
\phantom{$\mt{pref\_init}(x_1, x_2, i)$ :- }   
$\mt{not}~{P'_0}(\mt{Y0},t'^0_1,\dots,t'^0_{n'_0}),\dots,\mt{not}~{P'_{k'}}(\mt{Yk'},t'^{k'}_1,\dots,t'^{k'}_{n'_{k'}}),$
 & 
 $\bigwedge_{i=0}^{k'}\neg P'_i(t'^i_1,\dots,t'^i_{n'_i})\wedge $
\\
&
\phantom{$\mt{pref\_init}(x_1, x_2, i)$ :- }  
${Q_0}(l^0_1,\dots,l^0_{p_0}),\dots,{Q_m}(l^m_1,\dots,l^m_{p_m}),$
&
$\bigwedge_{i=0}^m Q_i(l^i_1,\dots,l^i_{p_i}) \wedge$
\\
&
\phantom{$\mt{pref\_init}(x_1, x_2, i)$ :- }  
$\mt{not}~{Q'_0}(l'^0_1,\dots,l'^0_{p'_0}),\dots,\mt{not}~{Q'_{m'}}(l'^{m'}_1,\dots,l^{m'}_{p'_{m'}}),$
&
$\bigwedge_{i=0}^{m'} \neg Q'_i(l'^i_1,\dots,l'^i_{p'_i}) \wedge$
\\

&
\phantom{$\mt{pref\_init}(x_1, x_2, i)$ :- }  
$f^0_1\bowtie f^0_2,\dots,f^r_1\bowtie f^r_2,$ 
&
$\bigwedge_{\ell=0}^r f^\ell_1\bowtie f^\ell_2\wedge$ 
\\
&
\phantom{$\mt{pref\_init}(x_1, x_2, i)$ :- }  
$P''_1(t_1, t''^1_1,\dots,t''^1_{n'_1}), \dots, P''_q(t_q, t''^q_1,\dots,t''^q_{n'_q}).$
 &
 $ \bigwedge_{i=1}^q t_i=\id(P''_i(t''^i_1,\dots,t''^i_{n''_i}))$
\\
&$\mt{level}(i).$
\\
\bottomrule
\end{tabular*}
}
\caption{Logic programs encoding the input. $P,P_i,P'_i,P''_i\in \signature{\Dmc}$, $Q,Q_j\in \signature{\data}$, terms are in 
$\inds\cup\vars$ 
and $\bowtie\in\!\{=,\neq, >,<, \geq, \leq\}$. 
}\label{tab:program-input}
\end{table*}

\subsection{Resolving Cycles to Get a Priority Relation} 
Ideally the preference ruleset would satisfy the introduced acyclicity conditions, 
but this cannot be assumed in general. Indeed, we have seen that it may be undecidable 
to determine whether a given ruleset satisfies the conditions. Furthermore, 
cycles can naturally arise when users create rules that capture different criteria, 
e.g.\ prefer more recent facts and prefer facts from more trusted sources. 
To ensure acyclicity in such cases, 
one would need to create more complex rules 
whose bodies consider different combinations of the criteria, 
making rules 
much harder for users to specify and understand. 
We thus advocate a pragmatic approach: give users free rein to specify 
preferences as they see fit, then apply cycle resolution techniques to extract
a suitable acyclic sub-relation should any cycles arise. 

To enable a more fine-grained specification of the preferences, we allow users to partition the set $\Sigma$ of preference rules into priority levels $\Sigma_1,\dots,\Sigma_n$, so that a preference induced by a preference rule from $\Sigma_i$ is considered more important than one induced by a preference rule from $\Sigma_j$ with $j>i$, and will thus be preferably kept in the cycle elimination process. If no such partition is specified, then all rules are assigned to $\Sigma_1$. 
For every $\pref{\id(\alpha)}{\id(\beta)}\in\Sigma(\Kmc,\meta)$, 
we denote by $\level(\alpha,\beta)$ the \emph{minimal} index $i$ such that $\pref{\id(\alpha)}{\id(\beta)}\in\Sigma_i(\Kmc,\meta)$. 
We consider several ways of removing cycles to obtain a priority relation $\succ$ from $\fullsucc$
(abbreviated 
to $\succ_\Sigma$ in what follows): 

\begin{itemize}
\item Going up ($\succup$): Let $\succup:=\emptyset$ and $i:=1$. Then while $\succup\!\cup\!\succ_{\Sigma_i}$ is acyclic, let $\succup:=\succup\!\cup\!\succ_{\Sigma_i}$ and increment $i$. 
\item Going down ($\succdown$): Let $\succdown:=\succ_{\Sigma}$ and $i:=n$. Then while $\succdown$ is cyclic, let $\succdown:=\succdown\!\setminus \{(\alpha,\beta)\mid \level(\alpha,\beta)=i, (\alpha,\beta) \text{ is in a cycle \wrt}\succdown\}$ and decrement $i$. 
\item 
Refined going up ($\succrefup$): Let $\succrefup:=\succ_{\Sigma_1}$, then remove every $(\alpha,\beta)$ that occurs in a cycle w.r.t.\ $\succ_{\Sigma_1}$. Then for $i=2$ to 
$n$, add to $\succrefup$ all pairs $(\alpha,\beta)$ such that $\level(\alpha,\beta)=i$ and $(\alpha,\beta)$ does not belong to any cycle \wrt $\succrefup\!\cup\!\succ_{\Sigma_i}$. 
\item Grounded ($\succground$): Let $\succground:=\emptyset$. Then until a fixpoint is reached, add to $\succground$ all pairs $(\alpha,\beta)$ such that $\alpha\succ_\Sigma\beta$ and for every cycle $c$ of $\succ_\Sigma$ containing $(\alpha,\beta)$, either there is $(\gamma,\delta)\in c$ such that $\level(\alpha,\beta)<\level(\gamma,\delta)$, or there is $(\gamma,\delta)\in c$ such that $ \succground\cup\{(\gamma,\delta)\}$ is cyclic. 
\end{itemize}

We next relate the preceding 
strategies to notions that have been proposed 
in the literature to select a single consistent set of facts from a KB whose dataset is partitioned into priority levels. 
Indeed, one can define the KB $\Kmc^{cy}=(\Dmc^{cy},\Tmc^{cy})$ with $\Dmc^{cy}=\{ R(\id(\alpha),\id(\beta)) \mid \alpha\succ_\Sigma\beta\}$ and $\Tmc^{cy}=\{R(x,y)\wedge R(y,z)\rightarrow R(x,z),\ R(x,x)\rightarrow\bot\}$, whose conflicts correspond exactly to the minimal cycles of $\succ_\Sigma$, and further partition $\Dmc^{cy}$ into $\Dmc^{cy}_1,\dots,\Dmc^{cy}_n$ as follows: $R(\id(\alpha),\id(\beta))\in \Dmc^{cy}_i$ iff $\level(\alpha,\beta)=i$. 
For such a KB $\Kmc$ whose dataset is partitioned into priority levels $\Dmc_1,\dots,\Dmc_n$, \citeauthor{DBLP:conf/ijcai/BenferhatBT15}~\shortcite{DBLP:conf/ijcai/BenferhatBT15} defined the \emph{possibilistic repair} $\mn{Poss}(\Kmc)=\Dmc_1\cup\dots\cup\Dmc_{\mi{inc}(\Kmc)-1}$ where $\mi{inc}(\Kmc)$ is the inconsistency degree of $\Kmc$, \ie the minimal $i$ such that $\Dmc_1\cup\dots\cup\Dmc_i$ is inconsistent;  
the \emph{non-defeated repair} $\mn{NonDef}(\Kmc)$, defined as the union of the intersections of the (subset) repairs of $\Dmc_1$, $\Dmc_1\cup\Dmc_2$, $\dots$, $\Dmc_1\cup\dots\cup\Dmc_n$; and the \emph{prioritized inclusion-based non-defeated repair} $\mn{Prio}(\Kmc)$, defined 
similarly to $\mn{NonDef}(\Kmc)$ but considering optimal repairs instead of subset repairs. Indeed, when a priority relation is induced from priority levels (called 
score-structured 
in the literature), the three notions of optimal repairs coincide, and can be defined directly from the priority levels \cite{DBLP:phd/hal/Bourgaux16,DBLP:conf/pods/LivshitsK17}. 
Finally, \citeauthor{DBLP:conf/kr/BienvenuB20}~\shortcite{DBLP:conf/kr/BienvenuB20} defined the preference-based set-based argumentation framework associated with a prioritized KB $\Kmc$, whose arguments are the KB facts and attacks are defined from the KB conflicts, and considered its \emph{grounded extension} $\mn{Grd}(\Kmc)$.

\begin{restatable}{theorem}{relationprefrepairslevels}\label{relationprefrepairslevels}
It holds that:
\begin{itemize}
\item $\alpha\succup\beta$ iff $R(\id(\alpha),\id(\beta))\in\mn{Poss}(\Kmc^{cy})$, 
\item $\alpha\succdown\beta$ iff $R(\id(\alpha),\id(\beta))\in\mn{NonDef}(\Kmc^{cy})$, 
\item $\alpha\succground\beta$ iff $R(\id(\alpha),\id(\beta))\in\mn{Grd}(\Kmc^{cy})$. 
\end{itemize}
\end{restatable}

It has been shown that $\mn{Poss}(\Kmc)\subseteq \mn{NonDef}(\Kmc) \subseteq \mn{Grd}(\Kmc)\subseteq\mn{Prio}(\Kmc)$ and that all these sets of facts can be computed in polynomial time except for $\mn{Prio}(\Kmc)$ \cite{DBLP:conf/ijcai/BenferhatBT15,DBLP:conf/kr/BienvenuB20}. 
Combined with Theorem~\ref{relationprefrepairslevels}, these results can help us show the following theorems:

\begin{restatable}{theorem}{relationsucc}\label{relationsucc}
$\succup\subseteq \succdown \subseteq \succground$ and 
$\succup\subseteq \succdown \subseteq \succrefup$.
\end{restatable}

\begin{restatable}{theorem}{relationcomplexity}\label{relationcomplexity}
Each of the relations $\succup, \succdown, \succground, \succrefup$
can be computed in polynomial time from the relations $\succ_{\Sigma_i}$.
\end{restatable}

Examples~\ref{ex-ru-g-1} and~\ref{ex-ru-g-2} show that $\succground$ and $\succrefup$ are incomparable and that it may be the case that $\alpha\succrefup\beta$ while $R(\id(\alpha),\id(\beta))\notin\mn{Prio}(\Kmc^{cy})$.

\begin{example}\label{ex-ru-g-1}
Assume that $\succ_{\Sigma_1}=\{(\alpha,\beta), (\beta,\gamma)\}$, and that $\succ_{\Sigma_2}=\{(\alpha,\gamma),(\gamma,\alpha)\}$. Then $\succrefup=\{(\alpha,\beta),(\beta,\gamma)\}$ while $\succground=\{(\alpha,\beta),(\beta,\gamma),(\alpha,\gamma)\}$, so $\succrefup\subsetneq\succground$.
\end{example}

\begin{example}\label{ex-ru-g-2}
Assume that $\succ_{\Sigma_1}=\{(\alpha,\beta), (\gamma,\delta)\}$,  $\succ_{\Sigma_2}=\{(\beta,\gamma),(\delta,\alpha)\}$, and $\succ_{\Sigma_3}=\{(\gamma,\beta)\}$. 
Then $\succrefup=\{(\alpha,\beta), (\gamma,\delta), (\gamma,\beta)\}$ while $\succground=\{(\alpha,\beta), (\gamma,\delta)\}$, so $\succground\subsetneq\succrefup$. Note that $R(\id(\gamma),\id(\beta))\notin\mn{Prio}(\Kmc^{cy})$ since $\{R(\id(\alpha),\id(\beta)), R(\id(\gamma),\id(\delta)), R(\id(\beta),\id(\gamma))\}$ is an optimal repair of $\Kmc^{cy}$.
\end{example}


\section{ASP Implementation}\label{sec:implem}

We implement our approach using \emph{answer set programming} (ASP) \cite{DBLP:books/sp/Lifschitz19,DBLP:series/synthesis/2012Gebser}. 
We consider ASP \emph{programs} consisting of \emph{rules} of the form $$\gamma\text{ :- } \alpha_1,\dots, \alpha_n, \mt{not}\beta_1,\dots,\mt{not}\beta_m.$$ where $\gamma, \alpha_i,\beta_j$ are atoms built from predicates, variables, constants and comparison operators. 
Every variable occurring in the head $\gamma$ of a rule must also occur in some positive literal of its body $\alpha_1,\dots, \alpha_n, \mt{not}\beta_1,\dots,\mt{not}\beta_m$. A rule with an empty body is 
a \emph{fact}, and a rule with an empty head a \emph{constraint}. 
We also use \emph{choice rules} to select exactly or at least one atom from a set. Importantly, it is possible to use a tuple of terms as a predicate argument. We use this to define, e.g., conflict identifiers as the tuple of the identifiers of their facts. 
ASP is based on the \emph{stable model} semantics. 

We implement several building blocks, which provide an almost end-to-end approach to querying inconsistent KBs. 
Our system takes as input logic programs representing the input, and computes the query answers under the chosen semantics among $X$-\brave, $X$-\AR or $X$-\IAR with $X\in\{S,P,C\}$ \wrt $\succ^x$ for the chosen $x\in\{u,d,ru,g\}$. 
All building blocks can be encoded into ASP programs that a \mn{\mn{Python}} program combines and passes to the ASP solver \mn{clingo}\footnote{\url{https://github.com/potassco/clingo}} \cite{DBLP:journals/aicom/GebserKKOSS11} to check whether the resulting program has a stable model. 
However, we found more efficient in practice to split the computation into several steps and implement some of them in \mn{\mn{Python}} (see Section~\ref{sec:inputencoding}).

\begin{table*}
\renewcommand{\arraystretch}{0.7}
\scalebox{0.85}{
\begin{tabular*}{1.15\textwidth}{l l }
\toprule

$\Pi_{\succup}$ & 
\mt{trans\_cl(X, Y, I)\text{ :- }pref\_init(X, Y, I), not~blocked(I)}.
\\&
\mt{trans\_cl(X, Y, I)\text{ :- }level(I), trans\_cl(X, Y, J), J<I, not~blocked(I)}.
\\&
\mt{trans\_cl(X, Y, I)\text{ :- }pref\_init(X, Z, J), trans\_cl(Z, Y, I), J<=I, not~blocked(I)}. 
\\&
\mt{cycle(I)\text{ :- }trans\_cl(X, X, I)}. 

\\&
\mt{blocked(I)\text{ :- }level(I), cycle(J), J<I}. 
\quad\quad\quad\quad\quad\quad\quad\quad
\mt{pref(X, Y)\text{ :- }pref\_init(X, Y, I), not~cycle(I), not~blocked(I)}.
\\
\midrule
$\Pi_{\succdown}$ &
\mt{trans\_cl(X, Y, I)\text{ :- }pref\_init(X, Y, I)}.
\\&
\mt{trans\_cl(X, Y, I)\text{ :- }pref\_init(X, Z, I), trans\_cl(Z, Y, J), J<=I}.
\\&
\mt{trans\_cl(X, Y, I)\text{ :- }pref\_init(X, Z, J), trans\_cl(Z, Y, I), J<=I}.
\\&
\mt{cycle(X, Y, I)\text{ :- }pref\_init(X, Y, I), trans\_cl(Y, X, I)}.
\quad\quad\quad\quad\quad\quad\quad\quad
\mt{pref(X, Y)\text{ :- }pref\_init(X, Y, I), not~cycle(X, Y, I)}.
\\
\midrule
$\Pi_{\succrefup}$ &
\mt{trans\_cl(X, Y, I)\text{ :- }pref\_init(X, Y, I)}.
\\&
\mt{trans\_cl(X, Y, I)\text{ :- }level(I), rel(X, Y, J), J < I}.\\&
\mt{trans\_cl(X, Y, I)\text{ :- }pref\_init(X, Z, I), trans\_cl(Z, Y, I).}\\&
\mt{trans\_cl(X, Y, I)\text{ :- }trans\_cl(X, Z, I), trans\_cl(Z, Y, I).}
\\&
\mt{cycle(X, Y, I)\text{ :- }pref\_init(X, Y, I), trans\_cl(Y, X, I)}. \\&
\mt{rel(X, Y, I)\text{ :- }pref\_init(X, Y, I), not~cycle(X, Y, I)}.\\&
\mt{rel(X, Y, I)\text{ :- }rel(X, Z, J), rel(Z, Y, I), J<=I}. 
\quad\quad\quad\quad\quad\quad\quad\quad
\mt{pref(X, Y)\text{ :- }pref\_init(X, Y, I), rel(X, Y, I)}.
\\
\bottomrule
\end{tabular*}
}
\caption{Logic programs to compute $\succ^x$ from facts on predicates \mt{conf}, \mt{inConf}, $\mt{pref\_init}$ and \mt{level}.}\label{tab:program-prio-from-pref}
\end{table*}

\subsection{Input, Conflicts, Causes and Preferences}\label{sec:inputencoding}
Our approach applies to any logical theory $\Tmc$ such that: 
\begin{enumerate}
\item there exists a set $\mi{Inc}(\Tmc)$ of rules of the form $q\rightarrow \bot$ with $q$ a Boolean CQ with inequalities, such that for every dataset $\Dmc$, $(\Dmc,\Tmc)\models\bot$ iff there exists $q\rightarrow \bot\in\mi{Inc}(\Tmc)$ such that $\Dmc\models q$; and 
\item for every CQ $q(\vec{x})$ there exists a set $\mi{Rew}(q,\Tmc)$ of rules of the form $q'(\vec{x})\rightarrow q(\vec{x})$ with $q'$ a CQ such that for every $\Dmc$ s.t.~$(\Dmc,\Tmc)\not\models\bot$ and tuple 
$\vec{a}$, $(\Dmc,\Tmc)\models q(\vec{a})$ iff there exists $q'(\vec{x})\rightarrow q(\vec{x})\in\mi{Rew}(q,\Tmc)$ s.t.~$(\Dmc,\Tmc)\models q'(\vec{a})$. 
\end{enumerate}
These conditions are fulfilled, \eg when $\Tmc$ is a set of denial constraints (then, $\mi{Inc}(\Tmc)=\Tmc$ and $\mi{Rew}(q,\Tmc)=\{q\rightarrow q\}$), or when $\Tmc$ is a DL-Lite ontology. 
Regarding preference rules, we handle rules whose bodies are CQs with negation and comparison operators (see Table~\ref{tab:program-input} for the syntax).

We expect that the KB $\Kmc=(\Dmc,\Tmc)$, meta-database $\meta=(\id,\data)$, preference rules $\Sigma=\Sigma_1\cup\dots\cup\Sigma_n$, and query $q$ have been transformed into the five ASP programs 
given in Table~\ref{tab:program-input}. 
Programs $\Pi_\Dmc$ and $\Pi_\data$ represent the dataset $\Dmc$ and the identifier function $\id$, and the meta-database respectively, and can be obtained quite straightforwardly from various data formats. 
Constructing $\Pi_C$ and $\Pi_Q$, which encode 
the constraints and queries, is more demanding since it requires to compute the sets $\mi{Inc}(\Tmc)$ and $\mi{Rew}(q,\Tmc)$. 

\begin{restatable}{proposition}{propInputEncoding}\label{propInputEncoding}
The program $\Pi_\Dmc\cup\Pi_Q$ has a single stable model $\Smc$, and for every $\{\alpha_0,\dots,\alpha_k\}\in\causes{q(\vec{a}),\Kmc}$, $\Smc$ contains the facts $\mt{cause}((\vec{a}),(\id(\alpha_0),\dots,\id(\alpha_k)))$ and $\mt{inCause}((\id(\alpha_0),\dots,\id(\alpha_k)), \id(\alpha_j))$, $0\!\leq\! j\!\leq\! k$. 
Moreover, if $\mt{cause}((\vec{a}),(\id(\alpha_0),\dots,\id(\alpha_k)))\in\Smc$, then $(\{\alpha_0,\dots,\alpha_k\},\Tmc)\models q(\vec{a})$.
\end{restatable}
Essentially, $\Pi_\Dmc\cup\Pi_Q$ computes a \emph{superset} of $\causes{q(\ans),\Kmc}$, such that each superfluous $\Bmc$ either includes a real cause of $q(\ans)$ or contains a conflict. 
Similarly, $\Pi_\Dmc\cup\Pi_C$ computes a \emph{superset} of $\conflicts{\Kmc}$, such that each superfluous $\Bmc$ contains an actual conflict. 
To obtain $\conflicts{\Kmc}$, we filter out these non-minimal \Tmc-inconsistent subsets 
either via an ASP program $\Pi_{\mi{minC}}$ 
or by 
a \mn{Python} program, 
which we found faster in practice. In the case where conflicts are of size at most two, we further optimize the program by relying on the fact that non-minimal \Tmc-inconsistent subsets we compute are not conflicts only if they contain some self-inconsistent fact.
We do not~need to filter out the superfluous sets from the superset of $\causes{q(\ans),\Kmc}$, and only need to ensure that they do not~contain some self-inconsistent fact (\cf \cite[Section 4]{DBLP:conf/kr/BienvenuB22}), which we do using \mn{Python}.

\begin{restatable}{proposition}{propMinConf}\label{propMinConf}
The program $\Pi_\Dmc\cup\Pi_C\cup\Pi_{\mi{minC}}$ has a single stable model $\Smc$, which is such that  
$\{\alpha_0,\dots,\alpha_k\}\in\conflicts{\Kmc}$ iff $\Smc$ contains $\mt{conf}((\id(\alpha_0),\dots,\id(\alpha_k)))$ and $\mt{inConf}((\id(\alpha_0),\dots,\id(\alpha_k)), \id(\alpha_j))$, $0\leq j\leq k $. 
\end{restatable}

Finally, $\Pi_P$ encodes the preference rules with their priority levels. Note that we add in the preference rule body the condition that the two facts compared in the head belong to the same conflict to compute directly the $\succ_{\Sigma_i}$'s. 

\begin{restatable}{proposition}{propInputPref}\label{propInputPref}
The program $\Pi_\Dmc\cup\Pi_\data\cup\Pi_C\cup\Pi_{\mi{minC}}\cup\Pi_P$ has a single stable model $\Smc$, which is such that for every $\alpha,\beta\in\Dmc$, $\alpha\succ_{\Sigma_i}\beta$ iff $\mt{pref\_init}(\id(\alpha), \id(\beta), i)\in\Smc$. 
\end{restatable}

\subsection{Computing the Priority Relation}\label{sec:oriented-conflict-hypergraph}

We compute $\succ^x$ for the chosen $x\in\{u,d,ru,g\}$ from the conflicts given by facts on predicates \mt{conf}, \mt{inConf}, and the $\succ_{\Sigma_i}$'s given by $\mt{pref\_init}$ with $\Pi_{\succ^x}$. For $x\in\{u,d,ru\}$, $\Pi_{\succ^x}$ is given in Table~\ref{tab:program-prio-from-pref} (for space reasons, we omit $\Pi_{\succground}$, which draws inspiration from the ASP encoding of the grounded extension from \cite{DBLP:conf/iclp/EglyGW08}). 

\begin{restatable}{proposition}{propMinConfAndPrio}\label{propMinConfAndPrio}
The program $\Pi_\Dmc\cup\Pi_\data\cup\Pi_C\cup\Pi_{\mi{minC}}\cup\Pi_P\cup\Pi_{\succ^x}$ has a single stable model $\Smc$ which is such that for all $\alpha,\beta\in\Dmc$, $\alpha\succ^x\beta$ iff $\mt{pref}(\id(\alpha), \id(\beta))\in\Smc$. 
\end{restatable}

\subsection{Optimal Repair-Based Semantics}

After preliminary experiments, 
we found it more efficient to treat each potential answer separately, so we transform (using \mn{\mn{Python}}) 
the 
$\mt{cause}((\vec{a}),(\id(\alpha_0),\dots,\id(\alpha_k)))$ facts built by $\Pi_\Dmc\cup\Pi_Q$ into a set of programs $\Pi_{\ans}$ representing 
causes of each 
$\ans$ with facts $\mt{cause}((\id(\alpha_0),\dots,\id(\alpha_k)))$ and $\mt{inCause}((\id(\alpha_0),\dots,\id(\alpha_k)), \id(\alpha_j))$. 
For the ease of presentation, we also denote by $\Pi_{\mi{conf}_\succ}$ the logic program that contains the conflicts and priority relation. 
We say that a conflict $\Cmc$ \emph{attacks} a fact $\alpha$, written $\Cmc\rightsquigarrow\alpha$, if $\alpha\in\Cmc$ and $\alpha\not\succ\beta$ for every $\beta\in\Cmc$. We use a program $\Pi_{\mi{att}}$ to pre-compute the attack relation $\rightsquigarrow$ ($\mt{att}$) from $\Pi_{\mi{conf}_\succ}$.

For $X\in\{S,P,C\}$ and $\sem\in\{\brave,\AR,\IAR\}$, we define $\Pi_{X\text{-}\sem}$ from building blocks inspired by the SAT encodings given by \citeauthor{DBLP:conf/kr/BienvenuB22}~\shortcite{DBLP:conf/kr/BienvenuB22}. 
Note, however, that the latter are implemented for \emph{binary conflicts}, so our system is the first implementing optimal repair-based inconsistency-tolerant semantics for \emph{conflicts of arbitrary size}. 
For $\sem\in\{\brave, \AR\}$, 
$\Pi_{X\text{-}\sem}$ is the union of: 
\begin{itemize}
\item $\Pi_{\mi{loc}}$, which localizes the attack relation to relevant facts (those that are reachable from the causes); 
\item $\Pi_{\mi{cons}}$, which selects (using a choice rule) a consistent set of facts among the relevant facts by enforcing that at least one fact per relevant conflict is removed;
\item $\Pi_{\mi{brave}}$ if $\sem=\brave$, which ensures that $\Pi_{X\text{-}\sem}$ is satisfiable only if all facts of some cause are selected;
\item $\Pi_{\mi{AR}}$ if $\sem=\AR$, which ensures that $\Pi_{X\text{-}\sem}$ is satisfiable only if every cause is contradicted by the selected facts, meaning that these facts include $\Cmc\setminus\{\alpha\}$ for some $\Cmc\rightsquigarrow\alpha$ with $\alpha$ a fact of the cause;
\item $\Pi_{\mi{Pareto}}$ if $X=P$ (\resp $\Pi_{\mi{Completion}}$ if $X=C$), which ensures that $\Pi_{X\text{-}\sem}$ is satisfiable only if the selected facts can be extended into a Pareto- (\resp completion-) 
repair.
\end{itemize}
For $\sem=\IAR$, $\Pi_{X\text{-}\sem}$ intuitively checks whether each cause can be contradicted by a consistent set of facts. It is similar to $\Pi_{X\text{-}\AR}$, except that predicates in $\Pi_{\mi{loc}}$, $\Pi_{\mi{cons}}$, $\Pi_{\mi{AR}}$ and $\Pi_{\mi{Pareto}}$ or $\Pi_{\mi{Completion}}$ are extended with an extra argument that keeps the identifier of the cause considered. 

\begin{restatable}{proposition}{propMinConfAndPrio}\label{propMinConfAndPrio}
The program $\Pi_{\mi{conf}_\succ}\cup\Pi_{\ans}\cup\Pi_{\mi{att}}\cup\Pi_{X\text{-}\sem}$ has a stable model iff 
\begin{enumerate}
\item $\Kmc_\succ \semmodels{X} q(\ans)$ if $\sem=\brave$;
\item $\Kmc_\succ \not\semmodels{X} q(\ans)$ if $\sem\in\{\AR,\IAR\}$.
\end{enumerate}
\end{restatable}


\section{Experiments}\label{sec:expe}

Our main goal is to compare the different approaches to obtaining a priority relation from preferences rules, in terms of run time and size of the priority relation. We also compare our ASP implementation of the optimal repair-based semantics with {\sc orbits}, the existing SAT-based implementation. 

\begin{table}
{\scriptsize
\hspace{-2mm}
\scalebox{1}{
\renewcommand{\arraystretch}{0.7}
\setlength\tabcolsep{1.8pt}
\begin{tabular*}{0.47\textwidth}{l @{\extracolsep{\fill}} r  r r r r r r r r}
\toprule
& &\multicolumn{4}{c}{$\Sigma^{a}_1\cup\Sigma^{a}_2\cup\Sigma^{a}_3$} & \multicolumn{3}{c}{$\Sigma^{c}_1$} & \multicolumn{1}{c}{$\Sigma^{d}_1$} 
\\
&\multicolumn{1}{c}{\#conf.} &\multicolumn{1}{c}{$\succ_\Sigma$} &  \multicolumn{1}{c}{$\succup$}  & \multicolumn{1}{c}{$\succdown$} & \multicolumn{1}{c}{$\succground$} &  \multicolumn{1}{c}{$\succ_\Sigma$} &\multicolumn{1}{c}{$\succup$} &\multicolumn{1}{c}{$\succ^{d,g}$}  & \multicolumn{1}{c}{$\succ$} 
\\
\midrule
\mn{u1c1}& 2,354 & 7,068 & 3,041 & 3,644 & 3,703 & 5,633 & 0 & 1,510 & 0
\\
\midrule
\mn{u1c5}& 8,516 & 17,804 & 7,624 & 8,944 & 9,324 & 14,517 & 0 & 3,356 & 1
\\
\midrule
\mn{u1c10}& 14,301 & 27,927 & 2 & 14,402 & 14,808 & 22,634 & 0 & 6,082 & 2
\\
\midrule
\mn{u1c20}& 28,272 & 52,361 & 4 & 27,185 & 27,948 & 42,032 & 0 & 12,300 & 4
\\
\midrule
\mn{u1c30}& 45,524 & 82,531 & 6 & 41,300 & 42,601 & 65,142 &$\_$& 16,361 & 6
\\
\midrule
\mn{u1c50}& 81,344 & 145,193 & $\_$ & 69,454 & $\_$ & 113,857 & $\_$ & 19,966 & 8
\\
\midrule
\mn{u5c1}& 12,024 & 23,932 & 10,241 & 13,275 & 13,339 & 18,570 & 0 & 7,821 & 0
\\
\midrule
\mn{u5c5}& 53,438 & 96,307 & $\_$ & 52,084 & 52,820 & 76,045 & 0 & 28,017 & 1
\\
\midrule
\mn{u5c10}& 109,493 & 194,306 & 6 & 103,094 &  $\_$ & 154,271 &  $\_$ & 50,673 & 6
\\
\midrule
\mn{u5c20}& 231,811 &  $\_$ &  $\_$ &  $\_$ & $\_$ & 319,549 &$\_$ & 87,006 & 14
\\
\midrule
\mn{u20c1}& 73,252 &  131,103 &  2 &  73,157 & 73,583 & 103,260 & 0 & 43,062 & 2
\\
\midrule
\mn{u20c50}& 3,130,417 & $\_$ & $\_$ & $\_$ &$\_$ & $\_$ & $\_$ & $\_$ & 159
\\
\bottomrule
\end{tabular*}
}}
\caption{Number of conflicts, \mt{pref\_init} facts ($\succ_\Sigma$), and \mt{pref} facts computed for $\succup$, $\succdown$ and $\succground$, for scenarios (a), (c) and (d) (which directly yields an acyclic relation).  Empty cells indicate that clingo overflows or reaches a 30 min time-out. 
We fail to compute priority relations on omitted datasets in all scenarios but (d). 
}\label{tab:results-prio-numbers}
\end{table}

\subsection{Experimental Setting}

We use the \mn{CQAPri} benchmark \cite{DBLP:phd/hal/Bourgaux16}, a synthetic benchmark adapted from LUBM$^\exists_{20}$ \cite{DBLP:conf/semweb/LutzSTW13} to evaluate inconsistency-tolerant query answering over DL-Lite KBs. We also consider its extension with two priority relations 
given by the {\sc orbits} benchmark \cite{DBLP:conf/kr/BienvenuB22} for the comparison with {\sc orbits}.  
In this case, we translate the oriented conflict graph and causes for potential answers provided in the benchmark into $\Pi_{\mi{conf}_\succ}$ and $\Pi_{\ans}$ (for each potential answer $\ans$). 
Experiments were run with 16Go of RAM in a cluster node running CentOS Linux 7.6.1810 (Core) with linux kernel 3.10.0, with processor 2x 16-core Skylake Intel Xeon Gold 6142 @ 2.6 GHz. Reported times are averaged over 5 runs. 

\mypar{Datasets and meta-database} We build  
programs $\Pi_\Dmc$ for the \mn{uXcY} datasets of the \mn{CQAPri} benchmark with $\mn{X}\in\{1, 5, 20\}$ and $\mn{Y}\in\{1, 5, 10, 20, 30, 50\}$. These datasets are such that $\mn{uXcY}\subseteq\mn{uXcY'}$ for $\mn{Y}\leq\mn{Y'}$ and $\mn{uXcY}\subseteq\mn{uX'cY}$ for $\mn{X}\leq\mn{X'}$, with \mn{X} and \mn{Y} related to the size and the proportion of facts involved in some conflicts respectively.  Their sizes range from 75K to 2M facts and their proportions of facts involved in some (binary) conflict from 3\% to 46\%. 
We ensure that for every fact $\alpha$, $\id(\alpha)$ is the same in all $\mn{uXcY}$, and we obtain each program $\Pi_\data$ as a subset of the one generated for the largest dataset $\mn{u20c50}$, so that the same meta-data is used across the \mn{uXcY}. 
For $\Pi_\data$, 
we randomly generate two facts per $\alpha\in\Dmc$: 
$\mt{date}(\id(\alpha), n)$ 
and $\mt{source}(\id(\alpha), k)$, where $n,k$ are integers between 0 and 1000. For each $k$, we also 
generate a fact $\mt{reliability}(k,m)$ with $m$ an integer between 0 and~3.  

\mypar{Ontology and queries} 
We use the DL-Lite ontology and CQs of the \mn{CQAPri} benchmark to generate $\Pi_C$ and $\Pi_Q$. 
For $\Pi_C$, we first build a denial constraint per concept or role disjointness axiom. To experiment with non-binary conflicts, we also add a denial constraint with 10 relational atoms. 
We then rewrite all these constraints \wrt the ontology using \mn{Rapid} \cite{CTS-CADE-11}. 
For the queries, we similarly rewrite each query into a set of CQs.

\begin{figure}
\definecolor{solving}{HTML}{696969}
\definecolor{grounding}{HTML}{BEBEBE}

\begin{tikzpicture}
\begin{axis}[ybar stacked, bar width=1mm, width=3.8cm, height=4cm, ymin=0, ymax=600]\addplot[grounding,fill=grounding] coordinates{
(2.354,  6.409) 
(8.516,  15.489) 
(14.301,  20.415) 
(28.272,  33.526) 
(45.524,  68.858) 
(81.344,  0) 
};
\addplot[solving,fill=solving] coordinates{
(2.354,  27.688) 
(8.516,  158.593) 
(14.301,  18.820) 
(28.272,  46.575) 
(45.524,  132.365) 
(81.344,  0) 
};
\path (200,500) node {(a) $\succup$};
\end{axis} 
\end{tikzpicture}
\begin{tikzpicture}
\begin{axis}[ybar stacked, bar width=1mm, width=3.8cm, height=4cm, ymin=0, ymax=600,ytick=\empty]\addplot[grounding,fill=grounding] coordinates{
(2.354,  4.779) 
(8.516,  11.028) 
(14.301,  14.266) 
(28.272,  24.460) 
(45.524,  52.1612) 
(81.344,  301.057) 
};
\path (200,500) node {(a) $\succdown$};
\addplot[solving,fill=solving] coordinates{
(2.354, 0.927) 
(8.516,  2.488) 
(14.301,  2.985) 
(28.272,  5.062) 
(45.524,  10.187) 
(81.344,  35.166) 
};
\end{axis} 
\end{tikzpicture}
\begin{tikzpicture}
\begin{axis}[ybar stacked, bar width=1mm, width=3.8cm, height=4cm, ymin=0, ymax=600, ytick=\empty]\addplot[grounding,fill=grounding] coordinates{
(2.354,  39.079) 
(8.516,  86.8382) 
(14.301,  114.640) 
(28.272,  196.7766) 
(45.524,  446.008) 
(81.344,  0) 
};
\addplot[solving,fill=solving] coordinates{
(2.354,  12.427) 
(8.516,  26.122) 
(14.301,  34.235) 
(28.272,  57.912) 
(45.524,  131.151) 
(81.344,  0) 
};
\path (200,500) node {(a) $\succground$};
\end{axis} 
\end{tikzpicture}

\begin{tikzpicture}
\begin{axis}[ybar stacked, bar width=1mm, width=3.8cm, height=4cm, ymin=0, ymax=900]\addplot[grounding,fill=grounding] coordinates{
(2.354,  7.306) 
(8.516,  16.027) 
(14.301, 22.418) 
(28.272,  40.594) 
(45.524,  79.800) 
(81.344,  464.632) 
};
\addplot[solving,fill=solving] coordinates{
(2.354,  45.549) 
(8.516,  200.382) 
(14.301,  15.076) 
(28.272,  31.807) 
(45.524, 69.602) 
(81.344,  373.049) 
};
\path (200,750) node {(b) $\succup$};
\end{axis} 
\end{tikzpicture}
\begin{tikzpicture}
\begin{axis}[ybar stacked, bar width=1mm, width=3.8cm, height=4cm, ymin=0, ymax=900,ytick=\empty]\addplot[grounding,fill=grounding] coordinates{
(2.354,  5.185) 
(8.516,  10.783) 
(14.301,  15.216) 
(28.272,  29.948) 
(45.524,  57.357) 
(81.344,  294.292) 
};
\path (200,750) node {(b) $\succdown$};
\addplot[solving,fill=solving] coordinates{
(2.354, 0.960) 
(8.516,  2.315) 
(14.301,  3.560) 
(28.272,  6.979) 
(45.524,  12.629) 
(81.344,  58.855) 
};
\end{axis} 
\end{tikzpicture}
\begin{tikzpicture}
\begin{axis}[ybar stacked, bar width=1mm, width=3.8cm, height=4cm, ymin=0, ymax=900, ytick=\empty]\addplot[grounding,fill=grounding] coordinates{
(2.354,  26.115) 
(8.516,  52.106) 
(14.301,  79.815) 
(28.272,  146.006) 
(45.524,  272.404) 
(81.344,  0) 
};
\addplot[solving,fill=solving] coordinates{
(2.354,  7.270) 
(8.516,  14.572) 
(14.301,  21.787) 
(28.272,  41.696) 
(45.524,  77.645) 
(81.344,  0) 
};
\path (200,750) node {(b) $\succground$};
\end{axis} 
\end{tikzpicture}

\caption{Time (in sec.) to compute $\succ^x$ from the pre-computed conflicts for $\mn{u1cY}$ given as a program $\Pi_{\mi{conf}}$ and $\Pi_\Dmc\cup\Pi_\data\cup\Pi_P\cup\Pi_{\succ^x}$ in scenarios (a) and (b) \wrt the number (in thousands) of conflicts. Empty bars for $\mn{u1c50}$ (81K conflicts) mean t.o. or oom. The lower part of each bar (light grey) shows the time to ground the ASP program while the upper part is the time to solve it. 
}\label{fig:results-prio-times}
\end{figure}

\mypar{Preference rules} 
We use the following preferences rules, and test four scenarios: (a) $\Sigma^{a}_1=\{\rho_3,\rho_4\}$, $\Sigma^{a}_2=\{\rho_2\}$, $\Sigma^{a}_3=\{\rho_1\}$; (b) $\Sigma^{b}_1=\{\rho_1,\rho_3,\rho_4\}$, $\Sigma^{b}_2=\{\rho_2\}$; (c) $\Sigma^{c}_1=\{\rho_1,\rho_2, \rho_3,\rho_4\}$; (d) $\Sigma^{d}_1=\{\rho_3,\rho_4\}$ (dropping $\rho_1,\rho_2$).
\begin{align*}
\rho_1&:\mn{date}(x_1,y_1)\wedge\mn{date}(x_2,y_2)\wedge y_1>y_2\rightarrow \pref{x_1}{x_2}\\
\rho_2&:\mn{source}(x_1,y_1)\wedge\mn{source}(x_2,y_2)\wedge \mn{reliability}(y_1,z_1) \\&\wedge\mn{reliability}(y_2,z_2)\wedge z_1>z_2\rightarrow \pref{x_1}{x_2}\\
\rho_3&:x_1=\id(\mn{FPr}(y)) \wedge x_2=\id(\mn{APr}(y)) \rightarrow \pref{x_1}{x_2}\\
\rho_4&:x_1=\id(\mn{APr}(y)) \wedge x_2=\id(\mn{GrSt}(y)) \rightarrow \pref{x_1}{x_2}
\end{align*}

\subsection{Experimental Results}
Table~\ref{tab:results-prio-numbers} and Figure~\ref{fig:results-prio-times} present some results of the evaluation of the priority relation computation. We were not able to compute $\succrefup$ even on \mn{u1c1} because it overflows the number of atoms \mn{clingo} can handle. However, we managed to compute the other priority relations for almost all small datasets ($>$75K), several medium size datasets ($>$463K), and one large dataset ($>$1,983K) even in cases with a large proportion of facts in conflicts (44\% for \mn{u1c50}) or high numbers of  \mt{pref\_init} facts (319K for \mn{u5c20} in scenario (c)). 
All datasets have exactly 40 conflicts of size 10, which yields 1,718 pairs of facts, and other conflicts are binary (so that \eg \mn{u1c1} has 4,032 pairs of conflicting facts). Several preference statements ($\succ_\Sigma$) can be made on each such pair (in both directions and on different priority levels) while the priority relation ($\succ^x$) compares each pair of facts at most once so that, \eg $\succground$ compares 92\% of the pairs of conflicting facts of \mn{u1c1} in scenario (a). 
Interestingly, $\succdown$ and $\succground$ often coincide and never differ by more than 5\% of \mt{pref} facts on instances for which we computed them, while $\succup$ is often reduced to the empty relation.  
From a computational point of view, 
$\succdown$ is significantly faster to compute than $\succground$ and $\succup$ (except in scenario (d) which yields a very small and acyclic $\succ_\Sigma$). Hence $\succdown$ may be a good method in practice.

The times given in Figure~\ref{fig:results-prio-times} do not include the time needed to compute the conflicts, which may be far from negligible: while the evaluation of $\Pi_\Dmc\cup\Pi_C$ never takes more than about 1min ($\mn{u20c50}$), the time needed to minimize the conflicts takes from less than 1sec to 206sec for the \mn{u1cY} cases, but more than 45 hours for \mn{u20c50}! In the case where conflicts have size at most two, however, this takes at most 1.5sec for the \mn{u1cY} cases and no more than 42sec (\mn{u20c50}).

Regarding the computation of optimal repair-based semantics, 
we select 8 queries with a reasonable number of potential answers (between 3 and 16,969) because 
very high numbers of answers lead to time-out (30 minutes per query). 
Our system is in general by far slower than {\sc orbits} on datasets that are large or with a high percentage of conflicting facts: \eg on \mn{u20c1} and \mn{u1c50}, our implementation always takes more than 16 times longer and up to more than 1,200 times longer to filter answers under $P$-AR or $P$-brave semantics. On the simplest case, \mn{u1c1}, the difference is less striking (at least if we include the time to load the input in the computation time for {\sc orbits}), but still of orders of magnitude for many queries. 
However, it is notable that we do manage to answer a few queries under $C$-AR and $C$-brave semantics in cases where {\sc orbits} fails.


\section{Related Work}\label{sec:related-work}
We draw inspiration from different preference specification formalisms 
defined for 
related settings, 
such as preference-based query answering over databases \cite{DBLP:journals/tods/StefanidisKP11} 
or 
(consistent) KBs \cite{DBLP:conf/ijcai/LukasiewiczMS13}. 
In the latter work, for example, preference formulas 
consist of a condition, given by an FO-formula, which induces a preference between two atoms. 
In the context of controlled query evaluation in DL, 
\citeauthor{DBLP:conf/semweb/CimaLMRS21}~\shortcite{DBLP:conf/semweb/CimaLMRS21} define a preference relation among ontology predicates, which straightforwardly induces one among the facts. 
Closer to our own work, 
 \citeauthor{DBLP:journals/ai/CalauttiGMT22}~\shortcite{DBLP:journals/ai/CalauttiGMT22} 
consider preference rules that generate preferences between atoms in order to select preferred repairs of inconsistent KBs. 
Differently from us, their preference rules are evaluated over the repairs themselves,
whereas our rules are evaluated over the dataset (and meta-database) and yield a priority relation between facts, 
which is then lifted to get optimal repairs. 

Our preference rules generalize the preceding preference formalisms 
by allowing rule bodies that 
express more complex conditions, e.g., that may refer to meta-data, include negated atoms, or quantify over ontology predicates. 
In this manner, we obtain an easy and flexible way of defining inconsistency management policies, 
as considered in \cite{DBLP:journals/ijar/MartinezPPSS14}.  
Moreover, a distinguishing contribution of our work is that we propose methods for dealing with 
cycles among the induced preference statements. Our cycle resolution techniques can take into account 
priorities amongst the preference rules themselves. 
Rules with priorities are also considered in prioritized logic programming \cite{DBLP:journals/ai/SakamaI00,DBLP:journals/ai/BrewkaE99},
but there serve the purpose of identifying preferred answer sets. 

Our work has high-level similarities with \cite{DBLP:journals/tods/FaginKRV16}, 
which employs optimal repairs from \cite{DBLP:journals/amai/StaworkoCM12} to clean
inconsistencies arising amongst facts extracted using document spanners. 
They introduce priority-generating dependencies to define a priority relation  
and explore some properties of the induced relations. 
However, the formalization and techniques differ
significantly due to the very different settings.

Another line of related work  
uses logic programming for consistent query answering over relational databases
\cite{DBLP:journals/tkde/GrecoGZ03,DBLP:journals/tods/EiterFGL08,DBLP:journals/tplp/MannaRT13}. 
These works consider different kinds of repair: 
on the one hand, they allow repairs to restore consistency by adding facts, while we focus on subset repairs, only involving deletions, which are standard for KBs interpreted under the open-world assumption; on the other hand, we consider priority-based optimal repairs. 
\citeauthor{DBLP:journals/tkde/GrecoGZ03}~\shortcite{DBLP:journals/tkde/GrecoGZ03}, however, define constraints that express conditions on the insertion or deletion of atoms, and rules defining priorities among such updates, sharing the intuition that the user should be able to specify preferences on how to treat inconsistency. 
On the implementation side, we remark that compared to our experimental setting, the evaluations of previous ASP approaches typically either use databases with very few conflicts (few hundreds), or whose conflicts have a specific structure that ensures that the conflicts form small independent connected components.


\section{Conclusion and Future Work}\label{conclusion}

In this paper, we present a rule-based approach to specifying a priority relation between conflicting facts, in order to adopt optimal repair-based inconsistency-tolerant semantics. We investigate the problem of deciding whether the relation induced by a set of preference rules is guaranteed to be acyclic and propose several strategies to remove cycles. 
We also present an implementation of the approach, including the computation of query answers that hold under a given semantics, which was not yet implemented for the case of non-binary conflicts and optimal repairs. 
While our comparison show that existing SAT implementation is more efficient for the latter task (though the SAT implementation is optimized for binary conflicts while ours handles conflicts of any size so the comparison is not entirely fair), ASP retains a number of advantages. Besides allowing the user to directly and easily express preference rules, logic programs are 
easy to modify to treat other problems (such as the computation of repairs, which is not tackled by  {\sc orbits}). Moreover, ASP is more expressive than SAT, so that it is theoretically possible to employ ASP to compute answers under globally-optimal-repair-based semantics (which have $\sigmaptwo$ / $\piptwo$ complexity), even if finding an efficient encoding remains a challenge. 

There are several directions for future work. First, we could  
extend the static analysis of 
Section~\ref{sec:analysis}, by considering more classes of logical theories and preference rules. Besides the problem of deciding whether a theory and set of preference rules ensure that the induced relation is acyclic, one could wonder whether they guarantee that there exists a unique optimal repair. 
On the practical side, we want to implement the last missing blocks to have a truly end-to-end system for query answering over inconsistent KBs with preference rules (in particular to generate the input logic programs of Table~\ref{tab:program-input} from 
data/theory given in various formats).

\section*{Acknowledgments}
This work was supported by the ANR AI Chair INTENDED (ANR-19-CHIA-0014) and JST CREST Grant Number JPMJCR22D3, Japan.

\bibliographystyle{kr}
\bibliography{cqa-priority} 

@inproceedings{DBLP:conf/iclp/EglyGW08,
  author       = {Uwe Egly and
                  Sarah Alice Gaggl and
                  Stefan Woltran},
  title        = {{ASPARTIX:} Implementing Argumentation Frameworks Using Answer-Set
                  Programming},
  booktitle    = {Proceedings of {ICLP}},
  year         = {2008}
}

@article{Baroni&al/grounded,
author = {Baroni, Pietro and Caminada, Martin and Giacomin, Massimiliano},
year = {2011},
pages = {365-410},
title = {An introduction to argumentation semantics},
volume = {26},
journal = {Knowledge Eng. Review},
doi = {10.1017/S0269888911000166}
}

@article{DBLP:journals/tkde/GrecoGZ03,
  author       = {Gianluigi Greco and
                  Sergio Greco and
                  Ester Zumpano},
  title        = {A Logical Framework for Querying and Repairing Inconsistent Databases},
  journal      = {{IEEE} Trans. Knowl. Data Eng.},
  volume       = {15},
  number       = {6},
  pages        = {1389--1408},
  year         = {2003}
}

@article{DBLP:journals/ai/BrewkaE99,
  author       = {Gerhard Brewka and
                  Thomas Eiter},
  title        = {Preferred Answer Sets for Extended Logic Programs},
  journal      = {Artif. Intell.},
  volume       = {109},
  number       = {1-2},
  pages        = {297--356},
  year         = {1999}
}

@book{DBLP:books/sp/Lifschitz19,
  author       = {Vladimir Lifschitz},
  title        = {Answer Set Programming},
  publisher    = {Springer},
  year         = {2019},
  url          = {https://doi.org/10.1007/978-3-030-24658-7},
  doi          = {10.1007/978-3-030-24658-7},
  isbn         = {978-3-030-24657-0},
  bibsource    = {dblp computer science bibliography, https://dblp.org}
}

@book{DBLP:series/synthesis/2012Gebser,
  author       = {Martin Gebser and
                  Roland Kaminski and
                  Benjamin Kaufmann and
                  Torsten Schaub},
  title        = {Answer Set Solving in Practice},
  series       = {Synthesis Lectures on Artificial Intelligence and Machine Learning},
  publisher    = {Morgan {\&} Claypool Publishers},
  year         = {2012},
  url          = {https://doi.org/10.2200/S00457ED1V01Y201211AIM019},
  doi          = {10.2200/S00457ED1V01Y201211AIM019},
  isbn         = {978-3-031-00433-9},
  bibsource    = {dblp computer science bibliography, https://dblp.org}
}

@article{DBLP:journals/ki/Bienvenu20,
  author    = {Meghyn Bienvenu},
  title     = {A Short Survey on Inconsistency Handling in Ontology-Mediated Query
               Answering},
  journal   = {K{\"{u}}nstliche Intelligenz},
  volume    = {34},
  number    = {4},
  pages     = {443--451},
  year      = {2020}
}

@inproceedings{Bienvenu_TractableApproximation_long,
  author =	 {Meghyn Bienvenu and Riccardo Rosati},
  title =	 {Tractable Approximations of Consistent Query
                  Answering for Robust Ontology-based Data Access},
  booktitle =	 {Proceedings of {IJCAI}},
  year =	 2013
}

@article{DBLP:journals/tods/FaginKRV16,
  author    = {Ronald Fagin and
               Benny Kimelfeld and
               Frederick Reiss and
               Stijn Vansummeren},
  title     = {Declarative Cleaning of Inconsistencies in Information Extraction},
  journal   = {{ACM} Trans. Database Syst.},
  volume    = {41},
  number    = {1},
  pages     = {6:1--6:44},
  year      = {2016}
}

@inproceedings{DBLP:conf/pods/Bertossi19,
  author    = {Leopoldo E. Bertossi},
  title     = {Database Repairs and Consistent Query Answering: Origins and Further
               Developments},
  booktitle = {Proceedings of {PODS}},
  year      = {2019}
}

@inproceedings{DBLP:conf/kr/BienvenuB20,
  author    = {Meghyn Bienvenu and
               Camille Bourgaux},
  title     = {Querying and Repairing Inconsistent Prioritized Knowledge Bases: Complexity
               Analysis and Links with Abstract Argumentation},
  booktitle = {Proceedings of {KR}},
  year      = {2020}
}

@inproceedings{DBLP:conf/kr/BienvenuB22,
  author    = {Meghyn Bienvenu and
               Camille Bourgaux},
  title     = {Querying Inconsistent Prioritized Data with {ORBITS:} Algorithms,
               Implementation, and Experiments},
  booktitle = {Proceedings of {KR}},
  year      = {2022}
}

@inproceedings{DBLP:conf/kr/BienvenuB23,
  author       = {Meghyn Bienvenu and
                  Camille Bourgaux},
  title        = {Inconsistency Handling in Prioritized Databases with Universal Constraints:
                  Complexity Analysis and Links with Active Integrity Constraints},
  booktitle    = {Proceedings of {KR}},
  year         = {2023}
}

@inproceedings{DBLP:conf/dlog/Bourgaux24,
  author       = {Camille Bourgaux},
  title        = {Querying Inconsistent Prioritized Data (Abstract of Invited Talk)},
  booktitle    = {Proceedings of {DL}},
  year         = {2024}
}

@article{DBLP:journals/amai/StaworkoCM12,
  author    = {Slawek Staworko and
               Jan Chomicki and
               Jerzy Marcinkowski},
  title     = {Prioritized repairing and consistent query answering in relational
               databases},
  journal   = {Annals of Mathematics and Artificial Intelligence ({AMAI})},
  volume    = {64},
  number    = {2-3},
  pages     = {209--246},
  year      = {2012}
}

@inproceedings{DBLP:conf/icdt/KimelfeldLP17,
  author    = {Benny Kimelfeld and
               Ester Livshits and
               Liat Peterfreund},
  title     = {Detecting Ambiguity in Prioritized Database Repairing},
  booktitle = {Proceedings of {ICDT}},
  year      = {2017}
}

@inproceedings{DBLP:conf/pods/LivshitsK17,
  author    = {Ester Livshits and
               Benny Kimelfeld},
  title     = {Counting and Enumerating (Preferred) Database Repairs},
  booktitle = {Proceedings of {PODS}},
  year      = {2017}
}

@article{DBLP:journals/tcs/KimelfeldLP20,
  author       = {Benny Kimelfeld and
                  Ester Livshits and
                  Liat Peterfreund},
  title        = {Counting and enumerating preferred database repairs},
  journal      = {Theor. Comput. Sci.},
  volume       = {837},
  pages        = {115--157},
  year         = {2020}
}

@inproceedings{DBLP:conf/aaai/BienvenuBG14,
  author    = {Meghyn Bienvenu and
               Camille Bourgaux and
               Fran{\c{c}}ois Goasdou{\'{e}}},
  title     = {Querying Inconsistent Description Logic Knowledge Bases under Preferred
               Repair Semantics},
  booktitle = {Proceedings of {AAAI}},
  year      = {2014}
}

@phdthesis{DBLP:phd/hal/Bourgaux16,
  author    = {Camille Bourgaux},
  title     = {Inconsistency Handling in Ontology-Mediated Query Answering. (Gestion
               des incoh{\'{e}}rences pour l'acc{\`{e}}s aux donn{\'{e}}es
               en pr{\'{e}}sence d'ontologies)},
  school    = {University of Paris-Saclay, France},
  year      = {2016}
}

@inproceedings{DBLP:conf/ijcai/BenferhatBT15,
  author       = {Salem Benferhat and
                  Zied Bouraoui and
                  Karim Tabia},
  title        = {How to Select One Preferred Assertional-Based Repair from Inconsistent
                  and Prioritized {DL-Lite} Knowledge Bases?},
  booktitle    = {Proceedings of {IJCAI}},
  year         = {2015}
}

@inproceedings{DBLP:conf/ijcai/Rosati11,
  author    = {Riccardo Rosati},
  title     = {On the Complexity of Dealing with Inconsistency in Description Logic
               Ontologies},
  booktitle = {Proceedings of {IJCAI}},
  year      = {2011}
}

@article{calvaneseetal:dllite,
author       = {Diego Calvanese and
                  Giuseppe De Giacomo and
                  Domenico Lembo and
                  Maurizio Lenzerini and
                  Riccardo Rosati},
  title        = {Tractable Reasoning and Efficient Query Answering in Description Logics:
                  The {DL-Lite} Family},
  journal      = {J. Autom. Reason.},
  volume       = {39},
  number       = {3},
  pages        = {385--429},
  year         = {2007}
}

@inproceedings{DBLP:conf/icdt/LopatenkoB07,
  author       = {Andrei Lopatenko and
                  Leopoldo E. Bertossi},
  title        = {Complexity of Consistent Query Answering in Databases Under Cardinality-Based
                  and Incremental Repair Semantics},
  booktitle    = {Proceedings of {ICDT}},
  year         = {2007}
}

@article{DBLP:journals/kais/DuQS13,
  author       = {Jianfeng Du and
                  Guilin Qi and
                  Yi{-}Dong Shen},
  title        = {Weight-based consistent query answering over inconsistent {SHIQ} knowledge bases},
  journal      = {Knowl. Inf. Syst.},
  volume       = {34},
  number       = {2},
  pages        = {335--371},
  year         = {2013}
}

@article{DBLP:journals/ai/CalauttiGMT22,
  author    = {Marco Calautti and
               Sergio Greco and
               Cristian Molinaro and
               Irina Trubitsyna},
  title     = {Preference-based inconsistency-tolerant query answering under existential
               rules},
  journal   = {Artif. Intell.},
  volume    = {312},
  pages     = {103772},
  year      = {2022}
}

@inproceedings{DBLP:conf/kr/LukasiewiczMM23,
  author       = {Thomas Lukasiewicz and
                  Enrico Malizia and
                  Cristian Molinaro},
  title        = {Complexity of Inconsistency-Tolerant Query Answering in Datalog+/-
                  under Preferred Repairs},
  booktitle    = {Proceedings of {KR}},
  year         = {2023}
}

@article{DBLP:journals/tplp/MannaRT13,
  author       = {Marco Manna and
                  Francesco Ricca and
                  Giorgio Terracina},
  title        = {Consistent query answering via {ASP} from different perspectives:
                  Theory and practice},
  journal      = {Theory Pract. Log. Program.},
  volume       = {13},
  number       = {2},
  pages        = {227--252},
  year         = {2013}
}

@article{DBLP:journals/tods/EiterFGL08,
  author       = {Thomas Eiter and
                  Michael Fink and
                  Gianluigi Greco and
                  Domenico Lembo},
  title        = {Repair localization for query answering from inconsistent databases},
  journal      = {{ACM} Trans. Database Syst.},
  volume       = {33},
  number       = {2},
  pages        = {10:1--10:51},
  year         = {2008}
}

@article{DBLP:journals/ijar/MartinezPPSS14,
  author       = {Maria Vanina Martinez and
                  Francesco Parisi and
                  Andrea Pugliese and
                  Gerardo I. Simari and
                  V. S. Subrahmanian},
  title        = {Policy-based inconsistency management in relational databases},
  journal      = {Int. J. Approx. Reason.},
  volume       = {55},
  number       = {2},
  pages        = {501--528},
  year         = {2014}
}

@article{DBLP:journals/tods/StefanidisKP11,
  author       = {Kostas Stefanidis and
                  Georgia Koutrika and
                  Evaggelia Pitoura},
  title        = {A survey on representation, composition and application of preferences in database systems},
  journal      = {{ACM} Trans. Database Syst.},
  volume       = {36},
  number       = {3},
  pages        = {19:1--19:45},
  year         = {2011}
}

@inproceedings{DBLP:conf/ijcai/LukasiewiczMS13,
  author       = {Thomas Lukasiewicz and
                  Maria Vanina Martinez and
                  Gerardo Ignacio Simari},
  title        = {Preference-Based Query Answering in Datalog+/- Ontologies},
  booktitle    = {Proceedings of {IJCAI}},
  year         = {2013}
}

@inproceedings{DBLP:conf/semweb/CimaLMRS21,
  author       = {Gianluca Cima and
                  Domenico Lembo and
                  Lorenzo Marconi and
                  Riccardo Rosati and
                  Domenico Fabio Savo},
  title        = {Controlled Query Evaluation over Prioritized Ontologies with Expressive Data Protection Policies},
  booktitle    = {Proceedings of {ISWC}},
  year         = {2021}
}

@article{DBLP:journals/ai/SakamaI00,
  author       = {Chiaki Sakama and
                  Katsumi Inoue},
  title        = {Prioritized logic programming and its application to commonsense reasoning},
  journal      = {Artif. Intell.},
  volume       = {123},
  number       = {1-2},
  pages        = {185--222},
  year         = {2000}
}

@inproceedings{DBLP:conf/semweb/LutzSTW13,
  author    = {Carsten Lutz and
               Inan{\c{c}} Seylan and
               David Toman and
               Frank Wolter},
  title     = {The Combined Approach to {OBDA:} {Taming} Role Hierarchies Using Filters},
  booktitle = {Proceedings of {ISWC}},
  year      = {2013}
}

@inproceedings{CTS-CADE-11,
  author = {A. Chortaras and D. Trivela and G. Stamou},
  title = {Optimized Query Rewriting for {OWL 2 QL}},
  year = {2011},
  booktitle = {Proceedings of {CADE}}
}

@article{DBLP:journals/aicom/GebserKKOSS11,
  author       = {Martin Gebser and
                  Benjamin Kaufmann and
                  Roland Kaminski and
                  Max Ostrowski and
                  Torsten Schaub and
                  Marius Schneider},
  title        = {Potassco: The Potsdam Answer Set Solving Collection},
  journal      = {{AI} Commun.},
  volume       = {24},
  number       = {2},
  pages        = {107--124},
  year         = {2011}
}
\onecolumn
\appendix

\section{Appendix: Omitted Proofs}

\undecthm*
\begin{proof}
Let $\Tmc$ be a non-trivial theory, and let $\Cmc$ be a conflict for $\Tmc$. 
We reduce from the problem of deciding whether a FO-sentence is finitely satisfiable. 
Suppose that we are given a FO-sentence $\Phi$. We may assume w.l.o.g.\ 
that $\Phi$ only uses predicates from $\metapreds$. Then for every pair of predicates 
$P \in \preds_k$ and $P' \in \preds_\ell$ that occur in $\Tmc$, we create the following 
preference rule:
$$\Phi \wedge x_1=\id(P(\vec{y})) \wedge x_2=\id(P'(\vec{z})) \rightarrow \pref{x_1}{x_2}$$
where $\vec{y}$ and $\vec{z}$ are respectively a $k$-tuple and $\ell$-tuple of distinct variables. 
Let $\Sigma$ be the finite set consisting of all and only these preference rules. Note that
$\Sigma$ belongs to $\preflangdl$ as the syntax allows for 
arbitrary FO-sentences can be expressed 
in the rule bodies of $\preflangdl$. We claim that $\Sigma$ is $\Tmc$-acyclic iff 
$\Phi$ is not finitely satisfiable.

First suppose that $\Phi$ is finitely satisfiable, which means that there exists 
a finite interpretation that makes $\Phi$ true. 
We consider the KB $\Kmc = (\Cmc, \Tmc)$, 
use the satisfying interpretation for $\Phi$ as the set of facts $\Fmc$ (treating domain elements as constants),
and let $\id$ be any function that assigns identifiers to the facts in $\Cmc$, to get the meta-database $\meta$. 
By construction, for any pair of (possibly equal) facts $\alpha, \beta \in \Cmc$ 
there is a rule in $\Sigma$ (the one that mentions the predicates of $\alpha$ and $\beta$) 
that induces $\pref{\id(\alpha)}{\id(\beta)}$. This is because
$\Phi$ will evaluate to true over $\meta$ and the equality atoms
allow us to assign $x_1$ to $\id(\alpha)$ and $x_2$ to $\id(\beta)$. 
In the same manner we get $\pref{\id(\beta)}{\id(\alpha)}$. 
As $\alpha$ and $\beta$ are contained in the same conflict, 
we have $\alpha \fullsucc \beta$ and $\beta \fullsucc \alpha$, 
which means
$\Sigma$ is not $\Tmc$-acyclic. 

For the other direction, suppose that $\Phi$ is not finitely satisfiable. 
It follows that no matter which KB $\Kmc$ and meta-database $\meta$ we consider, 
the formula $\Phi$, which occurs in all rule bodies, will evaluate to false. 
Thus, $\Sigma(\Kmc,\meta)$ will always be empty, and no cycle can be generated. 
This implies in particular that $\Sigma$ is $\Tmc$-acyclic. 
\end{proof}

\tacyclicthm*
\begin{proof}
Let $\Tmc$ and $\Sigma$ be as stated. For simplicity, we assume that $\Tmc$
does not contain constants, but we mention at the end of the proof how we can easily adapt the 
argument to accommodate constants in the constraints. 

We aim to place a bound on the 
length of a minimal cycle, if one exists, from which decidability follows. 
To this end, let us suppose that 
$\Sigma$ is cyclic relative to $\Tmc$. Then there exists 
a knowledge base of the form $\Kmc=(\Dmc,\Tmc)$ and accompanying meta-database $\meta=(\id,\data)$
such that the induced relation $\succ_\Sigma$ contains a cycle $\alpha_1 \succ_\Sigma \ldots \succ_\Sigma \alpha_n \succ_\Sigma \alpha_1$
(we assume w.l.o.g.\ that this cycle is minimal, i.e.\ there do not exist $1 \leq i < j <n$  such that $\alpha_i = \alpha_j$). 
It follows that we can find a sequence of rules $\sigma_1, \ldots, \sigma_n \in \Sigma$ and variable substitutions $\nu_1, \ldots, \nu_n$ for $\rvars(\sigma_1), \ldots, \rvars(\sigma_n)$  
such that for every $1 \leq i \leq n$:
\begin{itemize}
\item $\nu_i(x_1)= \id(\alpha_i)$ and $\nu_i(x_2)=\id(\alpha_{i+1})$ 
\item $\Cmc_i = \{\alpha_i, \alpha_{i+1}\} \in \conflicts{\Kmc}$
\item $\nu_i(\beta) \in \Dmc \cup \data$ for every relational atom $\beta \in B_i$, which means in particular that 
if $\beta$ contains a variable $u$ with $\nu_i(u) = \id(\alpha) \in \idinds$, then $\alpha \in \Dmc$ 
\item $\nu_i(u)= \id(P(\nu(\vec{z})))$ for every atom $u = \id(P(\vec{z})) \in B_i$
\item $\nu_i(u) \neq \nu_i(v)$ if $u \neq v \in B_i$
\end{itemize}
where $x_1,x_2$ are the distinguished head variables (used in all rules), $B_i$ denotes the body of rule $\sigma_i$,
and $\alpha_{n+1} = \alpha_1$. Note that we know that $\{\alpha_i, \alpha_{i+1}\} \in \conflicts{\Kmc}$
due to the fact that $\{\alpha_i, \alpha_{i+1}\}$ must be contained in some conflict, and 
 $\Tmc$ only admits conflicts of size at most two.

Now let us suppose that there exist positions $1 \leq i+1 < j \leq n$ such that:
\begin{itemize}
\item[($\star$)] there is an isomorphism $\iso$ from $\alpha_{j}$ to $\alpha_{i+1}$ that is the identity on $\individuals{\alpha_1}$,
\end{itemize}
 i.e. for every  $c \in \individuals{\alpha_1}$, the $k$th argument of $\alpha_j$ is equal to $c$ iff the $k$th argument of $\alpha_{i+1}$ equals $c$.
We show how to create a strictly shorter cycle, intuitively by jumping straight from $\alpha_{i}$ to $\alpha_j$ (modulo some renaming of constants). 
Formally, we do this by renaming some constants in $\alpha_{j}, \ldots, \alpha_n$, then 
afterwards updating the variable assignments, database, and meta-database to reflect the renamed constants. 

We shall  begin by considering the first pair $\alpha_j \succ_\Sigma \alpha_{j+1}$, for which we set:
\begin{itemize}
\item $\alpha_{j}'= \iso(\alpha_{j})= \alpha_{i+1}$,
\item $\alpha'_{j+1}= \rho_{j+1}(\alpha_{j+1})$, where $\rho_{j+1}$ extends $\iso$ to the constants in 
$\individuals{\alpha_{j+1}} \setminus (\individuals{\alpha_1} \cup \individuals{\alpha_{j}})$
by mapping every such constant $c$ to a corresponding fresh constant $c'$. 
\end{itemize}
By construction, $\rho_{j+1}$ defines an isomorphism from $\alpha_{j+1}$ to $\alpha'_{j+1}$ that is the identity on $\individuals{\alpha_1}$. 
We may proceed in the same manner to define $\alpha'_{j+2}, \ldots, \alpha'_{n+1}$. 
Indeed, supposing we have already defined $\alpha'_j,\alpha'_{j+1}, \ldots, \alpha'_k$, via the isomorphisms $\mu,\rho_{j+1}, \ldots, \rho_k$, 
we can define $\alpha'_{k+1}$ as follows:
\begin{itemize}
\item set $\alpha'_{k+1}= \rho_{k+1}(\alpha_{k+1})$, where $\rho_{k+1}$ is obtained by first restricting 
$\rho_k$ to the constants in $\alpha_k$ and $\alpha_1$, then extending it to the constants in 
$\individuals{\alpha_{k+1}} \setminus (\individuals{\alpha_1} \cup \individuals{\alpha_{k}})$
by mapping every such constant $c$ to a corresponding fresh constant $c'$.
\end{itemize}
By induction, we have that 
$\rho_{k+1}$ defines an isomorphism from $\alpha_{k+1}$ to $\alpha'_{k+1}$ that is the identity on $\individuals{\alpha_1}$
and agrees with $\rho_k$ on $\individuals{\alpha_k}$. Observe that $\alpha_{n+1}=\alpha_1$ only contains constants from $\alpha_1$,
so $\alpha'_{n+1}=\alpha_{n+1}=\alpha_1$.

We now need to update the database and meta-database in order to show that 
the rules $\sigma_1, \ldots, \sigma_n$ allow us to obtain the shortened cycle 
\begin{equation}
\alpha_1 \succ_\Sigma \ldots \succ_\Sigma  \alpha_{i} \succ_\Sigma \alpha'_{j} \succ_\Sigma \ldots \succ_\Sigma \alpha'_n \succ_\Sigma \alpha_1
\label{shortcycle}
\end{equation}
For the updated meta-database, we shall need an updated function to handle facts with the freshly introduced constants,
so we let $\id'$ be an extension of $\id$ that assigns ids to all 
facts built from predicates in $\signature{\Kmc}$ and constants in $\Dmc \cup \{\alpha'_k \mid j < k \leq n\}$ (of course, typically only a subset of these facts 
will actually occur in the updated database). 
For every $j \leq k \leq n$, let $\rho^*_k(c) = \rho_k(c)$ if $\rho_k(c)$ is defined, else $\rho^*_k(c)=c$,
and consider the updated variable assignments $\nu'_j, \ldots, \nu'_n$ defined as follows:
\begin{itemize}
\item if $\nu_k(v) \in \inds \setminus \idinds$, then $\nu'_k(v)=\rho^*_k(\nu_k(v))$ 
\item if $\nu_k(v) \in \idinds$ with $\nu_k(v)=\id(\alpha_v)$, then 
set $\nu'_k(v)=\id'(\rho^*_k(\alpha_v))$
\end{itemize}
Observe that there is always a unique $\alpha_v$ such that $\nu_k(v)=\id(\alpha_v)$,
so $\nu'_k$ is well defined. Moreover, by construction, for every $j \leq k \leq n$, we have 
\begin{itemize}
\item $\nu'_k(x_1)= \id'(\alpha'_k)$ and $\nu'_k(x_2)=\id'(\alpha'_{k+1})$ 
\end{itemize}
The updated database $\Dmc'$ will contain all and only the following facts: 
\begin{itemize}
\item for every $1 \leq k \leq i$ and every relational atom $\beta \in B_k$ with predicate from $\signature{\Kmc}$, the fact $\nu_k(\beta)$
\item for every $j \leq k \leq n$ and every relational atom $\beta \in B_k$ with predicate from $\signature{\Kmc}$, the fact $\nu'_k(\beta)$
\item for every variable $v$ such that $\nu'_k(v)=\id'(\alpha_v)$, the fact~$\alpha_v$
\end{itemize}
The last item implies in particular that all facts $\alpha_1, \ldots, \alpha_{i}, \alpha'_{j}, \ldots, \alpha'_n$ occur in $\Dmc'$. 
Finally, we define the updated meta-database $\meta'=(\id',\data')$
by letting $\data'$ consist of the following facts: 
\begin{itemize}
\item for every $1 \leq k \leq i$ and every relational atom $\beta \in B_k$ with predicate from $\signature{\Fmc}$, the fact $\nu_k(\beta)$
\item for every $j \leq k \leq n$ and every relational atom $\beta \in B_k$ with predicate from $\signature{\Fmc}$, the fact $\nu'_k(\beta)$
\end{itemize}
Note that by construction the only constants from $\idinds$ that occur on $\data'$ 
have the form $\id'(\alpha)$ for some $\alpha \in \Dmc'$, so $\meta'$ is well defined. 

It remains to verify that $\Kmc'=(\Dmc',\Tmc)$ and $\meta'=(\id',\data')$
indeed yield the cycle from (1). First we note that for $1 \leq k \leq i$, 
$\pref{\id'(\alpha_k)}{\id'(\alpha_{k+1})}$ is induced by $\sigma_k$, as is witnessed by 
the original variable substitution $\nu_k$. This is simply because we have not 
modified facts $\alpha_k$ for $k\leq i$, and the required facts remain present 
in $\Dmc'$ and $\data'$. If instead we consider $i < k \leq n$, 
we can similarly show that $\pref{\id'(\alpha_k)}{\id'(\alpha_{k+1})}$
is induced by $\sigma_k$. Indeed, by construction, 
$\nu'_k(B_k)$ evaluates to true w.r.t.\ $\Kmc', \meta'$,
and we already know that 
$\nu'_k(x_1)= \id'(\alpha'_k)$ and $\nu_k(x_2)=\id'(\alpha'_{k+1})$. 
Finally, we note that $\alpha'_k$ and $\alpha'_{k+1}$
are isomorphic (w.r.t.\ $\individuals{\alpha_1}$)
to $\alpha_k$ and $\alpha_{k+1}$. 
Thus, since we know that $\{\alpha_k, \alpha_{k+1}\}$ belongs to 
$\conflicts{\Kmc}$, it follows that we must also have $\{\alpha'_k, \alpha'_{k+1}\} \in \conflicts{\Kmc'}$. 

To complete the argument, it suffices to remark that, if the maximum predicate arity is $m$
and there are $p$ predicates in $\signature{\Kmc}$, 
then there can be no more than $p \cdot (2m)^m$
pairwise non-isomorphic (w.r.t.\ the original $m$ constants) facts, so it suffices 
to consider sequences of facts of length at most $p \cdot (2m)^m+2$ (any longer 
sequences that produce cycles can be shortened, given the preceding argument). 

If we consider a class of theories $\Tmc$ of bounded arity (in particular, if $\Tmc$ is formulated 
in a description logic), then this length bound is polynomial. We can then decide $\Tmc$-\emph{cyclicity}
in $\np$ by guessing such a bounded sequence of facts $\alpha_1, \ldots, \alpha_n$, together
with the rules and instantiations used, and the underlying database and metadatabase 
(whose sizes are also polynomial as we only need to consider facts that result from the rule instantiations),
and checking that the guess indeed yields a $\succ_\Sigma$ cycle. 

Finally, note that if $\Tmc$ were to contain constants, we would simply need to 
consider isomorphisms that are the identity w.r.t.\ $\individuals{\alpha_1} \cup \individuals{\Tmc}$,
rather than $\individuals{\alpha_1}$, in order to ensure that the modified pairs of facts 
remain in conflict w.r.t.\ $\Tmc$. 
\end{proof}

\cordllite*
\begin{proof}
It is well known and easy to show (using e.g.\ existing rewriting algorithms like those in \cite{calvaneseetal:dllite})
that if $\Tmc$ is a DL-Lite ontology, then there exists 
a set of denial constraints $\Tmc'$ capturing the same consistency conditions, i.e.\ 
such that $(\Dmc,\Tmc) \models \bot$ iff $(\Dmc,\Tmc')\models \bot$ for all $\Dmc$. 
Moreover, if $\Tmc$ is in a so-called core dialect of DL-Lite, then 
it is enough to consider \emph{binary denial constraints}, having at most two relational atoms. 
 \end{proof}

\strongacyclicprop*
\begin{proof}
This is a fairly immediate consequence of the definitions. 
Indeed, if $\Sigma$ is not $\Tmc$-acyclic, 
then there exists a KB of the form $\Kmc=(\Dmc,\Tmc)$ and 
accompanying meta-database $\meta=(\id,\data)$
such that the induced relation $\fullsucc$ contains a cycle. 
It then suffices to recall that $\alpha\fullsucc \beta$ implies that $\pref{\id(\alpha)}{\id(\beta)}\in\Sigma(\Kmc,\meta)$. 
This means that the same cycle occurs in $\{(\alpha,\beta) \mid \pref{\id(\alpha)}{\id(\beta)}\in\Sigma(\Kmc,\meta)\}$, so $\Sigma$ is not strongly acyclic. 
\end{proof}

\relationprefrepairslevels*
\begin{proof}
	\begin{itemize}
	\item We show that $\alpha\succup\beta$ iff $R(\id(\alpha),\id(\beta))\in\mn{Poss}(\Kmc^{cy})$. 	
	First note that	
	\begin{align*}
	\mi{inc}(\Kmc^{cy}) = \min\{i \mid \Dmc^{cy}_1\cup\dots\cup\Dmc^{cy}_i \text{ is inconsistent \wrt}\Tmc^{cy} \}
	= \min\{ i \mid \succ_{\Sigma_1} \cup \dots\cup\succ_{\Sigma_i} \text{ is cyclic}\}
	\end{align*}	
	Let $i = \level(\alpha,\beta)$. We have	
	\begin{align*}
	\alpha \succup \beta &\text{ iff } i<\min\{ i \mid \succ_{\Sigma_1} \cup \dots\cup\succ_{\Sigma_i} \text{ is cyclic}\}
	\text{ iff } i<\mi{inc}(\Kmc^{cy})
	\text{ iff } R(\id(\alpha),\id(\beta)) \in \mn{Poss}(\Kmc^{cy})
	\end{align*}
	
	\item We now show that $\alpha\succdown\beta$ iff $R(\id(\alpha),\id(\beta))\in\mn{NonDef}(\Kmc^{cy})$. 	
	It follows from the characterization of the non-defeated repair in \cite{DBLP:conf/ijcai/BenferhatBT15} that $\mn{NonDef}(\Kmc^{cy}) = \bigcup_{i = 1}^n{free(\Dmc^{cy}_1\cup\dots\cup\Dmc^{cy}_i)}$, where for every dataset $\Bmc$, $free(\Bmc)=\bigcap_{\Rmc\in\reps{\Bmc,\Tmc^{cy}}}\Rmc$ is the set of facts from $\Bmc$ that are not involved in any conflict. 	
	Let $i = \level(\alpha,\beta)$ and let $\succdown_{(k)}$ be defined for $0\leq k \leq n$ by: 	 \begin{itemize}
    \item   $\succdown_{(0)} = \cup_{i=1}^n \succ_{\Sigma_i}$
    \item $\succdown_{(k+1)} = \succdown_{(k)} \setminus \{ \gamma \in \succdown \mid \level(\gamma) = n-k, \gamma \text{ belongs to a cycle in } \succdown_{(k)}\}$ for $0\leq k < n$
  \end{itemize}
  We have $\succdown_{(n)} = \succdown$ and for $1 \leq j \leq n, \cup_{k=1}^j \succ_{\Sigma_k} \subseteq \succdown_{(n-j)}$.

		\begin{itemize}
		
		\item[$\Rightarrow$] Suppose $\alpha\succdown\beta$. We then know that $(\alpha, \beta) \not\in \{ (\alpha', \beta') \mid \level(\alpha', \beta')=i \text{ and } (\alpha',\beta') \text{ belongs to a cycle in } \succdown_{(n-i)} \}$. As $\succ_{\Sigma_1} \cup \dots\cup\succ_{\Sigma_i} \subseteq \succdown_{(n-i)}$ we know that $(\alpha, \beta)$ does not belong to a cycle \wrt $\succ_{\Sigma_1} \cup \dots\cup\succ_{\Sigma_i}$. Therefore we have that $R(\id(\alpha),\id(\beta)) $ does not belong to a conflict of $\Dmc^{cy}_1\cup\dots\cup\Dmc^{cy}_i$ \wrt $\Tmc^{cy}$, thus $R(\id(\alpha),\id(\beta))\in free(\Dmc^{cy}_1\cup\dots\cup\Dmc^{cy}_i)$. We conclude that  $R(\id(\alpha),\id(\beta))\in\mn{NonDef}(\Kmc^{cy})$.

		\item[$\Leftarrow$] Suppose now that $R(\id(\alpha),\id(\beta))\in\mn{NonDef}(\Kmc^{cy})$. We then know that $R(\id(\alpha),\id(\beta))\in free(\Dmc^{cy}_1\cup\dots\cup\Dmc^{cy}_i)$. Hence we deduce that $(\alpha, \beta)$ does not belong to a cycle \wrt $\succ_{\Sigma_1} \cup \dots\cup\succ_{\Sigma_i}$. Let us now show by contradiction that $(\alpha, \beta)$  is not removed during the computation of $\succdown_{(n-i+1)}$. Suppose that $(\alpha, \beta)$ belongs to a cycle $C$ \wrt $\succdown_{(n-i+1)}$, which means that there exists $(\alpha',\beta')$ in $C$ with $\level(\alpha',\beta')= \max\{\level(\gamma,\delta) \mid (\gamma, \delta) \in C \}>i$ (because $(\alpha,\beta)$ does not belong to a cycle \wrt $\succ_{\Sigma_1} \cup \dots\cup\succ_{\Sigma_i}$). Let $j = \level(\alpha', \beta')$. Then at step $n-j+1$ of the algorithm, pair $(\alpha', \beta')$ would have been removed from $\succdown$ and thus would not appear in $\succdown_{(n-i+1)}$, which is a contradiction. Hence $(\alpha, \beta)$ is not removed at the $n-i+1$ step of the algorithm and $\alpha \succdown \beta$.
		\end{itemize}
	\item 	
	The preference-based set-based argumentation framework (PSETAF) associated to the prioritized KB $\Kmc^{cy}$ (with the priority relation $\succ$ that corresponds to the partition of $\Dmc^{cy}$ in priority levels: $R(\id(\alpha), \id(\beta)) \succ R(\id(\alpha'), \id(\beta'))$ iff $\level(\alpha,\beta) < \level(\alpha',\beta')$), as defined in \cite{DBLP:conf/kr/BienvenuB20}, is $F_{\Kmc^{cy},\succ}=(\Dmc^{cy},\rightsquigarrow,\succ)$ with 	
	$$\rightsquigarrow = \{ (C \setminus \{R(\id(\alpha), \id(\beta))\}, R(\id(\alpha), \id(\beta))) \mid C \in\conflicts{\Kmc^{cy}}, R(\id(\alpha), \id(\beta)) \in C \}.$$
		
	The corresponding SETAF \cite{DBLP:conf/kr/BienvenuB20} is $F=(\Dmc^{cy},\rightsquigarrow_\succ)$ with 
	\begin{align*}
	\rightsquigarrow_\succ = \{ (C \setminus \{R(\id(\alpha), \id(\beta))\}, R(\id(\alpha), \id(\beta))) \mid\ &C \in\conflicts{\Kmc^{cy}}, R(\id(\alpha), \id(\beta)) \in C, \\&\forall R(\id(\alpha'), \id(\beta')) \in C, R(\id(\alpha), \id(\beta)) \not\succ R(\id(\alpha'), \id(\beta'))\}.
	\end{align*}
		
	Given $A \subseteq \Dmc^{cy}$ the set of arguments attacked by $A$ in $F$ is $A^{+} = \{ R(\id(\alpha),\id(\beta))\mid C \rightsquigarrow_\succ R(\id(\alpha),\id(\beta)) \text{ for some } C\subseteq A \}$, 
		and A \emph{defends} $R(\id(\alpha),\id(\beta))$ in $F$ if $A^{+} \cap E \neq \emptyset$ whenever $E \rightsquigarrow_\succ R(\id(\alpha),\id(\beta))$. 	
	The characteristic function $\Gamma_F:2^{\Dmc^{cy}}\mapsto 2^{\Dmc^{cy}}$ of $F$ is then defined by $\Gamma_F(A) = \{ R(\id(\alpha),\id(\beta)) \mid A \text{ defends } R(\id(\alpha),\id(\beta))\text{ in } F\}$. 	
	The grounded extension of $F_{\Kmc^{cy},\succ}$ is the grounded extension $Grd(F)$ of $F$, which is equal to the least fixpoint of $\Gamma_F$ and can be characterized by $Grd(F) = \bigcup_{i = 1}^\infty \Gamma_F^i(\emptyset)$ \cite{Baroni&al/grounded}. 	
	We thus need to show that $\alpha\succground\beta$ iff $R(\id(\alpha),\id(\beta))\in Grd(F) $. 
	
If we let $\succground_i$ be $\succground$ computed by the algorithm at step $i$, and identify $\succground_i$ with the set of facts $\{R(\id(\alpha),\id(\beta))\mid \alpha\succground_i\beta\}$, the algorithm that constructs $\succground$ starts with $\succground_1 = \emptyset$ then at every step $i>1$ computes $\Gamma_F(\succground_{i-1})$ and adds it to $\succground_{i-1}$ to obtain $\succground_i$, until a fixpoint is reached. 
	Indeed, $R(\id(\alpha),\id(\beta))\in\Gamma_F(\succground_{i-1})$ iff for every $C\in \conflicts{\Kmc^{cy}}$ such that $R(\id(\alpha),\id(\beta))\in C$:
	\begin{itemize}
	\item either $C\setminus\{R(\id(\alpha),\id(\beta))\}$ does not attack $R(\id(\alpha),\id(\beta))$ in $F$, \ie there is $R(\id(\alpha'),\id(\beta'))\in C$ such that $\level(\alpha,\beta)<\level(\alpha',\beta')$;
	\item or there is $R(\id(\alpha'),\id(\beta'))\in C$ and $E\subseteq \succground_{i-1}$ such that $E\rightsquigarrow_\succ R(\id(\alpha'),\id(\beta'))$, which implies that there is $C'\in \conflicts{\Kmc^{cy}}$ such that $R(\id(\alpha'),\id(\beta'))\in C'$ and $C'\setminus\{R(\id(\alpha'),\id(\beta'))\}\subseteq \succground_{i-1}$;
	\end{itemize}
	and it is easy to check that these two conditions correspond to the conditions the algorithm uses to add $(\alpha,\beta)$ to $\succground$ since conflicts of $\conflicts{\Kmc^{cy}}$ correspond to the cycles of $\succ_\Sigma$.\qedhere
\end{itemize}
\end{proof}

\relationsucc*
\begin{proof}
We get $\succup\subseteq \succdown \subseteq \succground$ from Theorem~\ref{relationprefrepairslevels} and the known relationship $\mn{Poss}(\Kmc^{cy})\subseteq \mn{NonDef}(\Kmc^{cy}) \subseteq \mn{Grd}(\Kmc^{cy})$ \cite{DBLP:conf/ijcai/BenferhatBT15,DBLP:conf/kr/BienvenuB20}. It thus remains to show that $\succdown \subseteq \succrefup$. 
Let $(\alpha, \beta) \in \succdown$, and let $i=\level(\alpha,\beta)$. As $(\alpha, \beta) \in \succdown$, we know, by construction of $\succdown$, that 
\begin{align*}
(\alpha,\beta)\not \in \{(\alpha',\beta')\mid   \level(\alpha',\beta')=i, (\alpha',\beta')
 \text{ is in a cycle \wrt} \bigcup_{j \leq i} \succ_{\Sigma_j}\}.
\end{align*}
Therefore $(\alpha,\beta)$ will be added in $\succrefup$ during the $i^{th}$ step of its construction. Indeed, we have $\succrefup\!\cup\!\succ_{\Sigma_i} \subseteq \bigcup_{j \leq i} \succ_{\Sigma_j}$ so if $(\alpha,\beta)$ does not belong to any cycle \wrt $\bigcup_{j \leq i} \succ_{\Sigma_j}$, it does not belong either to any cycle \wrt $\succrefup\!\cup\!\succ_{\Sigma_i}$. 
\end{proof}


\begin{table*}
\scalebox{0.85}{
\begin{tabular*}{1.15\textwidth}{l l }
\toprule
$\Pi_{\mi{minC}}$&
\mt{-included(X, Y)\text{ :- }conf\_init(X), conf\_init(Y), inConf\_init(X, A), not~inConf\_init(Y, A).}
\\&
\mt{minimal(Y)\text{ :- }conf\_init(X), conf\_init(Y), not~-included(X, Y), -included(Y, X).}
\\&
\mt{conf(X)\text{ :- }conf\_init(X), not~minimal(X).} \
\hfill
\mt{inConf(X, Y)\text{ :- }inConf\_init(X, Y), conf(X).}
\\
\bottomrule
\end{tabular*}
}
\caption{Logic program to compute actual conflicts ($\mt{conf}$, $\mt{inConf}$) from $\mt{conf\_init}$ and $\mt{inConf\_init}$.}\label{tab:program-minconf}
\medskip 

\renewcommand{\arraystretch}{0.7}
\scalebox{0.85}{
\begin{tabular*}{1.15\textwidth}{l l }
\toprule
$\Pi_{\succground}$ &
\multicolumn{1}{r}{\emph{build successor relationship on levels: $\mt{succ}(k,k+1)$}}
\\&
\mt{-succ(I, J)\text{ :- }level(I), level(J), level(Z), I < Z, Z < J}.
\\&
\mt{succ(I, J)\text{ :- }level(I), level(J), I < J, not~-succ(I, J)}.
\\&
\multicolumn{1}{r}{\emph{computation of $\Gamma^k(\emptyset)^+$: edges that belong to some cycle whose other edges are in $\Gamma^k(\emptyset)$ and of less or equal level}}
\\&
\mt{trans\_cl\_bis(X, Y, I, K)\text{ :- }level(K), pref\_init(X, Y, I), gamma(X, Y, I, K)}.
\\&
\mt{trans\_cl\_bis(X, Y, I, K)\text{ :- }level(I), level(K), trans\_cl\_bis(X, Y, J, K), J<=I}.
\\&
\mt{trans\_cl\_bis(X, Y, I, K)\text{ :- }level(K), pref\_init(X, Z, J), trans\_cl\_bis(Z, Y, I, K), J<=I, gamma(X, Z, I, K)}.
\\&
\mt{gamma\_plus(X, Y, K)\text{ :- }level(I), level(K), pref\_init(X, Y, I), trans\_cl\_bis(Y, X, J, K), J <= I}.
\\&
\multicolumn{1}{r}{\emph{computation of cycles with edges of level less or equal to $i$ and that do not belong to $\Gamma^k(\emptyset)^+$}}
\\&
\mt{trans\_cl(X, Y, I, K)\text{ :- }level(I), level(K), pref\_init(X, Y, I), not~gamma\_plus(X, Y, K), K>1}.
\\&
\mt{trans\_cl(X, Y, I, K)\text{ :- }level(I), level(J), level(K), trans\_cl(X, Y, J, K), J<=I, not~gamma\_plus(X, Y, K), K>1}.
\\&
\mt{trans\_cl(X, Y, I, K)\text{ :- }level(I), level(J), level(K), pref\_init(X, Z, J), trans\_cl(Z, Y, I, K), J<=I, not~gamma\_plus(X, Y, K),}
\\&\phantom{\mt{trans\_cl(X, Y, I, K)\text{ :- }}} \mt{ K>1}.
\\&
\mt{cycle(X, Y, I, K)\text{ :- }level(J), level(K), pref\_init(X, Y, I), trans\_cl(Y, X, J, K), J <= I}.
\\&
\multicolumn{1}{r}{\emph{computation of $\Gamma^1(\emptyset)$: edges s.t. every cycle they belong to contains some edge of higher level}}
\\&
\mt{trans\_cl(X, Y, I, 1)\text{ :- }level(I), pref\_init(X, Y, I)}.
\\&
\mt{trans\_cl(X, Y, I, 1)\text{ :- }level(I), level(J), trans\_cl(X, Y, J, 1), J<=I}.
\\&
\mt{trans\_cl(X, Y, I, 1)\text{ :- }level(I), level(J), pref\_init(X, Z, J), trans\_cl(Z, Y, I, 1), J<=I}.
\\&
\mt{gamma(X, Y, I, 1)\text{ :- } pref\_init(X, Y, I), not~cycle(X, Y, I, 1)}.
\\&
\multicolumn{1}{r}{\emph{computation of $\Gamma^{k+1}(\emptyset)$: edges s.t. every cycle they belong to contains some edge of higher level or in $\Gamma^{k}(\emptyset)^+$}}
\\&
\mt{gamma(X, Y, I, L)\text{ :- }level(K), pref\_init(X, Y, I), not~cycle(X, Y, I, K), succ(K, L)}.
\\&
\multicolumn{1}{r}{\emph{continue if $k$ is not the max level and $\Gamma^k(\emptyset)\neq \Gamma^{k+1}(\emptyset)$}}
\\&
\mt{unstopped(K)\text{ :- }level(I), level(L), gamma(X, Y, I, L), not~gamma(X, Y, I, K), succ(K, L)}.
\\&
\multicolumn{1}{r}{\emph{when $\Gamma^k(\emptyset)= \Gamma^{k+1}(\emptyset)$ or $k$ is the max level, output $\Gamma^k(\emptyset)$}}
\\&
\mt{pref(X, Y)\text{ :- }level(K), gamma(X, Y, I, K), not~unstopped(K)}.
\\
\bottomrule
\end{tabular*}
}
\caption{Logic program to compute $\succground$ from facts on predicates \mt{conf}, \mt{inConf}, $\mt{pref\_init}$ and \mt{level}, based on the characterization of $\succground$ given by Theorem~\ref{relationprefrepairslevels}. Intuitively, \mt{gamma(X,Y,I,L)} represents that the fact $R(\mt{X},\mt{Y})$ of $\Dmc^{cy}$, which represents the edge between facts of $\Dmc$ with identifiers \mt{X} and \mt{Y} in the preference statements, is such that $R(\mt{X},\mt{Y})\in\Dmc^{cy}_i$ (``of level $i$'') and $R(\mt{X},\mt{Y})\in \Gamma^\ell(\emptyset)$ (\cf Section~\ref{sec:app-computing-gr}). Similarly, \mt{gamma\_plus(X, Y, K)} represents that $R(\mt{X},\mt{Y})$ is attacked by some subset of $\Gamma^k(\emptyset)$.
}\label{app-tab:program-prio-from-pref-grounded}
\end{table*}

\section{Appendix: Omitted ASP Programs}

\subsection{Conflict Minimization}
Table~\ref{tab:program-minconf} provides the ASP program $\Pi_{\mi{minC}}$ that can be used to compute the real conflicts (\mt{conf}, \mt{inConf}) from facts on \mt{conf\_init} and \mt{inConf\_init} computed by $\Pi_\Dmc\cup\Pi_C$. 

\subsection{Computing $\succground$}\label{sec:app-computing-gr} Table~\ref{app-tab:program-prio-from-pref-grounded} provides the ASP program that computes $\succground$ from the conflicts given by facts on predicates \mt{conf}, \mt{inConf}, and the $\succ_{\Sigma_i}$'s given by $\mt{pref\_init}$, based on  the characterization of $\succground$ given by Theorem~\ref{relationprefrepairslevels}: 
$\alpha\succground\beta$ iff $R(\id(\alpha),\id(\beta))\in\mn{Grd}(\Kmc^{cy})$. 

Let $F=(\Dmc^{cy},\rightsquigarrow_\succ)$ be the SETAF defined as in the proof of Theorem \ref{relationprefrepairslevels}. 
Recall that the attack relation $\rightsquigarrow_\succ$ is defined as follows: given a set $C \subseteq \Dmc^{cy}$ and $R(\id(\alpha), \id(\beta)) \in \Dmc^{cy}$, we have an attack $C \rightsquigarrow_\succ R(\id(\alpha), \id(\beta))$ if and only if 
\begin{itemize}
\item $C \cup \{R(\id(\alpha), \id(\beta))\} \in \conflicts{\Kmc^{cy}}$, and 
\item for all $R(\id(\alpha'), \id(\beta')) \in C$,  $R(\id(\alpha), \id(\beta)) \not\succ R(\id(\alpha'), \id(\beta'))$.
\end{itemize}
Due to the definitions of $\conflicts{\Kmc^{cy}}$ and $\succ$, this holds just in the case that $C \cup \{R(\id(\alpha), \id(\beta))\}$ is a (minimal) $R$-cycle in $\Dmc^{cy}$ and $R(\id(\alpha), \id(\beta))$ is in the least important priority level among facts taking part in the cycle. 

Now given a set $A \subseteq \Dmc^{cy}$ and $R(\id(\alpha), \id(\beta) )\in \Dmc^{cy}$, we have that $A$ defends $R(\id(\alpha), \id(\beta))$ if and only if for every set $C$ that attacks $R(\id(\alpha), \id(\beta))$ (i.e. $R(\id(\alpha), \id(\beta))$ belongs to a cycle in $C \cup \{R(\id(\alpha), \id(\beta))\}$ and occurs in the least important priority level among elements of the cycle), there exists $R(\id(\gamma), \id(\delta)) \in C$ that is attacked by $A$ (i.e. $ R(\id(\gamma), \id(\delta))$ belongs to the least important priority level of a cycle contained in $A \cup R(\id(\gamma), \id(\delta))$).

For readability, we will use $\Gamma$ to refer to the characteristic function $\Gamma_F$ of the SETAF $F$. Recall that for $A \subseteq \Dmc^{cy}$, we have
$\Gamma(A) = \{ R(\id(\alpha),\id(\beta)) \in  \Dmc^{cy} | A \text{ defends } R(\id(\alpha),\id(\beta))\}$ and that $\mn{Grd}(\Kmc^{cy})= \bigcup_{i = 1}^\infty \Gamma^i(\emptyset)$ (\cf proof of Theorem~\ref{relationprefrepairslevels}). 	

To prove the correctness of the implementation of $\succground$ with $\Pi_{\succground}$ given in Table~\ref{app-tab:program-prio-from-pref-grounded}, 
we shall prove that it is sufficient to compute the iterations of $\Gamma$ up to the $n^{th}$ step, \ie that $\mn{Grd}(\Kmc^{cy})= \bigcup_{i = 1}^n \Gamma^i(\emptyset)$. To this end, 
for integers $i \geq 1$ and $k\geq0$, we introduce the following sets:
\begin{align*}
\succground_{\leq i} &= \{ R(\id(\alpha),\id(\beta)) \in \mn{Grd}(\Kmc^{cy}) \mid \level(\alpha,\beta) \leq i \} \\
\Gamma^k(\emptyset)_{\leq i} &= \{R(\id(\alpha),\id(\beta)) \in \Gamma^k(\emptyset) \mid \level(\alpha, \beta) \leq i\}
\end{align*}
Now we can prove the main property that allows us to deduce it is enough to compute $\Gamma^k(\emptyset)$ up to $\Gamma^n(\emptyset)$.

\begin{theorem}
Let $\Sigma$ be a set of preference rules partitioned into $\Sigma_1, \ldots, \Sigma_n$. Then for every $1 \leq j \leq n$: 
$\succground_{\leq j} = \Gamma^{(j)}(\emptyset)_{\leq j}$.
\end{theorem}
\begin{proof}
We first observe that the inclusion $ \Gamma^{(j)}(\emptyset)_{\leq j} \subseteq \succground_{\leq j}$ is immediate, since $\mn{Grd}(\Kmc^{cy})= \bigcup_{i = 1}^\infty \Gamma^i(\emptyset)$. 
It thus remains to show the other inclusion, namely that $\succground_{\leq j} \subseteq  \Gamma^{(j)}(\emptyset)_{\leq j} $. 
We prove this statement by induction on $j$.\medskip

\noindent\textbf{Base case ($j = 1$)}. 	
	Suppose that there is $R(\id(\alpha),\id(\beta)) \in \succground_{\leq 1} \setminus \Gamma(\emptyset)_{\leq 1}$,
and let $k = \min\{l \mid R(\id(\alpha),\id(\beta)) \in \Gamma^{(l)}(\emptyset)\}$. 
Note that since $R(\id(\alpha),\id(\beta)) \not \in \Gamma(\emptyset)_{\leq 1}$, we must have $k>1$. We consider two cases, depending on the value of $k$:
	\begin{itemize}
	
	\item Case $k = 2$:  
	 We know that $R(\id(\alpha),\id(\beta)) \not \in \Gamma(\emptyset)$, 
	 so there must exist some  $C \rightsquigarrow_\succ R(\id(\alpha),\id(\beta))$. 
	 Due to the definition of $\rightsquigarrow_\succ$, $R(\id(\alpha),\id(\beta))$ must belong to the least important priority level among all elements of 
	 $C \cup \{R(\id(\alpha),\id(\beta))\}$. As  $R(\id(\alpha),\id(\beta))$ belongs to priority level 1 (the most important level), it follows that all elements of $C$ 	
	 are also in priority level 1. 
	Next note that since $R(\id(\alpha),\id(\beta)) \in \Gamma^{(2)}(\emptyset)$, then 
	 $\Gamma(\emptyset)$ defends $R(\id(\alpha),\id(\beta))$. 
	 This means that there exists $R(\id(\gamma),\id(\delta)) \in C \cap \Gamma(\emptyset)^+$, \ie there is a 
	 subset $S\subseteq \Gamma(\emptyset)$ such that  $S \rightsquigarrow_\succ R(\id(\gamma),\id(\delta))$. 
	 Again referring to the definition of $\rightsquigarrow_\succ$, this means that 
	  $S \cup \{R(\id(\gamma),\id(\delta))\} \in \conflicts{\Kmc^{cy}}$ and that $R(\id(\gamma),\id(\delta))$ belongs to the least important priority level within this conflict. 
	  But since $R(\id(\gamma),\id(\delta))$ belongs to priority level 1, so too must all elements of $S$. 
	  It follows that for any $\pi \in S$, we have the attack $S \setminus \{\pi\} \cup \{R(\id(\gamma),\id(\delta))\}  \rightsquigarrow_\succ \pi$.
	  This contradicts $S \subseteq \Gamma(\emptyset)$.

	\item Case $k > 2$: We will employ the following lemma that shows that if a fact of level $1$ is added during the $k^{th}$ iteration of $\Gamma$, then there must exist another fact of level $1$ that was added during the $(k-1)^{th}$ iteration. 
	We can therefore recursively apply Lemma \ref{lem:init-grounded} to obtain a contradiction using the argument from case $k=2$.
	\end{itemize}

\begin{lemma}\label{lem:init-grounded}
Let  $(\alpha, \beta)$ be a pair of facts such that  $\level(\alpha, \beta) = 1$.
If $R(\id(\alpha),\id(\beta)) \in \Gamma^{(l)}(\emptyset) \setminus \Gamma^{(l-1)}(\emptyset)$ with $l>2$, then there exists $R(\id(\eta),\id(\theta)) \in \Gamma^{(l-1)}(\emptyset) \setminus \Gamma^{(l-2)}(\emptyset)$ with $\level(\eta, \theta) = 1$.
\end{lemma}
\begin{proof}
Suppose $R(\id(\alpha),\id(\beta)) \in \Gamma^{(l)}(\emptyset) \setminus \Gamma^{(l-1)}(\emptyset)$ with $l > 2$. 
We know that $\Gamma^{(l-2)}(\emptyset)$ does not defend $R(\id(\alpha),\id(\beta))$, so there exists some attack 
 $C_a \rightsquigarrow_\succ R(\id(\alpha),\id(\beta))$ against which $R(\id(\alpha),\id(\beta))$ is not defended by $\Gamma^{(l-2)}(\emptyset)$. 
 Since $R(\id(\alpha),\id(\beta)) \in \Gamma^{(l)}(\emptyset)$, however, $\Gamma^{(l-1)}(\emptyset)$ defends $R(\id(\alpha),\id(\beta))$. In particular, this means that there exist $C_d \subseteq \Gamma^{(l-1)}(\emptyset)$ and $R(\id(\gamma),\id(\delta)) \in C_a$ such that 
 $C_d \rightsquigarrow_\succ R(\id(\gamma),\id(\delta))$.
 From  $C_a \rightsquigarrow_\succ R(\id(\alpha),\id(\beta))$ and $\level(\alpha, \beta) = 1$, 
 we can infer that $\level(\gamma, \delta) = 1$. 
 This in turn can be combined with  $C_d \rightsquigarrow_\succ R(\id(\gamma),\id(\delta))$
 to infer that all elements of $C_d$ have priority level 1. 
Since $\Gamma^{(l-2)}(\emptyset)$ does not defend $R(\id(\alpha),\id(\beta))$ against $C_a \rightsquigarrow_\succ R(\id(\alpha),\id(\beta))$, $C_d\not\subseteq \Gamma^{(l-2)}(\emptyset)$ so there exists $R(\id(\eta),\id(\theta)) \in C_d$ such that $R(\id(\eta),\id(\theta)) \in \Gamma^{(l-1)}(\emptyset) \setminus \Gamma^{(l-2)}(\emptyset)$ and $\level(\eta, \theta) = 1$, as desired. 
\end{proof}

\noindent\textbf{Induction step}. Let $1 \leq i < n$ and suppose that for all $1 \leq j \leq i$, $\succground_{\leq j} \subseteq \Gamma^{(j)}(\emptyset)_{\leq j}$. We want to show $\succground_{\leq i+1} \subseteq \Gamma^{(i+1)}(\emptyset)_{\leq i+1}$. 
Suppose that there is $R(\id(\alpha),\id(\beta)) \in \succground_{\leq i+1} \setminus \Gamma^{(i+1)}(\emptyset)_{\leq i+1}$,
and let $k = \min\{l \mid R(\id(\alpha),\id(\beta)) \in \Gamma^{(l)}(\emptyset)\}$. 
Note that since $R(\id(\alpha),\id(\beta)) \not \in \Gamma^{(i+1)}(\emptyset)_{\leq i+1}$, we must have $k>i+1$. We consider two cases, depending on the value of $k$:
	\begin{itemize}
	
	\item Case $k = i+2$: We know that $R(\id(\alpha),\id(\beta)) \not \in \Gamma^{(i+1)}(\emptyset)$ so there exists some $C' \rightsquigarrow_\succ R(\id(\alpha),\id(\beta))$ against which $\Gamma^{(i)}(\emptyset)$ does not defend $R(\id(\alpha),\id(\beta))$. 
	Since $R(\id(\alpha),\id(\beta)) \in \Gamma^{(i+2)}(\emptyset)$,  $ \Gamma^{(i+1)}(\emptyset)$ defends $R(\id(\alpha),\id(\beta))$. In particular, there must exist $C \subseteq \Gamma^{(i+1)}(\emptyset)$ and $R(\id(\gamma),\id(\delta)) \in C' $ such that $C \rightsquigarrow_\succ R(\id(\gamma),\id(\delta))$. 	
Since $C' \rightsquigarrow_\succ R(\id(\alpha),\id(\beta))$, it must be the case that $\level(\gamma, \delta) \leq \level(\alpha,\beta) \leq i+1$, so since $C \rightsquigarrow_\succ R(\id(\gamma),\id(\delta))$, it must be the case that $l = \max\{\level(\eta,\theta) \mid R(\id(\eta),\id(\theta)) \in C \}\leq \level(\gamma, \delta)  \leq i+1$. We have to consider two cases: 
	
	\begin{itemize}
			\item If $l < i+1$, we have by induction hypothesis that $\succground_{\leq l} \subseteq \Gamma^{(l)}(\emptyset)_{\leq l}$ thus $C \subseteq \Gamma^{(l)}(\emptyset)_{\leq l}\subseteq \Gamma^{(i)}(\emptyset)$. This contradicts the fact that $\Gamma^{(i)}(\emptyset)$ does not defend $R(\id(\alpha),\id(\beta))$ against $C' \rightsquigarrow_\succ R(\id(\alpha),\id(\beta))$. 
			\item If $l = i+1$ then there exists $R(\id(\eta),\id(\theta)) \in C$ with $\level(\eta,\theta) = i+1$, from which we can infer that $\level(\gamma,\delta) = i+1$, since $C \rightsquigarrow_\succ R(\id(\gamma),\id(\delta))$. However, it follows that $(C \setminus \{R(\id(\eta),\id(\theta))\} \cup \{R(\id(\gamma),\id(\delta)) \}\rightsquigarrow_\succ R(\id(\eta),\id(\theta))$. Thus as $C \subseteq \Gamma^{(i+1)}(\emptyset)$ it must be the case that  $\Gamma^{(i)}(\emptyset)$ defends $R(\id(\eta),\id(\theta))$. So there exists $C''\subseteq\Gamma^{(i)}(\emptyset)$ and $R(\id(\gamma'),\id(\delta')) \in C \setminus \{R(\id(\eta),\id(\theta))\} \cup \{R(\id(\gamma),\id(\delta))\}$ such that $C'' \rightsquigarrow_\succ R(\id(\gamma'),\id(\delta'))$. 
			
			If $(\gamma', \delta') = (\gamma, \delta)$, it follows that $\Gamma^{(i)}(\emptyset)$ defends $R(\id(\alpha),\id(\beta))$ against $C' \rightsquigarrow_\succ R(\id(\alpha),\id(\beta))$, which contradicts the definition of $C' \rightsquigarrow_\succ R(\id(\alpha),\id(\beta))$. 
			Otherwise, $R(\id(\gamma'),\id(\delta')) \in C \setminus \{R(\id(\eta),\id(\theta))\}$ which contradicts $C \subseteq \Gamma^{(i+1)}(\emptyset)$ since $C''\subseteq\Gamma^{(i)}(\emptyset)$ and $C'' \rightsquigarrow_\succ R(\id(\gamma'),\id(\delta'))$.
	\end{itemize}

	\item Case $k > i+2$: We will employ the following lemma that shows that if a fact of level $i+1$ is added during the $k^{th}$ iteration of $\Gamma$ with $k>i+2$, then there must exist another fact of level $i+1$ that was added during the $(k-1)^{th}$ iteration. 
	We can therefore recursively apply Lemma \ref{lem:rec-grounded} to obtain a contradiction using the argument from case $k=i+2$.
		
\end{itemize}
\begin{lemma}\label{lem:rec-grounded}
Let  $(\alpha, \beta)$ be a pair of facts such that  $\level(\alpha, \beta) = i+1$.
If $R(\id(\alpha),\id(\beta)) \in \Gamma^{(l)}(\emptyset) \setminus \Gamma^{(l-1)}(\emptyset)$ with $l>i+2$, then there exists $R(\id(\eta),\id(\theta)) \in \Gamma^{(l-1)}(\emptyset) \setminus \Gamma^{(l-2)}(\emptyset)$ with $\level(\eta, \theta) = i+1$. 
\end{lemma}
\begin{proof}
Suppose $R(\id(\alpha),\id(\beta)) \in \Gamma^{(l)}(\emptyset) \setminus \Gamma^{(l-1)}(\emptyset)$ with $l > i+2$. 
We know that $\Gamma^{(l-2)}(\emptyset)$ does not defend $R(\id(\alpha),\id(\beta))$, so there exists some attack $C_a \rightsquigarrow_\succ R(\id(\alpha),\id(\beta))$ against which $R(\id(\alpha),\id(\beta))$ is not defended by $\Gamma^{(l-2)}(\emptyset)$. 
Since $R(\id(\alpha),\id(\beta)) \in \Gamma^{(l)}(\emptyset)$, however, $\Gamma^{(l-1)}(\emptyset)$ defends $R(\id(\alpha),\id(\beta))$. In particular, this means that there exist $C_d \subseteq \Gamma^{(l-1)}(\emptyset)$ and $R(\id(\gamma),\id(\delta)) \in C_a$ such that 
 $C_d \rightsquigarrow_\succ R(\id(\gamma),\id(\delta))$. Since $\level(\alpha, \beta) = i+1$, it must be the case that all facts in $C_a$ are in priority levels of indexes lower or equal to $i+1$. This holds in particular for $R(\id(\gamma),\id(\delta))$ so this is also the case for the facts in $C_d$. 
By induction hypothesis, we know that $\succground_{\leq i} \subseteq \Gamma^{(i)}(\emptyset)_{\leq i}$. Hence, if we denote by $C_{d,\leq i}$ the set of facts in $C_d$ that belong to levels of indexes lower or equal to $i$, $C_{d,\leq i} \subseteq \Gamma^{(i)}(\emptyset)\subseteq \Gamma^{(l-2)}(\emptyset)$, since $ i< l-2$. Since $C_d\not\subseteq\Gamma^{(l-2)}(\emptyset)$ (as $\Gamma^{(l-2)}(\emptyset)$ does not defend $R(\id(\alpha),\id(\beta))$ against $C_a \rightsquigarrow_\succ R(\id(\alpha),\id(\beta))$), there exists $R(\id(\eta),\id(\theta)) \in C_d\setminus\Gamma^{(l-2)}(\emptyset)$ and since $C_{d,\leq i} \subseteq \Gamma^{(l-2)}(\emptyset)$, it must be the case that $\level(\eta, \theta) >i$. Hence  $R(\id(\eta),\id(\theta)) \in \Gamma^{(l-1)}(\emptyset) \setminus \Gamma^{(l-2)}(\emptyset)$ and $\level(\eta, \theta)=i+1$ as required.
\end{proof}
\end{proof}

\begin{corollary}
Let $\Sigma$ be a set of preference rules partitioned into $\Sigma_1, \ldots, \Sigma_n$. Then
	$\succground = \bigcup_{i = 1}^\infty \Gamma^i(\emptyset) = \bigcup_{i = 1}^n \Gamma^i(\emptyset)$. 
\end{corollary}

\begin{table*}
\scalebox{0.85}{
\begin{tabular*}{1.15\textwidth}{l l }
\toprule
$\Pi_{\mi{att}}$&
\mt{-att(X, A)\text{ :- }inConf(X, A), inConf(X, B), not~A = B, pref(A, B).}
\\&
\mt{att(X, A)\text{ :- }inConf(X, A), not~-att(X, A).}
\\
\midrule
$\Pi_{\mi{loc}}$&
\mt{cause\_fact(A)\text{ :- }inCause(C, A), cause(C).}
\\&
\mt{reachable(A)\text{ :- }cause\_fact(A).}
\\&
\mt{reachable(A)\text{ :- }reachable(B), att(X, B), inConf(X, A).}
\\
\midrule
$\Pi_{\mi{loc\_att}}$&
\mt{cause\_fact(A)\text{ :- }inCause(C, A), cause(C).}
\\&
\mt{reachable(A)\text{ :- }cause\_fact(A).}
\\&
\mt{-att(X, A)\text{ :- }reachable(A), inConf(X, A), inConf(X, B), not~A = B, pref(A, B).}
\\&
\mt{att(X, A)\text{ :- }reachable(A), inConf(X, A), not~-att(X, A).}
\\&
\mt{reachable(A)\text{ :- }att(X, B), inConf(X, A).}
\\
\midrule
$\Pi_{\mi{cons}}$&
\mt{conf\_rel(X)\text{ :- }inConf(X, A), reachable(A).}
\\&
\mt{1 \{rem(A):inConf(X, A)\}\text{ :- }conf\_rel(X).}
\\&
\mt{in(A)\text{ :- }reachable(A), not~rem(A).}
\\
\midrule
$\Pi_{\mi{brave}}$&
\mt{-sat(C)\text{ :- }inCause(C, A), not~in(A).}
\\&
\mt{sat\text{ :- }cause(C), not~-sat(C).}
\\&
\mt{\text{ :- }not~sat.}
\\
\midrule
$\Pi_{\mi{AR}}$&
\mt{invalid\_conf(X, A)\text{ :- }reachable(A), att(X, A), inConf(X, B), not~in(B), not~A = B.}
\\&
\mt{neg(C)\text{ :- }cause(C), inCause(C, A), att(X, A), not~invalid\_conf(X, A).}
\\&
\mt{\text{ :- }cause(C), not~neg(C).}
\\
\midrule
$\Pi_{\mi{Pareto}}$&
\mt{valid(A)\text{ :- }reachable(A), in(A).}
\\&
\mt{invalid\_att(X, A)\text{ :- }reachable(A), att(X, A), inConf(X, B), not~in(B), not~A = B.}
\\&
\mt{valid(A)\text{ :- }reachable(A), conf(X), not~in(A), att(X, A), not~invalid\_att(X, A).}
\\&
\mt{\text{ :- }reachable(A), not~valid(A).}
\\
\midrule
$\Pi_{\mi{Completion}}$&
\mt{valid(A)\text{ :- }reachable(A), in(A).}
\\&
\mt{invalid\_att(X, A)\text{ :- }reachable(A), not~in(A), inConf(X, A), not~A = B, inConf(X, B), not~in(B).}
\\&
\mt{invalid\_att(X, A)\text{ :- }reachable(A), not~in(A), inConf(X, A), inConf(X, B), not~A = B,  pref\_comp(A, B).}
\\&
\mt{valid(A)\text{ :- }reachable(A), not~in(A), inConf(X, A), not~invalid\_att(X, A).}
\\&
\mt{\text{ :- }reachable(A), not~valid(A).}
\\&
\mt{pref\_comp(A, B)\text{ :- }reachable(A), reachable(B), pref(A, B).}
\\&
\mt{1 \{pref\_comp(A, B); pref\_comp(B, A)\} 1\text{ :- }reachable(A), reachable(B), inConf(X, A), inConf(X, B),} 
\\&\hfill{\mt{not~pref(A, B), not~pref(B, A), not~A = B.}}
\\&
\mt{trans\_cl\_comp(A, B)\text{ :- }pref\_comp(A, B).}
\\&
\mt{trans\_cl\_comp(A, B)\text{ :- }trans\_cl\_comp(A, Y), pref\_comp(Y, B).}
\\&
\mt{\text{ :- }trans\_cl\_comp(A, A).}
\\
\bottomrule
\end{tabular*}
}
\caption{Logic programs to filter query answers from facts on predicates \mt{conf}, \mt{inConf}, \mt{pref}, \mt{cause} and \mt{inCause}. Choice rules of the form
$
1\{\gamma_1;\dots;\gamma_k\}1 \text{ :- } \alpha_1,\dots, \alpha_n, \mt{not}\beta_1,\dots,\mt{not}\beta_m.$ or 
$1\{\gamma_1;\dots;\gamma_k\}\text{ :- } \alpha_1,\dots, \alpha_n, \mt{not}\beta_1,\dots,\mt{not}\beta_m.$ indicate the selection of exactly (\resp at least) one of the $\gamma_i$ (the set $\{\gamma_1;\dots;\gamma_k\}$ can also be expressed intensionally, \eg \mt{\{ p(X): q(X,Y) \}}). 
}\label{tab:programs-semantics-appendix}
\end{table*}

\subsection{Optimal Repair-Based Semantics}
Recall that we say that a conflict $\Cmc$ \emph{attacks} a fact $\alpha$, written $\Cmc\rightsquigarrow\alpha$, if $\alpha\in\Cmc$ and $\alpha\not\succ\beta$ for every $\beta\in\Cmc$. We use $\Pi_{\mi{att}}$ from Table~\ref{tab:programs-semantics-appendix} to pre-compute the attack relation $\rightsquigarrow$ ($\mt{att}$).

For $X\in\{S,P,C\}$ and $\sem\in\{\brave,\AR,\IAR\}$, we define $\Pi_{X\text{-}\sem}$ from the building blocks presented in Table~\ref{tab:programs-semantics-appendix}. 
For $\sem\in\{\brave, \AR\}$, 
$\Pi_{X\text{-}\sem}$ is the union of the following programs: 
\begin{itemize}
\item $\Pi_{\mi{loc}}$, which localizes the attack relation to relevant facts to filter the potential answer at hand ($\mt{reachable}$), \ie those that are reachable from the facts of the causes in the attack hypergraph (we also consider a variant where we do not pre-compute the whole attack relation but instead use the program $\Pi_{\mi{loc\_att}}$ which computes it only with facts that are potentially relevant); 
\item $\Pi_{\mi{cons}}$, which selects a consistent set of facts (\mt{in}) among the relevant facts ($\mt{reachable}$) by enforcing that at least one fact per relevant conflict (\mt{conf\_rel}) is removed (\mt{rem}) using a choice rule;
\item $\Pi_{\mi{brave}}$ if $\sem=\brave$, which ensures that the program is satisfiable only if some cause is included in the selected facts (\mt{in});
\item $\Pi_{\mi{AR}}$ if $\sem=\AR$, which ensures that the program is satisfiable only if every cause is contradicted (\mt{neg}) by the selected facts (\mt{in}), meaning that these facts include $\Cmc\setminus\{\alpha\}$ for some $\Cmc\rightsquigarrow\alpha$ with $\alpha$ a fact of the cause (\ie $\Cmc$ is not an \mt{invalid\_conf} such that some $\beta\in\Cmc$ different from $\alpha$ is not selected (\mt{in}));
\item $\Pi_{\mi{Pareto}}$ if $X=P$, which ensures that the program is satisfiable only if every relevant fact ($\mt{reachable}$) $\alpha$ is either selected (\mt{in}), or attacked by a conflict $\Cmc$ such that all facts in $\Cmc\setminus\{\alpha\}$ are selected (\ie $\Cmc$ is not an \mt{invalid\_att}), which ensures that the set of selected facts can be extended to a Pareto-optimal repair;
\item $\Pi_{\mi{Completion}}$ if $X=C$, which ensures that the program is satisfiable only if the set of selected facts can be extended to a completion-optimal repair: there exists a completion (\mt{pref\_comp}, built with a choice rule that enforces that all conflicting facts are compared and a constraint that ensures acyclicity) $\succ'$ of $\succ$ (\mt{pref}) such that every relevant fact ($\mt{reachable}$) $\alpha$ is either selected (\mt{in}), or attacked (\wrt $\succ'$) by a conflict $\Cmc$ such that all facts in $\Cmc\setminus\{\alpha\}$ are selected (\ie $\Cmc$ is not an \mt{invalid\_att}).
\end{itemize}
For $\sem=\IAR$, $\Pi_{X\text{-}\sem}$ intuitively checks whether each cause can be contradicted by a consistent set of facts. It is similar to $\Pi_{AR\text{-}\sem}$, except that predicates in $\Pi_{\mi{loc}}$, $\Pi_{\mi{cons}}$, $\Pi_{\mi{AR}}$ and $\Pi_{\mi{Pareto}}$ or $\Pi_{\mi{Completion}}$ are extended with an extra argument that keeps the identifier of the cause considered.


\section{Appendix: Experiments}

\mypar{Additional non-binary conflicts and time to compute the conflicts} Figure~\ref{app-fig:additional-denial} shows the denial constraint we add to the ontology from the \mn{\mn{CQAPri}} benchmark to generate conflicts of size 10. We obtain 40 such conflicts on all datasets. Table \ref{app-tab:conflicts-nb-time} shows the number of conflicts and the time needed to minimize them using a \mn{Python} program, starting from the $\mt{conf\_init}$ facts obtained via $\Pi_\Dmc\cup\Pi_C$. In the case where the conflicts are guaranteed to be of size at most 2 (when we use the \mn{CQAPri} ontology without the additional denial constraint of Figure~\ref{app-fig:additional-denial}), we use the optimized version of the program that checks whether there is any self-inconsistent fact (\ie conflict of size 1), and if so removes all candidate conflicts that contain such fact. Note that the time needed to compute the $\mt{conf\_init}$ facts with $\Pi_\Dmc\cup\Pi_C$ is at most around 1 minute (\mn{u20c50} case).

\begin{figure*}
$\mn{memberOf}(x, \mi{dpt14u0})\wedge \bigwedge_{i=1}^3(\mn{takesCourse}(x,y_i) \wedge \mn{subj11Course}(y_i)\wedge \mn{graduateCourse}(y_i))\wedge \bigwedge_{1\leq i< j\leq 3} y_i\neq y_j\rightarrow\bot$
\caption{Additional denial constraint: members of department 14 of university 0 cannot take three distinct graduate courses on subject 11.}\label{app-fig:additional-denial}
\end{figure*}

\begin{table}
\renewcommand{\arraystretch}{0.8}
\scalebox{0.85}{
\begin{tabular*}{1.18\textwidth}{l r @{\extracolsep{\fill}}r r r}
\toprule
 & \multicolumn{2}{c}{Number of conflicts} &\multicolumn{2}{c}{Minimization time (ms)}  
 \\
 & \multicolumn{1}{c}{binary} &\multicolumn{1}{c}{non-binary}  & \multicolumn{1}{c}{binary} &\multicolumn{1}{c}{non-binary} 
  \\
 & \multicolumn{1}{c}{(\mn{CQAPri} onto.)} &\multicolumn{1}{c}{(\mn{CQAPri} onto. + Fig.~\ref{app-fig:additional-denial})}  & \multicolumn{1}{c}{(\mn{CQAPri} onto.)} &\multicolumn{1}{c}{(\mn{CQAPri} onto. + Fig.~\ref{app-fig:additional-denial})} 
\\
\midrule
\mn{u1c1} & 2,314 & 2,354 &
258 & 231

\\
\midrule
\mn{u1c5}  & 8,476 & 8,516 &
113 & 2,295
\\
\midrule
\mn{u1c10}  & 14,261 & 14,301 & 167 & 6,398

\\
\midrule
\mn{u1c20}  & 28,232 & 28,272& 324 & 24,355

\\
\midrule
\mn{u1c30} & 45,484 & 45,524 & 463 & 63,129

\\
\midrule
\mn{u1c50} & 81,304 & 81,344& 1,572 & 205,577

\\
\midrule
\mn{u5c1} & 11,984 & 12,024& 723 & 4,797

\\
\midrule
\mn{u5c5} & 53,398 & 53,438 &2,580 & 92,145

\\
\midrule
\mn{u5c10} & 109,453 & 109,493& 1,143 & 344,871

\\
\midrule
\mn{u5c20} & 231,771 & 231,811& 2,334 & 1,690,967

\\
\midrule
\mn{u20c1} & 73,212 & 73,252& 1,026 & 162,254

\\
\midrule
\mn{u20c50} & 3,130,377 & 3,130,417& 41,608 & 164,057,654
\\
\bottomrule
\end{tabular*}
}
\caption{Number of conflicts and time (in ms) to minimize them, depending on whether conflicts are guaranteed to be at most binary or not.
}\label{app-tab:conflicts-nb-time}
\end{table}

\mypar{Computation of $\succ^x$ from the conflicts} Table~\ref{app-tab:app-results-prio-numbers} shows the number of $\mt{pref\_init}$ facts ($\succ_\Sigma$) and $\mt{pref}$ facts computed for $\succup$, $\succdown$ and $\succground$, in the case where we use the \mn{CQAPri} ontology extended with the additional denial constraint of Figure~\ref{app-fig:additional-denial} to define the conflicts. 
Note that all datasets have exactly 40 conflicts of size 10, which yields 1,718 pairs of facts, and other conflicts are binary (so that \eg \mn{u1c1} has $1,718+2,314=4,032$ pairs of conflicting facts), and that several preference statements (\ie $\mt{pref\_init}$ facts) can be made on each such pair (in both directions and on different priority levels), which explains the high numbers of $\mt{pref\_init}$ facts compared to the number of conflicts. In contrast, the size of the priority relation ($\mt{pref}$ facts) is bounded by the number of pairs of conflicting facts, so that, \eg $\succground$ compares $\frac{3,703}{4,032}\approx 92\%$ of the pairs of conflicting facts of \mn{u1c1} in scenario (a). 
Interestingly, $\succdown$ and $\succground$ coincide in all scenarios but (a) and never differ by more than 5\% of \mt{pref} facts on instances for which we computed them, while $\succup$ is often reduced to the empty relation.  

Table~\ref{app-tab:time-prio-non-binary} shows the time taken to compute $\succup$, $\succdown$ and $\succground$ from the precomputed conflicts given as a program $\Pi_{\mi{conf}}$ and $\Pi_\Dmc\cup\Pi_\data\cup\Pi_P\cup\Pi_{\succ^x}$.  Figure~\ref{app-fig:results-prio-times} shows how these runtimes decompose between grounding and solving the ASP program for the \mn{u1cY} cases. It also helps visualizing how the total runtime evolves when the proportion of facts in conflicts increases: the \mn{u1cY} datasets have from 3\% to 44\% facts in conflicts (note that some facts belong to many conflicts, \eg \mn{u1c50} has 81K conflicts for 78K facts and 44\% facts in conflicts). 
We can see that $\succdown$ is significantly faster to compute than $\succup$ and $\succground$ (except in scenario (d) which yields a very small and acyclic $\succ_\Sigma$ and for which the three methods are comparable). Note that in contrast with $\succdown$ and $\succground$, the runtime of $\succup$ is not monotone \wrt the number of conflicts in \mn{u1cY} (Figure~\ref{app-fig:results-prio-times}): the solving time drops when the size of $\succup$ drops because of cycles in early levels. 
Figure~\ref{app-fig:results-prio-times-wrt-size} shows how the runtime evolves when the dataset size increases for the three datasets with the lowest proportion of facts in conflicts (\mn{uXc1}).  Finally, Figure~\ref{app-fig:results-prio-times-wrt-pref-init} shows how it evolves \wrt the number of \mt{pref\_init} facts ($\succ_\Sigma$). Note that in the latter case we plot on the same graph the times for all available datasets in scenario (a), which explains why the runtimes are not monotone: \eg \mn{u1c50} and \mn{u5c10} have 145,193 and 194,306 \mt{pref\_init} facts in $\succ_\Sigma$, respectively, but it takes 336 seconds and 170 seconds to compute $\succdown$ in these cases, respectively. Intuitively, this is because even if \mn{u1c50} yields less preference statements than \mn{u5c10}, these statements are more connected since the conflicts of \mn{u1c50} are more connected than those of \mn{u5c10}, as \mn{u1c50} has 44\% facts involved in some conflicts while this proportion is of 18\% for \mn{u5c10}.

\begin{table*}
\scalebox{0.85}{
\begin{tabular*}{1.18\textwidth}{l @{\extracolsep{\fill}} r  r r r r r r r r r r r}
\toprule
& &\multicolumn{4}{c}{$\Sigma^{a}_1\cup\Sigma^{a}_2\cup\Sigma^{a}_3$} &\multicolumn{3}{c}{$\Sigma^{b}_1\cup\Sigma^{b}_2$}& \multicolumn{3}{c}{$\Sigma^{c}_1$} & \multicolumn{1}{c}{$\Sigma^{d}_1$} 
\\
&\multicolumn{1}{c}{\#conf.} &\multicolumn{1}{c}{$\succ_\Sigma$} &  \multicolumn{1}{c}{$\succup$}  & \multicolumn{1}{c}{$\succdown$} & \multicolumn{1}{c}{$\succground$} &\multicolumn{1}{c}{$\succ_\Sigma$} &  \multicolumn{1}{c}{$\succup$}  & \multicolumn{1}{c}{$\succ^{d,g}$} &  \multicolumn{1}{c}{$\succ_\Sigma$} &\multicolumn{1}{c}{$\succup$} &\multicolumn{1}{c}{$\succ^{d,g}$}  & \multicolumn{1}{c}{$\succ$} 
\\
\midrule
\mn{u1c1}& 2,354 & 7,068 & 3,041 & 3,644 & 3,703 & 7,068& 4,027& 4,029 &
5,633 & 0 & 1,510 & 0
\\
\midrule
\mn{u1c5}& 8,516 & 17,804 & 7,624 & 8,944 & 9,324 & 17,803& 10,180 & 10,185&
 14,517 & 0 & 3,356 & 1
\\
\midrule
\mn{u1c10}& 14,301 & 27,927 & 2 & 14,402 & 14,808 & 27,926& 0 & 15,969&
 22,634 & 0 & 6,082 & 2
\\
\midrule
\mn{u1c20}& 28,272 & 52,361 & 4 & 27,185 & 27,948 & 52,359&0 & 29,937&
 42,032 & 0 & 12,300 & 4
\\
\midrule
\mn{u1c30}& 45,524 & 82,531 & 6 & 41,300 & 42,601 & 82,529&0 & 47,177&
 65,142 & $\_$ & 16,361 & 6
\\
\midrule
\mn{u1c50}& 81,344 & 145,193 & $\_$ & 69,454 &  $\_$ & 145,189& 0& 82,964& 
 113,857 & $\_$ & 19,966 & 8
\\
\midrule
\mn{u5c1}& 12,024 & 23,932 & 10,241 & 13,275 & 13,339 &23,932 & 13,691& 13,697& 
18,570 & 0 & 7,821 & 0
\\
\midrule
\mn{u5c5}& 53,438 & 96,307 &  $\_$ & 52,084 & 52,820 & 96,306& $\_$ & 55,082& 
 76,045 & 0 & 28,017 & 1
\\
\midrule
\mn{u5c10}& 109,493 & 194,306 & 6 & 103,094 &  $\_$ & 194,301&0 & 111,107& 
 154,271 &  $\_$ & 50,673 & 6
\\
\midrule
\mn{u5c20}& 231,811 &  $\_$ & $\_$ & $\_$ &$\_$ & $\_$ & $\_$ & $\_$ & 
 319,549 & $\_$ & 87,006 & 14
\\
\midrule
\mn{u20c1}& 73,252 &  131,103 &  2 &  73,157 & 73,583 & 131,103& 0& 74,909& 
 103,260 & $\_$ & 43,062 & 2
\\
\midrule
\mn{u20c50}& 3,130,417 & $\_$ & $\_$ & $\_$ & $\_$ &$\_$ &$\_$ & $\_$& 
$\_$ &$\_$& $\_$ & 159
\\
\bottomrule
\end{tabular*}
}
\caption{Number of conflicts, \mt{pref\_init} facts ($\succ_\Sigma$), and \mt{pref} facts computed for $\succup$, $\succdown$ and $\succground$, for scenarios (a), (b), (c) (for which $\succdown$ and $\succground$ coincide) and (d) (which directly yields an acyclic relation), in the case of non-binary conflicts. Empty cells indicate that we fail to compute the priority relation (time-out or \mn{clingo} overflow). 
We fail to compute priority relations on omitted datasets in all scenarios but (d).
}\label{app-tab:app-results-prio-numbers}

\bigskip

\scalebox{0.85}{
\begin{tabular*}{1.18\textwidth}{l r @{\extracolsep{\fill}}r r r r r r r r r r r r r r}
\toprule
 &  \multicolumn{3}{c}{$\Sigma^{a}_1\cup\Sigma^{a}_2\cup\Sigma^{a}_3$} &&\multicolumn{3}{c}{$\Sigma^{b}_1\cup\Sigma^{b}_2$} && \multicolumn{3}{c}{$\Sigma^{c}_1$} && \multicolumn{3}{c}{$\Sigma^{d}_1$}
\\
& \multicolumn{1}{c}{$\succup$}  & \multicolumn{1}{c}{$\succdown$} & \multicolumn{1}{c}{$\succground$} && \multicolumn{1}{c}{$\succup$}  & \multicolumn{1}{c}{$\succdown$} & \multicolumn{1}{c}{$\succground$} &
& \multicolumn{1}{c}{$\succup$}  & \multicolumn{1}{c}{$\succdown$} & \multicolumn{1}{c}{$\succground$} && \multicolumn{1}{c}{$\succup$}  & \multicolumn{1}{c}{$\succdown$} & \multicolumn{1}{c}{$\succground$}
\\
\midrule
\mn{u1c1} 
& 34,097 & 5,707 & 51,506&
& 52,855 & 6,146 & 33,386&
& 27,809 & 5,106 & 14,860 &  
& 748 & 769 & 766   
\\
\midrule
\mn{u1c5}     
& 174,082 & 13,516 & 112,960&
& 216,410 & 13,099 & 66,679&
& 166,529 & 11,542 & 33,256&
& 892 & 889 & 894
\\
\midrule
\mn{u1c10}     
& 39,236 & 17,252 & 148,875 &
& 37,494 & 18,777 & 101,603 &
& 248,375 & 15,847 &  46,939&
& 1,008 & 998 &  1,005
\\
\midrule
\mn{u1c20}  
& 80,101 & 29,523 & 254,688 &
& 72,401 & 36,927 & 187,702 &
& 718,958 & 27,171 & 85,366&
& 1,284 & 1,263 & 1,280
\\
\midrule
\mn{u1c30}    
& 201,223 & 62,349 & 577,159&
& 149,403 & 69,987 & 350,050&
& t.o. & 61,563 & 171,479&
& 1,578 & 1,581 & 1,596
\\
\midrule
\mn{u1c50}   
& t.o. & 336,223 & oom &
& 837,682 & 353,148 & t.o. &
& t.o. & 241,462 & 1,118,390&
& 2,263 & 2,257 & 2,280
\\
\midrule
\mn{u5c1}   
& 122,486 & 19,085 & 131,353 &
& 235,344 & 38,057 & 139,778&
& 80,311 & 17,548 & 37,034&
& 4,990 & 4,946 & 4,952
\\
\midrule
\mn{u5c5}  
& t.o.  & 67,188  & 528,963 &
& t.o.  & 51,533  & 294,507&
& 1,522,838 & 43,151 & 117,191&
& 5,862 & 5,778 & 5,811
\\
\midrule
\mn{u5c10}  
& 813,199  & 170,116  & t.o. &
& 336,306  & 133,543  & 710,729 &
& t.o. & 93,641 & 286,540&
& 6,936 & 6,841 & 6,928
\\
\midrule
\mn{u5c20}  
& t.o. & t.o. & t.o. &
& t.o. & t.o.  & oom &
& t.o. & 397,087 & t.o.&
& 9,490 & 9,381 & 9,379
\\
\midrule
\mn{u20c1}  
& 309,008 & 94,175 & 593,144	 &
& 156,349 & 110,224 & 410,489 &
& t.o. & 83,886 & 194,295&
& 22,600 & 22,533 & 22,763
\\
\midrule
\mn{u20c50}  
& t.o. & oom & oom &
& t.o. & oom & oom&
& t.o. & oom & oom &
& 88,367 & 87,716 & 89,180
\\
\bottomrule
\end{tabular*}
}
\caption{Total time (in ms) to compute $\succ^x$ from the pre-computed conflicts given as a program $\Pi_{\mi{conf}}$ and $\Pi_\Dmc\cup\Pi_\data\cup\Pi_P\cup\Pi_{\succ^x}$, corresponding to the time \mn{clingo} takes to ground and solve the program, in the case of non-binary conflicts. oom indicates that \mn{clingo} overflows and t.o. that we reach the time-out (30 min).}\label{app-tab:time-prio-non-binary}
\end{table*}

\begin{figure}
\definecolor{solving}{HTML}{696969}
\definecolor{grounding}{HTML}{BEBEBE}
\phantom{1,}
\begin{tikzpicture}
\begin{axis}[ybar stacked, bar width=1mm, width=4cm, height=4.4cm, ymin=0, ymax=600]\addplot[grounding,fill=grounding] coordinates{
(2.354,  6.409) 
(8.516,  15.489) 
(14.301,  20.415) 
(28.272,  33.526) 
(45.524,  68.858) 
(81.344,  0) 
};
\addplot[solving,fill=solving] coordinates{
(2.354,  27.688) 
(8.516,  158.593) 
(14.301,  18.820) 
(28.272,  46.575) 
(45.524,  132.365) 
(81.344,  0) 
};
\path (200,500) node {(a) $\succup$};
\end{axis} 
\end{tikzpicture}
\begin{tikzpicture}
\begin{axis}[ybar stacked, bar width=1mm, width=4cm, height=4.4cm, ymin=0, ymax=600,ytick=\empty]\addplot[grounding,fill=grounding] coordinates{
(2.354,  4.779) 
(8.516,  11.028) 
(14.301,  14.266) 
(28.272,  24.460) 
(45.524,  52.1612) 
(81.344,  301.057) 
};
\path (200,500) node {(a) $\succdown$};
\addplot[solving,fill=solving] coordinates{
(2.354, 0.927) 
(8.516,  2.488) 
(14.301,  2.985) 
(28.272,  5.062) 
(45.524,  10.187) 
(81.344,  35.166) 
};
\end{axis} 
\end{tikzpicture}
\begin{tikzpicture}
\begin{axis}[ybar stacked, bar width=1mm, width=4cm, height=4.4cm, ymin=0, ymax=600, ytick=\empty]\addplot[grounding,fill=grounding] coordinates{
(2.354,  39.079) 
(8.516,  86.8382) 
(14.301,  114.640) 
(28.272,  196.7766) 
(45.524,  446.008) 
(81.344,  0) 
};
\addplot[solving,fill=solving] coordinates{
(2.354,  12.427) 
(8.516,  26.122) 
(14.301,  34.235) 
(28.272,  57.912) 
(45.524,  131.151) 
(81.344,  0) 
};
\path (200,500) node {(a) $\succground$};
\end{axis} 
\end{tikzpicture}
\hfill
\begin{tikzpicture}
\begin{axis}[ybar stacked, bar width=1mm, width=4cm, height=4.4cm, ymin=0, ymax=900]\addplot[grounding,fill=grounding] coordinates{
(2.354,  7.306) 
(8.516,  16.027) 
(14.301, 22.418) 
(28.272,  40.594) 
(45.524,  79.800) 
(81.344,  464.632) 
};
\addplot[solving,fill=solving] coordinates{
(2.354,  45.549) 
(8.516,  200.382) 
(14.301,  15.076) 
(28.272,  31.807) 
(45.524, 69.602) 
(81.344,  373.049) 
};
\path (200,750) node {(b) $\succup$};
\end{axis} 
\end{tikzpicture}
\begin{tikzpicture}
\begin{axis}[ybar stacked, bar width=1mm, width=4cm, height=4.4cm, ymin=0, ymax=900,ytick=\empty]\addplot[grounding,fill=grounding] coordinates{
(2.354,  5.185) 
(8.516,  10.783) 
(14.301,  15.216) 
(28.272,  29.948) 
(45.524,  57.357) 
(81.344,  294.292) 
};
\path (200,750) node {(b) $\succdown$};
\addplot[solving,fill=solving] coordinates{
(2.354, 0.960) 
(8.516,  2.315) 
(14.301,  3.560) 
(28.272,  6.979) 
(45.524,  12.629) 
(81.344,  58.855) 
};
\end{axis} 
\end{tikzpicture}
\begin{tikzpicture}
\begin{axis}[ybar stacked, bar width=1mm, width=4cm, height=4.4cm, ymin=0, ymax=900, ytick=\empty]\addplot[grounding,fill=grounding] coordinates{
(2.354,  26.115) 
(8.516,  52.106) 
(14.301,  79.815) 
(28.272,  146.006) 
(45.524,  272.404) 
(81.344,  0) 
};
\addplot[solving,fill=solving] coordinates{
(2.354,  7.270) 
(8.516,  14.572) 
(14.301,  21.787) 
(28.272,  41.696) 
(45.524,  77.645) 
(81.344,  0) 
};
\path (200,750) node {(b) $\succground$};
\end{axis} 
\end{tikzpicture}

\begin{tikzpicture}
\begin{axis}[ybar stacked, bar width=1mm, width=4cm, height=4.4cm, ymin=0, ymax=1200]\addplot[grounding,fill=grounding] coordinates{
(2.354,  5.495) 
(8.516,  12.250) 
(14.301,  16.730) 
(28.272,  28.498) 
(45.524,  0) 
(81.344,  0) 
};
\addplot[solving,fill=solving] coordinates{
(2.354,  22.314) 
(8.516,  154.278) 
(14.301,  231.645) 
(28.272,  690.460) 
(45.524,  0) 
(81.344,  0) 
};
\path (200,100) node {(c) $\succup$};
\end{axis} 
\end{tikzpicture}
\begin{tikzpicture}
\begin{axis}[ybar stacked, bar width=1mm, width=4cm, height=4.4cm, ymin=0, ymax=1200,ytick=\empty]\addplot[grounding,fill=grounding] coordinates{
(2.354,  4.366) 
(8.516, 9.640 ) 
(14.301,  13.137) 
(28.272,  22.743) 
(45.524, 51.674 ) 
(81.344,  224.111) 
};
\path (200,100) node {(c) $\succdown$};
\addplot[solving,fill=solving] coordinates{
(2.354, 0.740) 
(8.516,  1.901) 
(14.301, 2.709 ) 
(28.272,  4.427) 
(45.524,  9.889) 
(81.344, 17.350 ) 
};
\end{axis} 
\end{tikzpicture}
\begin{tikzpicture}
\begin{axis}[ybar stacked, bar width=1mm, width=4cm, height=4.4cm, ymin=0, ymax=1200, ytick=\empty]\addplot[grounding,fill=grounding] coordinates{
(2.354, 12.875 ) 
(8.516,  29.284) 
(14.301, 41.309 ) 
(28.272,  75.577) 
(45.524, 150.315) 
(81.344,  944.237) 
};
\addplot[solving,fill=solving] coordinates{
(2.354, 1.985 ) 
(8.516,  3.972) 
(14.301, 5.630 ) 
(28.272, 9.788 ) 
(45.524, 21.164) 
(81.344,  174.153) 
};
\path (200,100) node {(c) $\succground$};
\end{axis} 
\end{tikzpicture}
\hfill
\begin{tikzpicture}
\begin{axis}[ybar stacked, bar width=1mm, width=4cm, height=4.4cm, ymin=0, ymax=5]\addplot[grounding,fill=grounding] coordinates{
(2.354,  0.620) 
(8.516,  0.734) 
(14.301, 0.838) 
(28.272, 1.079 ) 
(45.524,  1.342) 
(81.344,  1.950) 
};
\addplot[solving,fill=solving] coordinates{
(2.354, 0.128 ) 
(8.516, 0.158 ) 
(14.301, 0.170 ) 
(28.272, 0.205 ) 
(45.524, 0.235) 
(81.344,  0.312) 
};
\path (200,400) node {(d) $\succup$};
\end{axis} 
\end{tikzpicture}
\begin{tikzpicture}
\begin{axis}[ybar stacked, bar width=1mm, width=4cm, height=4.4cm, ymin=0, ymax=5,ytick=\empty]\addplot[grounding,fill=grounding] coordinates{
(2.354, 0.640 ) 
(8.516, 0.736 ) 
(14.301, 0.835 ) 
(28.272, 1.072 ) 
(45.524, 1.350 ) 
(81.344,  1.959) 
};
\path (200,400) node {(d) $\succdown$};
\addplot[solving,fill=solving] coordinates{
(2.354,0.128 ) 
(8.516,  0.152) 
(14.301,  0.163) 
(28.272, 0.190 ) 
(45.524, 0.230 ) 
(81.344,  0.297) 
};
\end{axis} 
\end{tikzpicture}
\begin{tikzpicture}
\begin{axis}[ybar stacked, bar width=1mm, width=4cm, height=4.4cm, ymin=0, ymax=5, ytick=\empty]\addplot[grounding,fill=grounding] coordinates{
(2.354,  0.639) 
(8.516,  0.736) 
(14.301, 0.839 ) 
(28.272,  1.078) 
(45.524,  1.354) 
(81.344,  1.964) 
};
\addplot[solving,fill=solving] coordinates{
(2.354,  0.126) 
(8.516,  0.157) 
(14.301,  0.166) 
(28.272,  0.201) 
(45.524, 0.241 ) 
(81.344, 0.315 )
};
\path (200,400) node {(d) $\succground$};
\end{axis} 
\end{tikzpicture}

\caption{Time (in sec.) to compute $\succ^x$ from the pre-computed conflicts given as a program $\Pi_{\mi{conf}}$ and $\Pi_\Dmc\cup\Pi_\data\cup\Pi_P\cup\Pi_{\succ^x}$ in scenarios (a), (b), (c) and (d) \wrt the number (in thousands) of conflicts for $\mn{u1cY}$ (75-78K facts). The lower part of the bars (light grey) is the time to ground the program while the upper part is the time to solve it. Empty bars mean a time-out or oom.}\label{app-fig:results-prio-times}
\bigskip

\begin{tikzpicture}
\begin{axis}[ybar stacked, bar width=1mm, width=4cm, height=4.4cm, ymin=0, ymax=600]\addplot[grounding,fill=grounding] coordinates{
(0.075724,  6.409) 
(0.463691,  16.929) 
(1.983493,  87.123) 
};
\addplot[solving,fill=solving] coordinates{
(0.075724,  27.688) 
(0.463691,  105.556) 
(1.983493,  221.885) 
};
\path (50,500) node {(a) $\succup$};
\end{axis} 
\end{tikzpicture}
\begin{tikzpicture}
\begin{axis}[ybar stacked, bar width=1mm, width=4cm, height=4.4cm, ymin=0, ymax=600,ytick=\empty]\addplot[grounding,fill=grounding] coordinates{
(0.075724,  4.779)  
(0.463691,  15.087) 
(1.983493,  66.655) 
};
\addplot[solving,fill=solving] coordinates{
(0.075724,  0.927)  
(0.463691,  3.998) 
(1.983493,  27.520)  
};
\path (50,500) node {(a) $\succdown$};
\end{axis} 
\end{tikzpicture}
\begin{tikzpicture}
\begin{axis}[ybar stacked, bar width=1mm, width=4cm, height=4.4cm, ymin=0, ymax=600, ytick=\empty]\addplot[grounding,fill=grounding] coordinates{
(0.075724,  39.079)  
(0.463691,  100.128) 
(1.983493,  448.303) 
};
\addplot[solving,fill=solving] coordinates{
(0.075724,  12.427)  
(0.463691,  31.224) 
(1.983493,  144.840) 
};
\path (50,500) node {(a) $\succground$};
\end{axis} 
\end{tikzpicture}
\hfill
\begin{tikzpicture}
\begin{axis}[ybar stacked, bar width=1mm, width=4cm, height=4.4cm, ymin=0, ymax=450]\addplot[grounding,fill=grounding] coordinates{
(0.075724,  7.306)  
(0.463691,  21.950) 
(1.983493,  93.921) 
};
\addplot[solving,fill=solving] coordinates{
(0.075724,  45.549)  
(0.463691,  213.393) 
(1.983493,  62.427) 
};
\path (50,380) node {(b) $\succup$};
\end{axis} 
\end{tikzpicture}
\begin{tikzpicture}
\begin{axis}[ybar stacked, bar width=1mm, width=4cm, height=4.4cm, ymin=0, ymax=450,ytick=\empty]\addplot[grounding,fill=grounding] coordinates{
(0.075724,  5.185)  
(0.463691,  29.874) 
(1.983493,  77.232) 
};
\addplot[solving,fill=solving] coordinates{
(0.075724,  0.960)  
(0.463691,  8.179) 
(1.983493,  32.992) 
};
\path (50,380) node {(b) $\succdown$};
\end{axis} 
\end{tikzpicture}
\begin{tikzpicture}
\begin{axis}[ybar stacked, bar width=1mm, width=4cm, height=4.4cm, ymin=0, ymax=450, ytick=\empty]\addplot[grounding,fill=grounding] coordinates{
(0.075724,  26.115)  
(0.463691,  104.212) 
(1.983493,  316.618) 
};
\addplot[solving,fill=solving] coordinates{
(0.075724,  7.270)  
(0.463691,  35.565) 
(1.983493,  93.871) 
};
\path (50,380) node {(b) $\succground$};
\end{axis} 
\end{tikzpicture}

\begin{tikzpicture}
\begin{axis}[ybar stacked, bar width=1mm, width=4cm, height=4.4cm, ymin=0, ymax=200]\addplot[grounding,fill=grounding] coordinates{
(0.075724,  5.495)  
(0.463691,  14.892) 
(1.983493,  0) 
};
\addplot[solving,fill=solving] coordinates{
(0.075724,  22.314)  
(0.463691,  65.419) 
(1.983493,  0) 
};
\path (50,170) node {(c) $\succup$};
\end{axis} 
\end{tikzpicture}
\begin{tikzpicture}
\begin{axis}[ybar stacked, bar width=1mm, width=4cm, height=4.4cm, ymin=0, ymax=200,ytick=\empty]\addplot[grounding,fill=grounding] coordinates{
(0.075724,  4.366)  
(0.463691,  12.910) 
(1.983493,  59.968) 
};
\addplot[solving,fill=solving] coordinates{
(0.075724,  0.740)  
(0.463691,  4.638) 
(1.983493,  23.918) 
};
\path (50,170) node {(c) $\succdown$};
\end{axis} 
\end{tikzpicture}
\begin{tikzpicture}
\begin{axis}[ybar stacked, bar width=1mm, width=4cm, height=4.4cm, ymin=0, ymax=200, ytick=\empty]\addplot[grounding,fill=grounding] coordinates{
(0.075724,  12.875)  
(0.463691,  31.497) 
(1.983493,  166.214) 
};
\addplot[solving,fill=solving] coordinates{
(0.075724,  1.985)  
(0.463691,  5.536) 
(1.983493,  28.081) 
};
\path (50,170) node {(c) $\succground$};
\end{axis} 
\end{tikzpicture}
\hfill
\begin{tikzpicture}
\begin{axis}[ybar stacked, bar width=1mm, width=4cm, height=4.4cm, ymin=0, ymax=25]\addplot[grounding,fill=grounding] coordinates{
(0.075724,  0.620)  
(0.463691,  3.996) 
(1.983493,  18.122) 
};
\addplot[solving,fill=solving] coordinates{
(0.075724,  0.128)  
(0.463691,  0.993) 
(1.983493,  4.477) 
};
\path (50,220) node {(d) $\succup$};
\end{axis} 
\end{tikzpicture}
\begin{tikzpicture}
\begin{axis}[ybar stacked, bar width=1mm, width=4cm, height=4.4cm, ymin=0, ymax=25,ytick=\empty]\addplot[grounding,fill=grounding] coordinates{
(0.075724,  0.640)  
(0.463691,  3.962) 
(1.983493,  18.226) 
};
\addplot[solving,fill=solving] coordinates{
(0.075724,  0.128)  
(0.463691,  0.984) 
(1.983493,  4.306) 
};
\path (50,220) node {(d) $\succdown$};
\end{axis} 
\end{tikzpicture}
\begin{tikzpicture}
\begin{axis}[ybar stacked, bar width=1mm, width=4cm, height=4.4cm, ymin=0, ymax=25, ytick=\empty]\addplot[grounding,fill=grounding] coordinates{
(0.075724,  0.639)  
(0.463691,  3.963) 
(1.983493,  18.233) 
};
\addplot[solving,fill=solving] coordinates{
(0.075724,  0.126)  
(0.463691,  0.988) 
(1.983493,  4.530) 
};
\path (50,220) node {(d) $\succground$};
\end{axis} 
\end{tikzpicture}

\caption{Time (in sec.) to compute $\succ^x$ from the pre-computed conflicts given as a program $\Pi_{\mi{conf}}$ and $\Pi_\Dmc\cup\Pi_\data\cup\Pi_P\cup\Pi_{\succ^x}$ in scenarios (a), (b), (c) and (d) \wrt the size of the dataset (in million of facts) for $\mn{uXc1}$ (about 3\%-4\% facts in conflicts). The lower part of the bars (light grey) is the time to ground the program while the upper part is the time to solve it. The empty bar means a time-out.}\label{app-fig:results-prio-times-wrt-size}
\bigskip

\begin{tikzpicture}
\begin{axis}[ybar stacked, bar width=1mm, width=7.15cm, height=4.4cm, ymin=0, ymax=850]
\addplot[grounding,fill=grounding] coordinates{
(7.068,  6.409) 
(17.804,  15.489) 
(27.927,  20.415) 
(52.361,  33.526) 
(82.531,  68.858) 
(145.193,  0) 
(23.932, 16.929) 
(96.307, 0) 
(194.306, 161.182) 
(131.103 , 87.123) 
};
\addplot[solving,fill=solving] coordinates{
(7.068,  27.688) 
(17.804,  158.593) 
(27.927,  18.820) 
(52.361,  46.575) 
(82.531,  132.365) 
(145.193,  0) 
(23.932, 105.556) 
(96.307, 0) 
(194.306, 652.017) 
(131.103 , 221.885) 
};
\path (10,700) node {(a) $\succup$};
\end{axis} 
\end{tikzpicture}
\begin{tikzpicture}
\begin{axis}[ybar stacked, bar width=1mm, width=7.15cm, height=4.4cm, ymin=0, ymax=850,ytick=\empty]\addplot[grounding,fill=grounding] coordinates{
(7.068,  4.779) 
(17.804,  11.028) 
(27.927,  14.266) 
(52.361,  24.460) 
(82.531,  52.1612) 
(145.193,  301.057) 
(23.932, 15.087) 
(96.307, 54.508) 
(194.306, 134.668) 
(131.103 , 66.655) 
};
\addplot[solving,fill=solving] coordinates{
(7.068, 0.927) 
(17.804,  2.488) 
(27.927,  2.985) 
(52.361,  5.062) 
(82.531,  10.187) 
(145.193,  35.166) 
(23.932, 3.998) 
(96.307, 12.679) 
(194.306, 35.447) 
(131.103 , 27.520) 
};
\path (10,700)  node {(a) $\succdown$};
\end{axis} 
\end{tikzpicture}
\begin{tikzpicture}
\begin{axis}[ybar stacked, bar width=1mm, width=7.15cm, height=4.4cm, ymin=0, ymax=850, ytick=\empty]\addplot[grounding,fill=grounding] coordinates{
(7.068,  39.079) 
(17.804,  86.8382) 
(27.927,  114.640) 
(52.361,  196.7766) 
(82.531,  446.008) 
(145.193,  0) 
(23.932, 100.128) 
(96.307, 405.889) 
(194.306, 0) 
(131.103 , 448.303) 
};
\addplot[solving,fill=solving] coordinates{
(7.068,  12.427) 
(17.804,  26.122) 
(27.927,  34.235) 
(52.361,  57.912) 
(82.531,  131.151) 
(145.193,  0) 
(23.932, 31.224) 
(96.307, 123.073) 
(194.306, 0) 
(131.103, 144.840) 
};
\path (10,700)  node {(a) $\succground$};
\end{axis} 
\end{tikzpicture}

\caption{Time (in sec.) to compute $\succ^x$ from the pre-computed conflicts given as a program $\Pi_{\mi{conf}}$ and $\Pi_\Dmc\cup\Pi_\data\cup\Pi_P\cup\Pi_{\succ^x}$ in scenario~(a) \wrt the number (in thousands) of $\mt{pref\_init}$ facts ($\succ_\Sigma$). The lower part of the bars (light grey) is the time to ground the program while the upper part is the time to solve it. Empty bars mean a time-out or oom.}\label{app-fig:results-prio-times-wrt-pref-init}
\end{figure}

\mypar{Comparison with the SAT-based implementation of optimal repair-based semantics ($P$-AR, $P$-brave, $C$-AR, $C$-brave)} 
Since our implementation struggles to handle the queries with too many potential answers (with a 30 minutes time-out), we select 8 queries among the 20 of the \mn{\mn{CQAPri}} benchmark with a reasonable number of potential answers (between 3 and 538 on \mn{u1c1}, between 10 and 16,969 on \mn{u20c50}). 
The comparison is done using the {\sc orbits} \mn{Simple} algorithm, with the $\mn{neg}_1$ variant of the SAT encoding in the case of $X$-AR semantics and both $\mn{P}_1$ and $\mn{P}_2$ variants of the SAT encoding for the $P$-AR and $P$-brave semantics \cite{DBLP:conf/kr/BienvenuB22}. 
Tables~\ref{app-tab:comparison-orbits-AR-score-sctruc-u1c1}, \ref{app-tab:comparison-orbits-AR-score-sctruc-u1c50} and \ref{app-tab:comparison-orbits-score-sctruc-u20c1} show the runtimes of the two systems for \mn{u1c1} (simplest case), \mn{u1c50} (small dataset with the highest proportion of conflicting facts) and \mn{u20c1} (large dataset with a small proportion of facts in conflicts), respectively, in the case where the priority relation is score-structured, so that Pareto-optimal and completion-optimal repairs coincide. Note that in this case we don't consider the ``completion encoding'' for {\sc orbits} since it has been shown that {\sc orbits} handles poorly completion-optimal repairs compared to Pareto-optimal ones. 
Tables~\ref{app-tab:comparison-orbits-non-score-sctruc-u1c1}, \ref{app-tab:comparison-orbits-non-score-sctruc-u1c50} and \ref{app-tab:comparison-orbits-non-score-sctruc-u20c1} show the runtimes of the two systems for \mn{u1c1}, \mn{u1c50} and \mn{u20c1}, respectively, in the case where the priority relation is not score-structured, so that $P$- and $C$- semantics differ.

\begin{table*}
\scalebox{1}{
\begin{tabular*}{\textwidth}{l r @{\extracolsep{\fill}}r r r r r r r r r }
\toprule
&& \multicolumn{4}{c}{$X$-AR ($X\in \{ P, C\}$), score-structured} && \multicolumn{4}{c}{$X$-brave ($X\in \{ P, C\}$), score-structured}
\\
&&  \multicolumn{2}{c}{ASP} & \multicolumn{2}{c}{{\sc orbits}} && \multicolumn{2}{c}{ASP} & \multicolumn{2}{c}{{\sc orbits}} 
\\
 && $P$ enc. & $C$ enc. & $\mn{P}_1$ enc. & $\mn{P}_2$ enc. && $P$ enc. & $C$ enc. & $\mn{P}_1$ enc. & $\mn{P}_2$ enc.
\\
\midrule
\mn{q3} &
& 6,874 &  6,966
& 487 &  542&
& 7,021 & 6,949 
&   502    &  498
\\
\midrule
\mn{q5}&
& 826 &  840
& 482 &  530&
& 838 & 840 
&   501    &   479
\\ 
\midrule
\mn{q7}&
& 11,039 &  11,203
& 481 &  522&
& 11,235 & 11,198 
&   494    &  474
\\
\midrule
\mn{q10}&
& 242 &  246
& 473 &  522&
& 247 & 245 
&  496     &   482
\\
\midrule
\mn{q11}&
& 43,435  &  45,016
& 509 &  568&
& 44,856 & 45,305 
&  516     &  504
\\
\midrule
\mn{q14}&
& 15,713 &  16,277
&  496 &  548&
& 16,330 & 16,276 
&   511    &   504
\\
\midrule
\mn{q15}&
& 41,622 &  53,304
& 669 &  608&
& 42,765 & 58,173 
&   871    &  1,265
\\
\midrule
\mn{q20}&
& 4,033 &  4,178
& 472 &  526&
& 4,229 & 4,089 
&  484     &  481
\\
\bottomrule
\end{tabular*}
}
\caption{Time (in ms) to filter all query answers given the conflicts, priority relation, and potential answers with their causes. $X$-AR and $X$-brave semantics (for $X\in \{ P, C\}$) on \mn{u1c1} with a score-structured priority relation (5 levels).  {\sc orbits} time includes the time to load the input (oriented conflict graph and potential answer and their causes, around 450ms for all queries) as well as the pure filtering time. 
}\label{app-tab:comparison-orbits-AR-score-sctruc-u1c1}
\bigskip

\scalebox{1}{
\begin{tabular*}{\textwidth}{l r @{\extracolsep{\fill}}r r r r r r r r r }
\toprule
&& \multicolumn{4}{c}{$X$-AR ($X\in \{ P, C\}$), score-structured} && \multicolumn{4}{c}{$X$-brave ($X\in \{ P, C\}$), score-structured}
\\
&&  \multicolumn{2}{c}{ASP} & \multicolumn{2}{c}{{\sc orbits}} && \multicolumn{2}{c}{ASP} & \multicolumn{2}{c}{{\sc orbits}} 
\\
 && $P$ enc. & $C$ enc. & $\mn{P}_1$ enc. & $\mn{P}_2$ enc. && $P$ enc. & $C$ enc. & $\mn{P}_1$ enc. & $\mn{P}_2$ enc.
\\
\midrule
\mn{q3}& 
& 240,753 & t.o. 
&   1,655    &  2,836&
& 240,541 & 512,278 
&   1,784    &  5,448
\\
\midrule
\mn{q5}&
& 26,084 & 747,967 
&   865    &  902&
& 26,210 & 36,693 
&  911     &   1,318
\\ 
\midrule
\mn{q7}&
& 389,745 & t.o. 
&   1,862    &  2,057&
& 391,848 & 575,443 
&  1,841     &  3,636
\\
\midrule
\mn{q10}&
& 18,344 & 551,658 
&   808    & 946&
& 18,433 & 24,618 
&  834     &   1,130
\\
\midrule
\mn{q11}&
& 1,604,343 & t.o. 
&  1,257     & 1,762&
& 1,620,604 & 1,752,162 
&  1,251     &  1,854
\\
\midrule
\mn{q14}&
& 499,775 & t.o. 
&   1,187    & 1,607&
& 502,325 & 546,376 
&   1,527    &   2,002
\\
\midrule
\mn{q15}&
& 1,365,539 & t.o. 
&   3,001    & 4,356&
& 1,371,207 & t.o. 
&  2,760     &  9,025
\\
\midrule
\mn{q20}&
& 133,979 & t.o. 
&  1,203     &  1,799&
& 133,783 & 192,442 
&  1,241     &  2,475
\\
\bottomrule
\end{tabular*}
}
\caption{Time (in ms) to filter all query answers given the conflicts, priority relation, and potential answers with their causes. $X$-AR and $X$-brave semantics (for $X\in \{ P, C\}$) on \mn{u1c50} with a score-structured priority relation (5 levels). {\sc orbits} time includes the time to load the input (oriented conflict graph and potential answer and their causes, around 660ms for all queries) as well as the pure filtering time.
}\label{app-tab:comparison-orbits-AR-score-sctruc-u1c50}
\bigskip

\scalebox{1}{
\begin{tabular*}{\textwidth}{l r @{\extracolsep{\fill}}r r r r r r r r r }
\toprule
&& \multicolumn{4}{c}{$X$-AR ($X\in \{ P, C\}$), score-structured} && \multicolumn{4}{c}{$X$-brave ($X\in \{ P, C\}$), score-structured}
\\
&&  \multicolumn{2}{c}{ASP} & \multicolumn{2}{c}{{\sc orbits}} && \multicolumn{2}{c}{ASP} & \multicolumn{2}{c}{{\sc orbits}} 
\\
& & $P$ enc. & $C$ enc. & $\mn{P}_1$ enc. & $\mn{P}_2$ enc. && $P$ enc. & $C$ enc. & $\mn{P}_1$ enc. & $\mn{P}_2$ enc.
\\
\midrule
\mn{q3} &
& 209,397 & 214,141 
&   760    &  875&
& 211,708 & 212,570 
&  795     &  791
\\
\midrule
\mn{q5}&
& 24,686 & 25,159 
&  754     &   856&
& 24,899 & 25,107 
&  789     &   785
\\ 
\midrule
\mn{q7}&
& 337,269 & 351,997 
&  826     &  926&
& 341,276 & 344,041 
&   960    &  999
\\
\midrule
\mn{q10}&
& 143,219 & 153,679 
&  800     &   946&
& 144,470 & 149,512 
&  903     &   985
\\
\midrule
\mn{q11}&
& t.o.  & t.o. 
&  1,075     &  1,762&
& t.o. & t.o. 
&  1,117     &  1,156
\\
\midrule
\mn{q14}&
& t.o. & t.o. 
&  1,013     &   1,607&
& t.o. & t.o. 
&  1,131     &   1,299
\\
\midrule
\mn{q15}&
& t.o. & t.o. 
&  2,398     &  2,008&
& t.o. & t.o. 
&   3,059    &  4,716
\\
\midrule
\mn{q20}&
& 123,674 & 129,242 
&  771     &  880&
& 124,662 & 128,426 
& 801      &  755
\\
\bottomrule
\end{tabular*}
}
\caption{Time (in ms) to filter all query answers given the conflicts, priority relation, and potential answers with their causes. $X$-AR and $X$-brave semantics (for $X\in \{ P, C\}$) on \mn{u20c1} with a score-structured priority relation (5 levels). {\sc orbits} time includes the time to load the input (oriented conflict graph and potential answer and their causes, around 700-800ms for all queries) as well as the pure filtering time. 
}\label{app-tab:comparison-orbits-score-sctruc-u20c1}
\end{table*}

\begin{table*}
\scalebox{1}{
\begin{tabular*}{\textwidth}{l r @{\extracolsep{\fill}}r r r r r r r r r r r r r}
\toprule
&& \multicolumn{3}{c}{$P$-AR} && \multicolumn{2}{c}{$C$-AR} && \multicolumn{3}{c}{$P$-brave} && \multicolumn{2}{c}{$C$-brave}
\\
&& ASP & \multicolumn{2}{c}{\sc orbits} && ASP & {\sc orbits} && ASP & \multicolumn{2}{c}{\sc orbits} && ASP & {\sc orbits}
\\
&&& $\mn{P}_1$ enc. & $\mn{P}_2$ enc. && & && & $\mn{P}_1$ enc. & $\mn{P}_2$ enc. && &
\\
\midrule
\mn{q3} &
& 7,021 & 510& 515	&
& 7,101 & 487&
& 6,909	& 496& 498	& 
& 7,127		& 522
\\
\midrule
\mn{q5}&
& 840 & 511& 508	&
& 858 & 487&
& 828	& 481& 497	&
& 866		&518
\\ 
\midrule
\mn{q7}&
&11,235 & 515& 518	&	
& 11,364 & 487&
& 11,054 & 484& 497	&
& 11,451 & 521
\\
\midrule
\mn{q10}&
& 247 & 508& 499	&
& 251 & 502&
& 247		& 475& 487	&
& 252		& 542
\\
\midrule
\mn{q11}&
& 44,888 & 563&	562	&
& 44,994 & 547&
& 43,689	& 516	& 535 &
& 45,960		& 562
\\
\midrule
\mn{q14}&
& 16,297& 528&	529	&
&16,105 & 594&
&15,834		&	507 & 517	&
&16,278		& 679
\\
\midrule
\mn{q15}&
&43,049 & 740& 594	&
& 67,194 & oom&
&41,523		& 931	& 1,831	&
&85,708		& t.o.
\\
\midrule
\mn{q20}&
&4,127 & 513& 520	&
& 4,164 & 490&
&4,075		& 484	& 495	&
&4,205		& 505
\\
\bottomrule
\end{tabular*}
}
\caption{Time (in ms) to filter all query answers given the conflicts, priority relation, and potential answers with their causes. $X$-AR and $X$-brave semantics (for $X\in \{ P, C\}$) on \mn{u1c1} with a non score-structured priority relation. {\sc orbits} time includes the time to load the input (oriented conflict graph and potential answer and their causes, around 450ms for all queries) as well as the pure filtering time. 
}\label{app-tab:comparison-orbits-non-score-sctruc-u1c1}
\bigskip

\scalebox{1}{
\begin{tabular*}{\textwidth}{l r @{\extracolsep{\fill}}r r r r r r r r r r r r r}
\toprule
&& \multicolumn{3}{c}{$P$-AR} && \multicolumn{2}{c}{$C$-AR} && \multicolumn{3}{c}{$P$-brave} && \multicolumn{2}{c}{$C$-brave}
\\
 && ASP & \multicolumn{2}{c}{\sc orbits} && ASP & {\sc orbits} && ASP & \multicolumn{2}{c}{\sc orbits} && ASP & {\sc orbits}
\\
&&& $\mn{P}_1$ enc. & $\mn{P}_2$ enc. && & && & $\mn{P}_1$ enc. & $\mn{P}_2$ enc. && &
\\
\midrule
\mn{q3} &
& 313,790 & 3,849&	53,883&
& t.o. & t.o. &
& 295,763	&4,176	&55,595	& 
& t.o.	&t.o.
\\
\midrule
\mn{q5}&
& 32,441 & 1,242& 2,838	&
& 776,021 & t.o. &
& 32,186	&1,339	&4,773	&
& t.o.		&t.o.
\\ 
\midrule
\mn{q7}&
&496,949 & 4,478& 57,830	&	
&  t.o. & t.o. &
& 474,829 &4,601	&57,495	&
& t.o. &t.o.
\\
\midrule
\mn{q10}&
& 22,298 & 1,345& 4,032	&
& 559,350 & t.o. &
& 21,719		&1,233	&4,029	&
& 156,827		&t.o.
\\
\midrule
\mn{q11}&
& 1,689,222 & 2,104& 9,952		&
&  t.o. & t.o. &
& 1,696,173	&1,920	& 10,364	&
& t.o.		&t.o.
\\
\midrule
\mn{q14}&
& 530,806 & 2,218&	13,204	&
&  t.o. & t.o. &
&526,382	& 2,713	& 14,191	&
&t.o.		&t.o.
\\
\midrule
\mn{q15}&
&1,571,824 & 6,838& 108,599	&
& t.o.  & t.o. &
&1,544,780	&7,047 	& 114,287	&
&t.o.		&t.o.
\\
\midrule
\mn{q20}&
&176,999& 3,057& 32,445	&
&  t.o. & t.o. &
&169,949	& 3,383	& 28,658	&
&t.o.		&t.o.
\\
\bottomrule
\end{tabular*}
}
\caption{Time (in ms) to filter all query answers given the conflicts, priority relation, and potential answers with their causes. $X$-AR and $X$-brave semantics (for $X\in \{ P, C\}$) on \mn{u1c50} with a non score-structured priority relation. {\sc orbits} time includes the time to load the input (oriented conflict graph and potential answer and their causes, around 660ms for all queries) as well as the pure filtering time. 
}\label{app-tab:comparison-orbits-non-score-sctruc-u1c50}
\bigskip

\scalebox{1}{
\begin{tabular*}{\textwidth}{l r @{\extracolsep{\fill}}r r r r r r r r r r r r r}
\toprule
&& \multicolumn{3}{c}{$P$-AR} && \multicolumn{2}{c}{$C$-AR} && \multicolumn{3}{c}{$P$-brave} && \multicolumn{2}{c}{$C$-brave}
\\
 && ASP & \multicolumn{2}{c}{\sc orbits} && ASP & {\sc orbits} && ASP & \multicolumn{2}{c}{\sc orbits} && ASP & {\sc orbits}
\\
&&& $\mn{P}_1$ enc. & $\mn{P}_2$ enc. && & && & $\mn{P}_1$ enc. & $\mn{P}_2$ enc. && &
\\
\midrule
\mn{q3} &
& 211,917 & 892& 900	&
& 214,141 & 801 &
& 213,361	& 828	& 787	& 
& 217,198		& 867
\\
\midrule
\mn{q5}&
& 24,942 & 899&	897&
& 25,159 & 957 &
& 25,429	& 821	& 789	&
& 25,846		& 848
\\ 
\midrule
\mn{q7}&
&340,581 & 975& 1,040	&	
& 351,997 &  561,516 &
& 344,389 &	 1,089 & 1,130	&
& 353,489 & 609,317
\\
\midrule
\mn{q10}&
& 145,508 & 982& 1,008	&
& 153,679  & t.o. &
& 147,271		&1,122  & 1,466	&
& 155,738		&t.o.
\\
\midrule
\mn{q11}&
& t.o. & 1,163&	1,212	&
& t.o. & 4,103 &
& t.o.	& 1,077	& 1,011	&
& t.o.		& 3,898
\\
\midrule
\mn{q14}&
& t.o. & 1,087&	1,105	&
& t.o. & t.o. &
& t.o.		& 1,278	& 1,445	&
&t.o.		&t.o.
\\
\midrule
\mn{q15}&
&t.o.  & 2,653& 2,353	&
& t.o. &  t.o. &
&t.o.		& 3,353	& 9,755	&
&t.o.		&t.o.
\\
\midrule
\mn{q20}&
&124,811 & 880& 902	&
& 129,242 & 791 &
&125,767		& 818	& 778	&
&123,272		&845
\\
\bottomrule
\end{tabular*}
}
\caption{Time (in ms) to filter all query answers given the conflicts, priority relation, and potential answers with their causes. $X$-AR and $X$-brave semantics (for $X\in \{ P, C\}$) on \mn{u20c1} with a non score-structured priority relation. {\sc orbits} time includes the time to load the input (oriented conflict graph and potential answer and their causes, around 700-800ms for all queries) as well as the pure filtering time. 
}\label{app-tab:comparison-orbits-non-score-sctruc-u20c1}
\end{table*}

\end{document}